\documentclass{jair}

\usepackage[T1]{fontenc}

\usepackage{amsmath}
\usepackage{algorithm}
\usepackage{algorithmic}
\usepackage[most]{tcolorbox}

\allowdisplaybreaks

\usepackage{tikz}
\usetikzlibrary{shapes, arrows.meta, positioning}

\tikzstyle{block} = [rectangle, draw, thick, minimum height=1cm, align=center]
\tikzstyle{arrow} = [thick,->,>=latex]

\tikzset{
  block/.style = {draw, rectangle, minimum height=2cm, minimum width=3.7cm, align=center},
  arrow/.style = {thick, -{Latex[width=2mm,length=2mm]}},
}

\theoremstyle{plain}
\newtheorem{assumption}{Assumption}
\newtheorem{theorem}{Theorem}[section]
\newtheorem{corollary}{Corollary}[theorem]
\newtheorem{lemma}[theorem]{Lemma}
\newtheorem{proposition}[theorem]{Proposition}
\theoremstyle{definition}
\newtheorem{definition}[theorem]{Definition}
\newtheorem{example}[theorem]{Example}

\theoremstyle{remark}
\newtheorem{remark}[theorem]{Remark}
\numberwithin{equation}{section}

\newcommand{\E}{\mathbb{E}}
\newcommand{\Var}{\mathrm{Var}}
\newcommand{\Bias}{\mathrm{Bias}}
\newcommand{\AMISE}{\mathrm{AMISE}}

\newcommand{\Vol}{\mathrm{Vol}}
\newcommand{\bs}{\boldsymbol}
\DeclareMathOperator*{\argmin}{arg\,min}

\setcopyright{cc}
\copyrightyear{2025}
\acmYear{2025}
\acmDOI{10.1613/jair.1.xxxxx}

\JAIRAE{Insert JAIR AE Name}
\JAIRTrack{Insert JAIR Track Name Here}
\acmVolume{4}
\acmArticle{111}
\acmMonth{8}
\acmYear{2025}

\begin{document}

\title[]{Non-parametric Formal Synthesis of Unknown Stochastic Systems: Asymptotic Convergence Guarantees}
\author{Zhi Zhang}
\authornote{Corresponding Author.}
\email{zhi.zhang90@outlook.com}
\affiliation{%
  \institution{University of Southampton}
  \city{Southampton}
  \country{United Kingdom}
}

\author{Sadegh Soudjani}
\email{sadegh@mpi-sws.org}
\affiliation{%
  \institution{MPI-SWS, Germany, and University of Birmingham}
  \country{United Kingdom}
}


\renewcommand{\shortauthors}{Zhang and Soudjani}

\begin{abstract}
Data-driven techniques have shown promising potential for checking behavior of complex systems operating in safety-critical domains against safety and other temporal requirements. This paper studies a class of data-driven techniques that are based on learning a representation of the system from data using non-parametric estimation.
The proposed approach is able to formally verify discrete-time stochastic dynamical systems against temporal logic specifications only using observation samples and without the knowledge of the model, and provides a probabilistic guarantee on the satisfaction of the specification.
We first consider finite abstract representations of the system in the form of Markov decision processes (MDPs) and derive asymptotic convergence guarantees between the transition probabilities of the abstract MDP and their estimation using Bernstein’s inequality and statistical properties of non-parametric estimators.
We then propose theoretical results for estimating the asymptotic upper bound of the \emph{Lipschitz constant} (LC) of the stochastic system, which can determine the size of the finite abstract MDP for a given precision error.
Under appropriate assumptions, our results prove that the asymptotic convergence rate of the estimations is $O(n^{-\frac{1}{3+\mathsf d}})$ for both the transition probabilities and the LC, where $\mathsf d$ is the dimension of the system and $n$ is the data scale.
By integrating these results, we can guarantee the asymptotic closeness in formal verification and policy synthesis performed on the original system and its finite abstraction based on the size of the dataset. Multiple case studies are presented to validate the effectiveness of the proposed method. 


%
%
\end{abstract}



\received{xxx}
\received[revised]{xxx}
\received[accepted]{xxx}

\maketitle

\section{Introduction}
\label{sec:intro}

Formal verification is extensively employed to ensure safety and satisfaction of other temporal requirements in safety-critical applications, including autonomous vehicles, power grids, medical robotics, and unmanned aerial systems.
Within this framework, model-based formal verification and synthesis \textcolor{black}{play} an essential role
\citep{kwiatkowska2002probabilistic,doyen2018verification,tabuada2009verification,belta2017formal}. However, complex systems interacting with unpredictable environments are challenging to model
\textcolor{black}{because of} black-box components \citep{sjoberg1995nonlinear} and the unpredictability of operating environments \citep{corso2021survey,kordabad2025data,yeh2018autonomous}.
The availability of large \textcolor{black}{amounts} of data from such systems necessitates developing data-driven techniques for formal verification and design of such systems with weak dependence on the information of the system's model.

 For systems under uncertainty, formal approaches rely on abstracting the system with \textcolor{black}{continuous} state space into a finite-state model such as Markov decision processes (MDPs) or interval MDPs that are amenable to automated analysis and computation \citep{baier2008principles,clarke1994model,Lavaei_Survey}. 
 For systems with an unknown model, constructing the MDP is not available, due to the difficulty in guaranteeing the closeness between the specifications of the original unknown system and its finite abstraction, and in the construction of transition probabilities between states of the finite abstraction.


The rapid progress in data-driven techniques--especially within the fields of machine learning and statistical learning, has opened new possibilities for formal verification and policy synthesis on dynamical systems, enabling the development of approaches that no longer rely on precise analytical models but instead learn directly from data \citep{fulton2018safe,xiang2018verification,makdesi2023data,schon2024btgp,kazemi2022datadriven}.
Approaches such as Reinforcement Learning (RL) \citep{jothimurugan2021compositional,fulton2018safe,lavaei2020formal,hasanbeig2023certified}, Deep Learning (DL) \citep{corsi2021formal}, Gaussian process regression \citep{jackson2021formal,reed2025error}, and non-parametric estimation \citep{zhang2024formal} have been integrated with formal verification and synthesis techniques, enabling the design of dynamical systems using data alone--without requiring precise system models.
However, relying on data inevitably introduces randomness into the obtained results. This inherent stochasticity raises important questions about whether these results can achieve asymptotic convergence, a key factor in determining their long-term reliability. To date, most existing works focus on derivations and implementations, without addressing the theoretical foundations of convergence in data-driven formal synthesis.

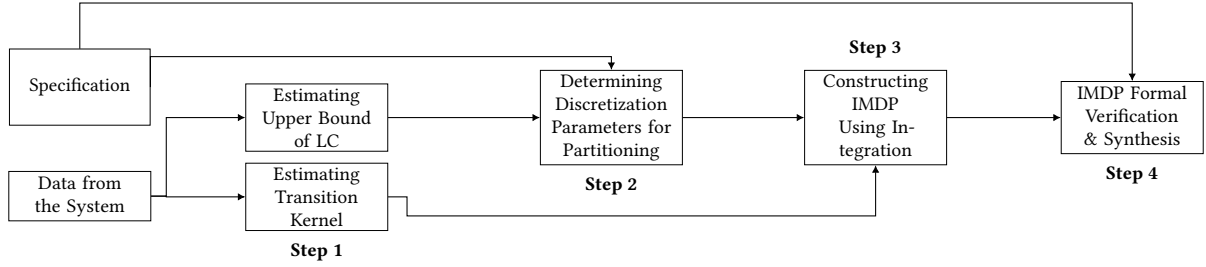
\begin{figure}
\centering
\resizebox{\linewidth}{!}{
\begin{tikzpicture}[node distance=1.2cm and 2.5cm, font=\LARGE]

\node[block] (spec) {Specification};
\node[block, below=of spec, text width=3.5cm, align=center, minimum height=1cm] (input) {Data from the System};

\node[block, right=of spec, yshift=-0.8cm, text width=3.5cm, align=center, minimum height=1cm] (upperLC) {Estimating\\Upper Bound of LC};
\node[block, below=0.3cm of upperLC,text width=3.5cm, align=center, minimum height=1cm] (kernel) {Estimating\\Transition Kernel};

\node[block, right=4.0cm of upperLC,text width=3.5cm, align=center, minimum height=1cm] (step2) {Determining Discretization\\Parameters for Partitioning};
\node[block, right=3.2cm of step2,text width=3.5cm, align=center, minimum height=1cm] (step3) {Constructing IMDP\\Using Integration};
\node[block, right=3.0cm of step3,text width=3.5cm, align=center, minimum height=1cm] (step4) {IMDP Formal\\Verification \& Synthesis};

\draw[arrow] (input.east) -- ++(0.4,0) |- (upperLC.west);
\draw[arrow] (input.east) -- ++(0.4,0) |- (kernel.west);
\draw[arrow] (upperLC.east) -- ++(0.5,0) -- ++(0,0) |- (step2.west);
\draw[arrow] (step2.east) -- ++(0.5,0) -- ++(0,0) |- (step3.west);
\draw[arrow] (step3.east) -- ++(0.3,0) -- (step4.west);

\draw[arrow] 
  (kernel.east) -- ++(0.5,0)      
  -- ++(0,-0.50)                   
  -- ++(11,0)                    
  -|                    
  (step3.south);               

\draw[arrow]
  (spec.east) -- ++(0,0.8) -| 
  (step2.north);

\draw[arrow] (spec.north) -- ++(0,1.2) -| (step4.north);

\node[below=0.2cm of kernel] {\textbf{Step 1}};
\node[below=0.2cm of step2] {\textbf{Step 2}};
\node[above=0.2cm of step3] {\textbf{Step 3}};
\node[below=0.2cm of step4] {\textbf{Step 4}};

\end{tikzpicture}
}
\caption{Workflow of the non-parametric formal verification framework.}
\label{verification_algorithm}
\end{figure}








In this paper, we introduce a novel formal verification and synthesis approach using non-parametric estimation (NPE) to guarantee asymptotic convergence when analyzing and synthesizing policies for unknown stochastic systems.
NPE has become a fundamental approach for analyzing problems without assuming a fixed parametric model \citep{hardle2004nonparametric,scott2015multivariate}. It has been successfully applied in a wide range of domains, including estimating the distribution of the Lipschitz constant in nonlinear components \citep{chakrabarty2020safe}, modeling the temporal evolution of dynamical systems \citep{hang2018kernel}, deep learning and short-term forecasting tasks \citep{huberman2021nonparametric}, and learning graphical models \citep{zhu2017learning}.
NPE also offers several favorable theoretical properties:
(i) its statistical behavior can be characterized by the mean and variance; and
(ii) it enables the estimation of conditional stochastic kernels in dynamical systems.
Despite these advantages, it remains unclear how to leverage NPE to establish asymptotic convergence guarantees for formal verification and policy synthesis of unknown stochastic systems.
To address this gap, we focus on the following key challenges:
(i) deriving asymptotic bounds on transition probabilities of finite abstractions of unknown systems by integrating the estimated stochastic kernel;
(ii) providing estimates of upper bounds on the \emph{Lipschitz constant} (LC) for stochastic systems;
and
(iii) incorporating closeness guarantees between the original system and its finite abstraction.

The primary contribution of this paper is to propose a non-parametric formal synthesis framework for unknown stochastic systems with asymptotic convergence guarantees as indicated in Fig.~\ref{verification_algorithm}.
Our key insight is that the transition probabilities in the MDP abstraction estimated using NPE admit asymptotic bounds with the convergence rate $O(n^{-\frac{1}{3+d}})$, where $d\ge1$ is the state dimension of the system and $n$ is the data scale. 
Using Bernstein's inequality, we derive asymptotic bounds on the error between the true and estimated transition probabilities.
We then adopt NPE to quantify asymptotic upper bound of the LC estimation, which shows the LC estimating range has asymptotic convergence rate $O(n^{-\frac{1}{3+d}})$. The bound on LC gives closeness guarantees between the system and its finite MDP abstraction.
These theoretical foundations demonstrate the reliability of our framework.
By incorporating these bounds and the asymptotic upper bounds of the estimated LC into the model-based closeness guarantees for finite abstractions of the system, we ensure the asymptotic convergence of the formal synthesis results for unknown stochastic systems.

\noindent
\textbf{Organization.}
The rest of this paper is structured as follows. After reviewing the related work, Section~\ref{sec2} introduces the preliminaries on the problem formulation, NPE for density functions, and formal abstraction-based policy synthesis using interval MDPs. Section~\ref{main} gives the main results of the paper, which are a systematic theoretical analysis of the asymptotic convergence properties
of the transition probabilities estimated through the NPE, LC estimation method (Step 1 in Fig.~\ref{verification_algorithm}), and the theoretical closeness guarantees for a given data scale. The asymptotic upper bound of the LC is used to determine discretization parameter of the partition for the unknown system (Step 2 in Fig.~\ref{verification_algorithm}).
Section~\ref{sec:IMDP_abstraction_NPE} presents the construction of the interval MDP using NPE and performing formal policy synthesis with convergence guarantees (Step 3 in Fig.~\ref{verification_algorithm}). Finally, the formal synthesis is performed on the constructed interval MDP (Step 4 in Fig.~\ref{verification_algorithm}) with case studies reported in Section~\ref{sec:case_studies}. 
Concluding remarks are in Section \ref{sec:conclusions}.
To maintain a clear focus on the main results, certain preliminaries and proofs have been relegated to the appendix.




\noindent
\textbf{Related Work.}
Formal verification of dynamical systems has traditionally relied on constructing precise mathematical models, which has been extensively studied for verifying safety, reachability, and temporal logic specifications in such systems, with formal analysis enabled by model-checking tools, e.g., \cite{kwiatkowska2009prism,hensel2022probabilistic,soudjani2015faust,van2023syscore}, with reports on available tool competitions \cite{abate2023arch,abate2024arch}.
\textcolor{black}{In this context, a substantial body of work has focused on the verification of Probabilistic Computation Tree Logic (PCTL) properties on stochastic systems, including approaches for efficient verification under model uncertainty \cite{puggelli2013polynomial} and, more recently, techniques for explaining and analyzing reinforcement learning policies through temporal logic specifications \cite{gross2025pctl}.}
However, these methods presuppose full knowledge of the system's dynamics, which is often impractical in real-world scenarios where systems are partially observable, high-dimensional, or entirely black-box.
Thus, recent research has shifted toward data-driven formal verification, where models or properties are estimated directly from system observations.


A growing body of research has focused on the theoretical foundations of data-driven formal verification and synthesis for unknown dynamical systems, aiming to provide rigorous guarantees despite limited model knowledge \cite{makdesi2023data,nazeri2025data,badings2023robust}.
The current advances in data-driven formal approaches for stochastic systems can be divided into two categories. The first category includes approaches that learn an abstraction from data \cite{gracia2023distributionally,banse2023data2,Chekan2023UncertainConstraints,makdesi2023data,schon2024btgp,kazemi2022datadriven,zhang2024formal,schon2024bayesian,nazeri2025data}.
The second category includes approaches that are based on learning barrier certificates ensuring satisfaction of the requirements using data without building a finite abstract model.
Many existing data-driven approaches to barrier certificates are restricted to specific system classes, particularly linear or control-affine dynamics, see, e.g., \cite{jagtap2020control,Cohen2022,Lopez2022uCBF}. In addition, several methods rely on known Lipschitz constants to provide formal guarantees or address only partially unknown system dynamics. For instance, Gaussian processes are used in \cite{Wang2018CBF} and \cite{jagtap2020control} to learn the unknown part of nonlinear system dynamics, while assuming that the control-affine component is known.
%

In contrast, approaches that handle fully unknown dynamics remain scarce. The work \cite{Salamati2021DDCBC} studies the synthesis of barrier functions for fully unknown discrete-time systems but still assume knowledge of the Lipschitz constant. For continuous-time systems, the paper \cite{wang2023stochastic} uses Bayesian inference and \textcolor{black}{relies} on local Lipschitz assumptions to manage uncertainty.
To address limitations in model structure, neural network-based methods have gained popularity for synthesizing expressive barrier functions—so-called neural barriers, due to their high representational capacity \cite{so2023train, abate2021fossil, Safe_Barrier_Neural}. Neural networks have also been used as compact representations for memory-efficient synthesis and verification \cite{majumdar2023neural}.
Adjacent lines of work explore temporal logic and reinforcement learning. The works \cite{kazemi2020fullLTL,kazemi2025average} investigate model-free reinforcement learning for temporal logic control with convergence guarantees. 
\textcolor{black}{Reinforcement learning has also been employed for statistical model checking of PCTL specifications on MDPs, enabling scalable, sampling-based verification in the absence of explicit models \cite{wang2020statistically}.}
Meanwhile, conditional mean embeddings have been applied in both abstraction-based and abstraction-free settings for correct-by-design policy synthesis \cite{Romao2023DRControl,schon2024DRObarrier}.

%



Non-parametric methods--such as Gaussian process regression \textcolor{black}{(GPR)} \citep{rasmussen2003gaussian}, non-parametric estimation (NPE) \citep{hardle2004nonparametric,scott2015multivariate}, and non-parametric least squares estimators \citep{ziemann2022single}--are widely used to model the dynamics of unknown systems using data alone, without assuming a specific parametric structure. These approaches have been applied in various contexts, including reachability analysis \citep{jackson2021formal,chowdhury2017kernelized} and safety verification \citep{jagtap2020control, ahmadi2017safety,reed2025error}.
NPE, in particular, has been employed to approximate the invariant density of dynamical systems \citep{hang2018kernel}. Among the methods mentioned above, NPE stands out for its suitability in settings where no underlying parametric form is assumed \citep{hardle2004nonparametric,scott2015multivariate}. For instance, NPE can estimate the probability density function of a random variable based on sampled data while preserving key statistical properties such as the mean and variance. It is also applicable to non-parametric function regression.
In the context of data-driven formal synthesis, NPE can be used to estimate the conditional density function, which captures the system’s probabilistic transitions. This makes it a valuable tool for verifying stochastic systems when model information is unavailable or incomplete.



\textcolor{black}{
Closely related to our work are the GPR-based methods proposed in \citep{chowdhury2017kernelized, reed2025error}, which also study stochastic systems within a non-parametric learning framework.
However, these approaches cannot be directly applied to our setting for the following reasons. First, the method in \citep{reed2025error} assumes \textcolor{black}{bounded-support} noise, while the method in \citep{chowdhury2017kernelized} requires the noise to be sub-Gaussian \textcolor{black}{(i.e., tail conditions)}.
In contrast, our work does not need these assumptions, as we consider unknown distributions for the noise.
Second, these GPR-based methods assume systems with additive noise structure of the form $  y = f(x) + w  $, whereas our setting does not impose such a restriction.
Third, the approach in \citep{reed2025error} \textcolor{black}{does not provide the estimation of} 
the upper bound of the \emph{Lipschitz constant} of transition kernel for the stochastic system. This estimation is used to determine the partition of the state space, which can decide the closeness between the original system and its finite abstraction with respect to satisfying temporal properties.
Finally, although the method in \cite{reed2025error} ensures safety via stochastic barrier functions under bounded-support noise, 
\textcolor{black}{it cannot be directly applied to the uniform cell-wise bounds required for IMDP abstraction in verifying PCTL properties.}
}

A subset of the results of this paper was presented at the International Conference on Artificial Intelligence and Statistics \cite{zhang2024formal}. This paper extends the results of \cite{zhang2024formal} substantially by providing:
(i) asymptotic convergence results for the estimated transition probabilities;
(ii) confidence bounds on the estimated LC of the system; and
(iii) asymptotic convergence of the estimated probabilities of satisfying the temporal specifications.

\section{Preliminaries and Problem Statement}
\label{sec2}
\textcolor{black}{Let $\Omega$ denote a generic sample space for all random variables of the this paper.}
In this section, 
we consider a discrete-time stochastic control system (DTSCS), which is a tuple $\Sigma_{ss}=(\mathcal S, U,w,f)$, 
where 
$\mathcal{S}\subset \mathbb R^{\textcolor{black}{\mathsf d}}$ is the continuous state space of the system,
\textcolor{black}{$U$ is the finite action space of the system,}
$w$ is a sequence of independent and identically-distributed (i.i.d.) random variables from \textcolor{black}{the} sample space $\Omega$ to the set $V_{w},$ i.e., $w:=\{ w(k):\Omega \to V_{w},k\in \mathbb N\},$
and $f:\mathcal{S}\times U\times V_{w}\to \mathcal{S}$ is a measurable function characterizing the state evolution of $\Sigma_{ss}$ as
\begin{align}
\label{dynamic_evolu}
    &\textcolor{black}{\bs x}(k+1)=f(\textcolor{black}{\bs x}(k),a(k),w(k)),\\
    &k\in \mathbb{N},~\textcolor{black}{\bs x}(k)\in \mathcal{S},~ a(k)\in U \text{ and } w(k)\in V_{w}.\nonumber
\end{align}
Also, we define a set $\mathcal{U}_{\mathfrak a}$ that is the collection of \textcolor{black}{ random sequences on the sample space $\Omega$},
$\{a(k):\Omega\to U,k\in \mathbb N \}$, with $a(k)$ being independent of $w(t)$ for any $k,t\in \mathbb N$ and $t\ge k$.
\textcolor{black}{Examples of such random sequences include Markov policies defined as follows.
Consider the trajectory given by $\nu_x=\bs x(0)\xrightarrow{a(0)}\bs x(1)\xrightarrow{a(1)}\bs x(2)\xrightarrow{a(2)}\ldots$.
Markov policies are characterized by functions $(\mu_0,\mu_1,\mu_2,\ldots)$ of the form $\mu_k: \mathcal{S} \to U$, such that the control action is selected as $a(k) = \mu_k(\bs x(k))$ for all $k\in\mathbb{N}$.
}

In our discussion, we assume that the system~\eqref{dynamic_evolu} is unknown but data from sampled trajectories is available.
Meanwhile, we expect system~\eqref{dynamic_evolu} satisfies a temporal specification $\psi$.


\smallskip
\noindent
\textcolor{black}{\textbf{Running Example.}
Consider the two-dimensional stochastic bistable switch adopted from \cite{dutreix2022abstraction} with dynamics
\begin{align}
\label{run_examp}
    \bs x(k+1)=f_r(\bs x(k)) + a(k) + w(k)=\begin{bmatrix}
        x_{1}(k)+[-\alpha x_{1}(k)+x_{2}(k)]\tau\\
        x_{2}(k)+\left( \frac{x^{2}_{1}(k)}{x^{2}_{1}(k)+1}-\beta x_{2}(k) \right)\tau
        \end{bmatrix} +\begin{bmatrix}
            u_1(k) \\
            u_2(k)
        \end{bmatrix} +\begin{bmatrix}
            w_{1}(k) \\
            w_{2}(k)
        \end{bmatrix},
\end{align} 
with $\alpha=1.3$, $\beta=0.25$, and $\tau=0.05$, state $(x_1(k),x_2(k))\in \mathcal{S} \subset  \mathbb R^2$, and action $a(k) = (u_1(k),u_2(k))\in U \subset  \mathbb R^2$. Although we do not assume prior knowledge of the noise distribution, it could for instance have Gaussian, Cauchy, or Laplace distributions characterized by their respective parameters.
Our ultimate objective is to design a policy for $a(k)$ to satisfy a given temporal specification with high probability. In the following sections, we will use this simple running example to illustrate the technical concepts and theorems.
}

The system~\eqref{dynamic_evolu} can be represented with a conditional density function $T(\textcolor{black}{\bs x}'|\textcolor{black}{\bs x},a)$ that maps the current state and \textcolor{black}{action} $(\textcolor{black}{\bs x},a)$ to a distribution over the next state $\textcolor{black}{\bs x}'\in \mathcal{S}$. We will utilize non-parametric estimation to estimate this conditional density function from data as discussed next.

\subsection{Non-parametric Estimation of Density Functions}\label{sec:prelem}
Let $X=(X_1,\ldots,X_\mathsf d)^T$ denote a $\mathsf d$-dimensional random vector which has a continuous probability density function $f_{X}:\mathbb R^{\mathsf d}\rightarrow\mathbb R_{\ge 0}$.
For a given set of i.i.d. random samples $ \{\hat{X}_i = (\hat{X}_{i1},\ldots,\hat{X}_{i\mathsf d})^T\in \mathbb{R}^{\mathsf d}| i=1,\ldots,n\}, n\in \mathbb N$, the general form of the multivariate kernel density estimator of $f_{X}(\cdot)$ is
\begin{equation}
\label{equ:MKDE}
\hat{f}_{X}(\bs{x})
=\frac{1}{n}
\sum_{i=1}^{n} 
K_{H}(\bs{x}-\hat{X}_i),\quad \forall \bs{x}\in \mathbb R^{\mathsf d},
\end{equation}
where
$K_{H}(\bs{u})=\frac{1}{|H|}K(H^{-1}\bs{u})$
and $K:\mathbb R^{\mathsf d}\rightarrow\mathbb R_{\ge 0}$ is a \emph{multivariate kernel function}
 satisfying two moment conditions $\int K(\bs{u}) \,d\bs{u}=1$ and $\int \bs{u}K(\bs{u}) \,d\bs{u}=\bs{0}$.
 $H$ is a non-singular \emph{bandwidth matrix} and
 $|H|$ denotes the determinant of $H$.
 Examples of the univariate kernel function $K(\cdot)$ include uniform, triangle, quartic, and Gaussian kernel functions (see Appendix~\ref{kernel_bandwidth}). Multivariate kernel functions are typically chosen to be the product of univariate kernel functions \citep{hardle2004nonparametric}, i.e., the same kernel function with different bandwidths in each dimension:
 $$K(u_1,u_2,\ldots,u_{\mathsf d}) = k(u_1)k(u_2)\ldots k(u_{\mathsf d})$$
for some univariate kernel function $k:\mathbb R\rightarrow\mathbb R_{\ge 0}$, and bandwidth matrix $H=diag(h_1,\ldots,h_{\mathsf d})$. A popular choice for the kernel function is the Gaussian kernel $k(u)=\frac{1}{\sqrt{2\pi}}\exp(-u^2/2) $, which leads to the following estimator for $f_X(\cdot):$
\begin{equation}
 \label{eq:gaussian}
    \hat{f}_{X}(\bs{x})
\!=\! \frac{(2\pi)^{-\mathsf d/2}}{nh_1\cdots h_{\mathsf d}} \cdot 
\sum_{i=1}^{n}
\prod_{j=1}^{\mathsf d}
\exp\left[-\frac{1}{2}\left(\frac{x_j-\hat{X}_{ij}}{h_j}\right)^2\right], 
\end{equation}
with $x_j$ and $\hat X_{ij}$ being the $j^{\text{th}}$ elements of $\bs{x}$ and $\hat X_i$, respectively.
We will use this estimator in the rest of this paper to establish our theoretical results.


The accuracy of the estimation is widely assessed using the mean integrated squared error (MISE), which is used for selecting the kernel function and the bandwidth matrix.
The \emph{asymptotic} MISE (AMISE), bias, and variance of the estimation is generally obtained by eliminating the higher-order terms.
To guarantee the best performance of the estimation, the choice of the bandwidth plays an important role in determining the AMISE.
An appropriately selected bandwidth ensures a balanced trade-off between bias and variance, thereby enhancing the accuracy of the estimator. In contrast, a poorly chosen bandwidth can lead to either under-smoothing—characterized by low bias but high variance—or over-smoothing, which yields low variance at the cost of increased bias. Minimizing the AMISE provides a principled way to navigate this trade-off and select an optimal bandwidth. The related details of the kernels and choice of the bandwidth can be found in Appendix \ref{kernel_bandwidth}.

\paragraph{Estimating Conditional Density Functions.}
\citet{rosenblatt1969conditional} introduced the standard kernel estimator of a conditional density function (CoDF) by replacing the estimates of the joint and marginal densities in the definition of the CoDF. Using a kernel estimator in \textcolor{black}{equation}~\eqref{equ:MKDE} that is the product of kernels for two random vectors $X$ and $Y$, the estimator of the joint density function of $(Y,X)$ and the marginal density function of $X$ are given by
\begin{align}
\label{eq:estimation_density_joint}
\hat{f}_{YX}(\bs{y},\bs{x})=&
\frac{1}{n}
\sum_{i=1}^n
K_{H_{\mathsf x}}(\bs{x}-\hat{X}_i)
\cdot K_{H_{\mathsf y}}(\bs{y}-\hat{Y}_i) \\
\label{eq:estimation_density}
\hat{f}_{X}(\bs{x})=&
\frac{1}{n}
\sum_{j=1}^n
K_{H_{\mathsf x}}(\bs{x}-\hat{X}_j),
\end{align}
where $K_{H_{\mathsf x}}$ and $K_{H_{\mathsf y}}$ are kernels with bandwidth matrices $H_{\mathsf x}$ and $H_{\mathsf y}$, respectively.
Thus, the kernel estimator of $f_{Y|X}$ is given by
\begin{align}
\hat{f}_{Y|X}(\bs{y},\bs{x})=
\frac{\hat{f}_{YX}(\bs{y},\bs{x})}{\hat{f}_{X}(\bs{x})}
=
\frac{\sum_{i=1}^n
K_{H_{\mathsf x}}(\bs{x}-\hat{X}_i)
\cdot K_{H_{\mathsf y}}(\bs{y}-\hat{Y}_i)}
{\sum_{j=1}^n
K_{H_{\mathsf x}}(\bs{x}-\hat{X}_j)}.
\label{condi_densityestima}
\end{align}

The asymptotic bias and variance of the conditional density estimator~\eqref{condi_densityestima} have been obtained by \citet{hyndman1996estimating} for univariate $X$ and $Y$, and can be found in the appendix in equation~\eqref{bias_from1996}-\eqref{var_from1996}. In the following sections, we use \textcolor{black}{equation}~\eqref{condi_densityestima} to estimate the CoDF of \textcolor{black}{equation}~\eqref{dynamic_evolu}, which gives the density function of the next state as a random vector conditioned on the current state and current \textcolor{black}{action}.
\textcolor{black}{For instance, in the running example, the conditional density associated with \eqref{run_examp} under action $a(k)$ is of the form
\begin{equation}
\label{run_examp_codf}
    f_{X_{k+1}\mid X_k}(\bs x(k+1)\mid \bs x(k)) = f_w(\bs x(k+1)-f_r(\bs x(k))-a(k)),
\end{equation}
where $f_w(\cdot)$ is the density function of the noise $w(\cdot)$. If the noise has a standard Cauchy distribution, then
\begin{align}
    \label{run_examp_codf1}
f_{X_{k+1}\mid X_k}(\bs x(k+1)\mid \bs x(k))
= \frac{1}{\pi^2\left[1+(x_1(k+1)-f_r^1(\bs x(k))-u_1(k))^2\right]\left[1+(x_2(k+1)-f_r^2(\bs x(k))-u_2(k))^2\right]},
%
\end{align}
where $f_r = (f_r^1,f_r^2)$ is defined in \eqref{run_examp}.
If the noise is Gaussian with zero mean and covariance $\Sigma$, then
\begin{align}
\label{run_examp_codf2}
f_{X_{k+1}\mid X_k}(\bs x(k+1)\mid \bs x(k))
= \frac{1}{(2\pi) |\Sigma|^{\frac{1}{2}}}
\exp\!\left(
-\frac{1}{2}
\left[\bs x(k+1)- f_r(\bs x(k))-a(k)\right]^{\mathsf T}
\Sigma^{-1}
\left[\bs x(k+1)- f_r(\bs x(k))-a(k)\right]
\right).
\end{align}
These conditional density functions can be estimated by the kernel-based estimator \eqref{condi_densityestima}.
}



\subsection{Interval Markov Decision Processes}
An interval Markov decision process (IMDP) is a type of Markov decision process in which the transition probabilities are not fixed values but instead lie within specified intervals \citep{givan2000bounded}. This representation captures uncertainty or variability in the system's dynamics, allowing each transition to be associated with a range of possible probabilities rather than a single exact value. IMDPs are particularly useful in data-driven and abstraction-based settings, where exact transition probabilities may be difficult to obtain but bounds can still be derived from data. This interval-based structure enables robust verification and control synthesis by accounting for all possible behaviors within the specified probability ranges, ensuring that properties such as safety or reachability hold under worst-case scenarios.

\begin{definition}[\textbf{IMDP}]
An IMDP is a tuple $\Sigma=( Q, S_{\mathfrak a}, P_{lo}, P_{up}, AP, L ),$ where $Q$ is a finite set of states, $S_{\mathfrak a}$ is a finite set of actions and $S_{\mathfrak a}(q)$ is the set of actions at state $q\in Q$, $P_{lo}: Q\times S_{\mathfrak a}\times Q\to [0,1]$ is a function representing the lower bound of the transition probability from $q\in Q$ to $q^*\in Q$ under action $a\in S_{\mathfrak a}$, $P_{up}: Q\times S_{\mathfrak a}\times Q\to [0,1]$ is a function representing the upper bound of the transition probability from $q$ to $q^*$ under action $a\in S_{\mathfrak a}$, $AP$ is a finite set of atomic propositions, and $L:Q\to 2^{AP}$ is a labeling function assigning possibly several elements of $AP$ to each state $q$.
\end{definition}

For any $q,q^{*}\in Q$ and $a\in S_{\mathfrak a}(q)$, it holds that $P_{lo}(q,a,q^{*})\leq P_{up}(q,a,q^*)$ and $\sum_{q\textcolor{blue}{^*}\in Q}P_{lo}(q,a,q^*)\leq1\leq \sum_{q\textcolor{blue}{^*}\in Q}P_{up}(q,a,q^*)$. This is to ensure a non-empty set of feasible probability distributions for transitions between the states. 
The set of probability distributions over $Q$ is denoted by $D(Q)$.
$\theta^{a}_{q}\in D(Q)$ represents a feasible distribution initiated from $q\in Q$ to all successor states in $Q$ under $a$, and satisfies $P_{lo}(q,a,q^*)\leq \theta^{a}_{q}(q^*)\leq P_{up}(q,a,q^*)$, where $q^*$ is the successor state. The set of all feasible distributions initiated from $q$ under $a$ is denoted by $\Theta^{a}_{q}$.
\textcolor{black}{A path of the IMDP is a sequence $\nu=q_{0}\xrightarrow{a_0}q_{1}\xrightarrow{a_1}q_{2}\xrightarrow{a_2}\ldots,$ where $a_{i}\in S_{\mathfrak a}(q_{i})$, and satisfies $P_{up}(q_{i}, a_{i}, q_{i+1})>0$ for all $i$. When the path is finite, the}
last state of a finite path $\nu^{\mathsf{fin}}$ is denoted by $\textsf{last}(\nu^{\mathsf{fin}})$.
The sets of all finite and infinite paths are denoted by $\textsf{Paths}^{\mathsf{fin}}$ and $\textsf{Paths}$, respectively.
Let a function $\varpi: \textsf{Paths}^{\mathsf{fin}}\to S_{\mathfrak a}$ denote a strategy on the IMDP $\Sigma$, which maps a finite path $\nu^{\mathsf{fin}}$ of $\Sigma$ onto an action in $S_{\mathfrak a}$. The set of all such strategies is denoted by $\Pi$.
An MDP is an IMDP with all probability intervals being a singleton (i.e., with $P_{lo} = P_{up}$).

\begin{definition}[\textbf{Adversary}]
Consider an IMDP $\Sigma$. An adversary is a function $\kappa: \textsf{Paths}^{\mathsf{fin}}\times S_{\mathfrak a}\to D(Q)$, which maps the path-action pair $(\nu^{\mathsf{fin}},a)$  with $a\in S_{\mathfrak a}(\textsf{last}(\nu^{\mathsf{fin}}))$ to a feasible distribution $\theta^{a}_{q}\in\Theta^{a}_{\textsf{last}(\nu^{\mathsf{fin}})}$. The set of all adversaries is denoted by $\mathcal{K}$.
\end{definition}

\subsection{Probabilistic Computation Tree Logic (PCTL)}
\label{PCTL}
PCTL \citep{hansson1994logic} is a formal language for expressing requirements on complex behaviors of stochastic systems. It extends classical temporal logics by introducing probabilistic operators, allowing one to specify properties such as ``the probability of eventually reaching a safe state is at least 0.95'' or ``with at least 90\% probability, a failure does not occur within 10 steps.'' PCTL is widely used in the formal verification of systems modeled by MDPs and IMDPs, where it enables reasoning about uncertainty and nondeterminism. Compared to purely qualitative logics such as Linear Temporal Logic (LTL) or Computation Tree Logic (CTL) \citep{baier2008principles}, PCTL is better suited for systems with stochastic dynamics since it allows quantitative reasoning over probabilities, enabling the specification of both functional behavior and probabilistic guarantees. By providing a rigorous specification language, PCTL serves as a foundation for verifying properties like safety, liveness, and reachability in stochastic systems under quantitative constraints.

\begin{definition}[\textbf{Syntax of PCTL}]
For a given set of atomic propositions $AP$, formulas in PCTL can be recursively defined as follows:
\begin{align*}
    &\text{State Formula } \phi:=\textsf{true}~|~\rho~|~\neg \phi~|~\phi \wedge \phi~|~ P_{\bowtie p}[\psi],\\
    &\text{Path Formula }~ \psi:= \mathcal{X} \phi~|~\phi~\mathcal{U}^{\leq k}~\phi~|~ \phi~\mathcal{U}~ \phi, 
\end{align*}
where $\rho \in AP$, $\neg$ is the negation operator, $\wedge$ is the conjunction operator, $P_{\bowtie p}$ is the probabilistic operator, $\bowtie\in \{\leq, <, \ge,>\} $ is a relation placeholder, and $p\in [0,1]$. $\mathcal{X}$ (next), $\mathcal{U}^{\leq k}$ (bounded until), and $\mathcal{U}$ (until) are temporal operators.
\end{definition}

\begin{definition}[\textbf{PCTL Semantics}]
For a labeling function $L:Q\rightarrow 2^{AP}$, the satisfaction relation $\models$ is defined inductively as follows. For any state $q\in Q$,
(1) $q\models\textsf{true}$ for all $q\in Q$;
(2) $q\models \rho \iff \rho \in L(q)$;
(3) $q\models (\phi_1 \wedge \phi_2) \iff (q\models \phi_1) \wedge (q\models \phi _2)$;
(4) $q\models \neg \phi \iff q\not\models \phi$;
(5) $q\models P_{\bowtie p} [\psi] \iff \textsf{Prob}_{q}(\psi)\bowtie p,$ where $\textsf{Prob}_{q}(\psi)$ is the probability that infinite trajectories originating from $q$ satisfy $\psi$.
Also, for any path $\upsilon\in \textsf{Paths}$, the satisfaction relation $\models$ is defined as: (1) $\upsilon \models \mathcal{X}\phi \iff \upsilon(1)\models \phi;$ (2) $\upsilon \models \phi_1 \mathcal{U}^{\leq k}\phi_2 \iff  \exists i\leq k~s.t.~\upsilon(i)\models \phi_2 \wedge \upsilon(j)\models \phi_1,~\forall~j\in [0,i);$ (3) $\upsilon\models \phi_1\mathcal{U}\phi_2 \iff \exists i\ge 0~s.t.~\upsilon(i)\models\phi_2 \wedge \upsilon(j)\models \phi_1,~\forall~j\in [0,i).$
\end{definition}

The specification \emph{bounded eventually} $\Diamond^{\leq k}$
is defined as $P_{\bowtie p}[\Diamond^{\leq k }\phi ]\equiv P_{\bowtie p}[\textsf{true}~ \mathcal{U}^{\leq k} \phi]$ representing that $\phi$ is satisfied within $k$ time steps.
The specification \emph{eventually} $\Diamond$ is defined as $P_{\bowtie p}[\Diamond \phi]\equiv P_{\bowtie p}[\textsf{true}~\mathcal{U} \phi]$ representing that $\phi$ is satisfied at some point in the future. 

\textcolor{black}{
For the running example \eqref{run_examp}, we consider a bounded-horizon specification $\psi=\neg r_{O} \mathcal{U}^{\leq K} r_D$. 
This specification requires that the system does not visit $r_O$ until visiting $r_D$ in $K$ steps, where $r_D$ and $r_O$
are subsets of the state space of the system indicating the destination and avoiding regions, respectively.}

\subsection{IMDP Policy Synthesis}
Here, we give an account of IMDP verification and policy synthesis against a specification described in PCTL.
For a PCTL path formula $\psi$ starting from an initial state $q\in Q$, the lower and upper bound of probabilities that the paths initialized at $q$ satisfy $\psi$ in $k$ steps can be defined as 
\begin{align}\label{lowprob_path} 
    P^{k}_{lo}(q)\!=\!\begin{cases}
       1,~\text{if}~q\in Q^{1},\\
    0,~\text{if}~q\in Q^{0}, \\
    0,~\text{if}~q\notin (Q^{0}\cup Q^{1})\wedge k=0,\\
    \textcolor{black}{\max_{a}}\min_{\theta^{a}_{q}} \sum_{q^{*}}\theta^{a}_{q}(q^{*})P^{k-1}_{lo}(q^{*}),~\text{otherwise},
\end{cases}
\end{align}
\begin{align}
\label{upprob_path}
     P^{k}_{up}(q)\!=\!\begin{cases}
       1,~\text{if}~q\in Q^{1},\\
    0,~\text{if}~q\in Q^{0},\\
    0,~\text{if}~q\notin (Q^{0}\cup Q^{1})\wedge k=0,\\
    \!\max_{a}\max_{\theta^{a}_{q}} \sum_{q^{*}}\theta^{a}_{q}(q^{*})P^{k-1}_{up}(q^{*}),~\text{otherwise},
\end{cases}
\end{align}
where $Q^{1}$ is the set of states that always satisfy the path formula $\psi$, $Q^{0}$ is the set of states that never satisfy $\psi$, and $\theta^{a}_{q}(q^*)\in [P_{lo}(q,a,q^*),P_{up}(q,a,q^*)]$, for any $q^{*}\in Q$.
The adversaries obtained from the procedures above determine a series of actions that lead to maximum and minimum probabilities satisfying path formula $\psi$ for each state.
The above recursive computation of the probability bounds can be performed in a finite number of steps for specifications with bounded until ($\mathcal{U}^{\leq k}$). The number of steps can be tuned with respect to any desired accuracy for specifications with unbounded until ($\mathcal{U}$) \cite{haddad2018interval}. Once the satisfaction of formulas of the form $P_{\bowtie p}[\psi]$ is checked, the satisfaction of other state formulas does not involve a probability operator and can be checked using usual binary techniques on subsets of the state space. \textcolor{black}{Due to these reasons, we focus on Markov policies that are sufficient for checking specifications of the form $P_{\bowtie p}[\phi_1\mathcal{U}\phi_2]$ and $P_{\bowtie p}[\phi_1\mathcal{U}^{\le k}\phi_2]$ with state formulas $\phi_1$ and $\phi_2$.} 

\subsection{Problem Formulation}\label{pro_formula}

\textcolor{black}{When the DTSCS $\Sigma_{ss}=(\mathcal S, U,w,f)$ is known, an IMDP $\bar \Sigma_{ss}=(Q, S_{\mathfrak a}, P_{lo},P_{up}, AP, L )$ is constructed as a finite abstraction of the system.
Define $Q$ as a partition of the state space $\mathcal S$ with partition sets denoted by $q\in Q$, where $n_{Q}:=|Q|.$
Define the transition probabilities $P_{ij,a}: q_{i}\to \mathbb R ,$ such that $P_{ij,a} (\bs x):=\textsf{Prob}_w(f(\bs x,a,w)\in q_j)$, where
$\textsf{Prob}_{w}(\cdot)$ denotes probability with respect to the distribution of the random noise $w$, 
$\bs x\in q_i$ and $q_{i}, q_{j}\in Q,$ $i,j\in\{1,\ldots,n_Q\}.$
Define $P_{lo,a}(q_{i},q_{j})=\min_{\bs x\in q_{i}}P_{ij,a}(\bs x)$ and $P_{up,a}(q_{i},q_{j})=\max_{\bs x\in q_{i}}P_{ij,a}(\bs x).$
The \textcolor{black}{action} space $S_{\mathfrak a} = U$. 
}
The length of the interval on transition probabilities of $\bar \Sigma_{ss}$ and its effect on the probability of satisfying the specification can be bounded through the Lipschitz constant of $\Sigma_{ss}$ as formally stated in Theorem~\ref{SA13_bound} in Section~\ref{supp_pro_formula}.
When the system $\Sigma_{ss}$ is unknown, it needs to be approximated by an IMDP $\hat \Sigma_{ss}=(Q, S_{\mathfrak a}, P_{lo},P_{up}, AP, L )$ based on the data from sampled trajectories. For this purpose, the transition probabilities $P_{ij,a}(\bs x)$ can be estimated through the non-parametric estimation based on the samples $\left\{\left(\textcolor{black}{\hat{X}_{i},\hat{ Y}_{i}} \right),~i=1,\ldots,n \right\}$, \textcolor{black}{where $\textcolor{black}{\hat{X}_{i}},~i=1,2,\ldots,n$ are i.i.d. random samples, and samples $\textcolor{black}{\hat{Y}_{i}}$ are generated from $(\textcolor{black}{ Y|\hat{ X}_{i}})$ for all $i=1,2,\ldots,n$}.
Let $ \hat{P}_{ij,a} (\bs x)$ represent the probability estimated using non-parametric methods, and refer to $\hat \Sigma_{ss}$ as the NPE abstraction of $\Sigma_{ss}$.
However, it is unclear whether the formal verification and synthesis based on this data-driven abstraction $\hat \Sigma_{ss}$ can guarantee probabilistic asymptotic convergence to the behavior of the model-based finite abstraction $\bar{\Sigma}_{ss}$. 
Such problem gives rise to a fundamental question: whether the probability of satisfying the specification $\psi$ under $\hat \Sigma_{ss}$ can converge, with high confidence, to that of 
$\Sigma_{ss}$ within a predefined threshold.
This problem can also be found in other data-driven formal verification frameworks, such as those presented in \cite{jackson2021formal, nazeri2025data}.
Thus, this paper aims at proposing a non-parametric framework for the formal verification and synthesis of unknown stochastic systems, ensuring asymptotic convergence guarantees.




\begin{tcolorbox}[colback=gray!10, colframe=gray!50, coltitle=black, sharp corners, enhanced]

\noindent\textbf{Main Problem.}
Let $\psi$ be a PCTL specification and $\Sigma_{ss}=(\mathcal S, U,w,f)$ a DTSCS as in \textcolor{black}{the system}~\eqref{dynamic_evolu}, where $f$, the distribution of $w$, and the Lipschitz constant of $\Sigma_{ss}$ are unknown.
Synthesize a control policy and verify the unknown DTSCS $\Sigma_{ss}$ against the specification $\psi$ through NPE abstraction $\hat \Sigma_{ss}$.
Provide guarantees that, as the data scale $n$ increases, the formal verification and control synthesis results obtained from the NPE abstraction $\hat \Sigma_{ss}$ converge, with high probability, to those of $\Sigma_{ss}$ within a predefined threshold.
\end{tcolorbox}



\section{Main Results}
\label{main}
In this section, we first present a systematic theoretical analysis of the asymptotic convergence properties of transition probabilities estimated through the NPE. Such property can ensure that the data-driven NPE abstraction $\hat{\Sigma}_{ss}$ asymptotically converges to the finite abstraction $\bar{\Sigma}_{ss}$ of the system \eqref{dynamic_evolu}, with a high probability. 
Next, we derive an asymptotic upper bound of the LC of CoDF of the stochastic system $\Sigma_{ss}$ using the NPE.
According to Theorem \ref{SA13_bound} \cite{SA13}, this bound guides the choice of partition parameters for constructing the finite abstraction of system \eqref{dynamic_evolu}, ensuring probabilities of satisfying specifications remain close between the original system and its abstraction, within a predefined threshold. 
Together, these results guarantee that the probability of satisfying the specification $\psi$ under $\hat \Sigma_{ss}$ can converge, with high probability, to that of  $\Sigma_{ss}$ within a predefined threshold. 
This provides a theoretical guarantee for the reliability of data-driven verification and synthesis.



\subsection{Preparation of the Main Results}

In this subsection, we theoretically analyze the asymptotic convergence of the estimated transition probabilities and the LC of the system \eqref{dynamic_evolu} on a given domain $D_{X}\times D_{Y}$ using \textcolor{black}{equation}~\eqref{condi_densityestima}.
In the following discussion throughout the rest of the paper, we assume that the bandwidth matrices $H_{\mathsf x}$ and $H_{\mathsf y}$ are selected as
$
H_{\mathsf x} = H_{\mathsf y} = h I_{\mathsf d},
$
where $h > 0$ is a scalar bandwidth parameter and $I_d$ denotes the $\mathsf d \times \mathsf d$ identity matrix.
\textcolor{black}{Also, for the sake of simplicity in our discussions, we assume that the state vector $ X=( X_{1},\ldots, X_{\mathsf d})$ is sampled uniformly at random from $D_{X}$, which means the corresponding samples $\{\hat{X}_{i},i=1,\ldots,n,n\in \mathbb{N}\}$ are i.i.d. with uniform distribution on $D_{X}$, whose probability density function is given by 
$f_X(\bs x)=\frac{\mathbf{1}_{D_X}(\bs x)}{\Vol(D_X)}$, where $\mathbf{1}_{D_X}(x)$ denotes the indicator function of $D_X$, and 
$\Vol(\cdot)$ denotes the volume (Lebesgue measure) of a set.
}


We use the symbol $\lesssim$ to denote the asymptotic bound when $n\to +\infty$ by eliminating higher order terms.
In addition, our results require the following assumption that assumes existence of bounds on higher order partial derivatives of the CoDF.

\begin{assumption}\label{asstwod_1}
\textbf{{\normalfont(a)}}
There exists a constant $C_{f}\!>\!0$ such that $| f_{Y|X}(\bs{y},\bs{x})|\!\leq\! C_{f}$, for all $(\bs{x},\bs{y})\in D_{\bs X}\!\times\! D_{\bs Y}$.
\textbf{{\normalfont(b)}}
There exist constants $C_{\textcolor{black}{j}},C^{*}_{i}\!>\!0$, $i,\textcolor{black}{j}\in \{1,\ldots,\mathsf d\}$, such that $|\frac{\partial^{2}}{\partial y_{j}^{2}}f_{Y|X}(\bs{y},\bs{x})|\!\leq\! C_{i}$ and $| \frac{\partial^{2}}{\partial x^{2}_{i}}f_{Y|X}(\bs{y},\bs{x})|\!\leq\! C^{*}_{i}$, for all $(\bs{x},\bs{y})\in D_{\bs X}\!\times\! D_{\bs Y}$.
\textbf{{\normalfont(c)}}
There exist constants $C_{ij}\!>\!0$, $i,j\in \{1,\ldots,\mathsf d\}$, such that $|\frac{\partial^{3}}{\partial x_{i}\partial y_{j}^{2}}f_{Y|X}(\bs{y},\bs{x})|\!\leq\! C_{ij}$, for all $(\bs{x},\bs{y})\in D_{\bs X}\!\times\! D_{\bs Y}$.
\textbf{{\normalfont(d)}}
There exist constants $\textcolor{black}{\bar{C}_{i\varsigma}}\!>\!0$, $i,\varsigma\in \{1,\ldots,\mathsf d\}$, such that $| \frac{\partial^{3}}{\partial x_{i}\partial x^{2}_{\varsigma}}f_{Y|X}(\bs{y},\bs{x})|\!\leq\! \bar{C}_{i\varsigma}$, for all $(\bs{x},\bs{y})\in D_{\bs X}\!\times\! D_{\bs Y}$.
\end{assumption}

\begin{remark}
\textcolor{black}{
    Assumption~\ref{asstwod_1} requires that the noise $w(k)$ in the system dynamics \eqref{dynamic_evolu} induces a one-step transition kernel that admits a sufficiently smooth conditional density $f_{Y \mid X}$. In particular, we require the conditional density function 
    $f_{Y \mid X}(\bs y|\bs x)$ to \textcolor{black}{admit bounded second-order partial derivatives with respect to both $\bs x$ and $\bs y$, together with bounded mixed third-order partial derivatives.}
    This can be satisfied by requiring a sufficiently smooth density function of the noise on its support and a sufficiently smooth vector field $f$ in \eqref{dynamic_evolu}. 
    The running example in \eqref{run_examp} already satisfies this requirement since $f_r(\bs x)$ is a rational function with a strictly positive polynomial in its denominator, $w(k)$ is an additive noise, and either of the density functions in \eqref{run_examp_codf1}-\eqref{run_examp_codf2} are \textcolor{black}{sufficiently smooth and satisfy the derivative bounds required by
Assumption~\ref{asstwod_1}.}
    More generally, our framework accommodates a broad class of noise distributions and system dynamics beyond the additive case, provided that \textcolor{black}{the induced transition density admits the required bounded second-order and mixed third-order derivatives.}
    }
\end{remark}

For the sake of notational simplification, we introduce the convolution operator $*.$
Let us denote $\bs k_{\bs x,h}(\cdot):\mathbb{R}^{\mathsf d}\to \mathbb{R}$ as 
$\bs k_{\bs x,h}(\bs z):=\bs k_{h}(\bs x-\bs z)=h^{-\mathsf d}K(\textcolor{black}{(}\bs x-\bs z \textcolor{black}{)}/h)=h^{-\mathsf d}\prod^{\mathsf d}_{i=1}k\left((x_{i}-\bs z_{i})/h\right).$
To carry out our theoretical discussion, we define the following terms for density functions $f_{X}(\bs x)$ and $f_{YX}(\bs y, \bs x)$: 
\begin{align*}
    &\Bar{f}(\bs x) 
    := \bs k_{\bs x,h} * f_{X}(\bs x) 
     = \int f_{X}(\bs x_{\tau}) \bs k_{h}(\bs x - \bs x_{\tau}) \, d\bs x_{\tau}, \\
    &\Bar{f}_{x}(\bs y, \bs x) 
    := \bs k_{\bs x,h} * f_{YX}(\bs y, \bs x) 
     = \int f_{YX}(\bs y, \bs x_{\tau}) \bs k_{h}(\bs x - \bs x_{\tau}) \, d\bs x_{\tau}, \\
    &\Bar{f}(\bs y, \bs x) 
    := \bs k_{\bs y,h} * \Bar{f}_{x}(\bs y, \bs x)
     = \int \Bar{f}_x(\bs y_{\tau}, \bs x) \bs k_{h}(\bs y - \bs y_{\tau}) \, d\bs y_{\tau}
    \phantom{:}= \iint f_{YX}(\bs y_{\tau}, \bs x_{\tau}) \bs k_h(\bs x - \bs x_{\tau}) \bs k_h(\bs y - \bs y_{\tau}) \, d\bs x_{\tau} d\bs y_{\tau}.
\end{align*}



Thus, we can have 
\begin{align*}
    \Bar{f}(\bs y,\bs x)=\E \left[  \bs k_{\bs x,\bs y,h}(\textcolor{black}{ X,  Y})\right] \text{ and }
    \Bar{f}(\bs x)=\E\left[ \bs k_{\bs x,h}(\textcolor{black}{ X}) \right], 
\end{align*}
where $\bs k_{\bs x,\bs y,h}(\textcolor{black}{ X,Y}):=\bs k_{\bs x,h}(\textcolor{black}{ X})\bs k_{\bs y,h}(\textcolor{black}{ Y}) .$



We need to introduce the Bernstein inequality \citep{bernstein1924modification} and the following lemma.

\begin{theorem} [\textbf{Bernstein Inequality} \citep{bernstein1924modification}]\label{Bern_ineq}
    Let $X_{1},\cdots,X_{n}$ be i.i.d. random variables  with $\mu=\E[X_{i}]$ and $\delta^2=\Var [X_{i}]$, and for any $i\in \{1,\cdots,n\},$ there exists a positive real number $R>0$, such that $|X_{i}-\mu |\leq R,$ then
    \begin{align*}
        P\left(\left| \frac{1}{n}\sum^{n}_{i=1} X_{i}-\mu \right|\leq \epsilon \right)\ge 1-2\exp(-\tau),
    \end{align*}
    where 
    \begin{align*}
        \epsilon:=\max\left\{\frac{\tau R}{3n}+\sqrt{ \frac{\tau^2 R^2}{9n^2}+\frac{2\tau\delta^2}{n}  }, \frac{\tau R}{3n}-\sqrt{ \frac{\tau^2R^2}{9n^2}+\frac{2\tau\delta^2}{n} }  \right\}.
    \end{align*}
\end{theorem}

Following a similar approach to that used in \cite{hyndman1996estimating}, we get the next lemma.

\begin{lemma}\label{multi_expect}
    Suppose that $X$ and $Y$ are $\mathsf d$-dimensional random variables, there is a CoDF $f_{Y|X}(\bs y,\bs x)$ for all $(\bs x,\bs y)\in D_X\times D_Y$ and a \textcolor{black}{smooth} density function $f_{X}(\bs x)$ for random variable $X$, and $k:\mathbb{R} \to \mathbb{R}$ is the Gaussian kernel function with bandwidth $h,$ 
    where $\bs x:=(x_{1},\ldots,x_{\mathsf d})^{T}$, and $\bs y:=(y_{1},\ldots,y_{\mathsf d})^{T}$. In addition, $\hat{X}$ is the random sample of $X$.
We have that if $h\to 0$ as $n\to+\infty$, then
\begin{align*}
&\E_{p}\left[ \bs k_{\bs{x},\bs{h} }(X) \right] 
= f_{X}( \bs{x} )+\frac{1}{2} h^{2} G_{12}(k)\sum^{\mathsf d}_{i=1}\frac{\partial^2}{\partial \bs x_{i}^2}f_{X}(\bs x) +O(h^{4}), \\[1ex]
&\E_{p}\left[ \frac{X_{i}-x_{i}}{h}\bs k_{\bs x,\bs h}(X) \right] 
= hG_{12}(k)\frac{\partial}{\partial\bs x_i}f_{X}(\bs x)+O(h^3), \\[1ex]
&\E_{p}\left[ \bs k^2_{\bs x,\bs h}(X) \right]
= h^{-\mathsf d} \left\{ G_{20}(k)f_{X}(\bs x)
    +\frac{1}{2}h^{2} G_{22}(k)\sum^{\mathsf d}_{i=1}\frac{\partial^{2}}{\partial\bs x_{i}^{2}}f_{X}(\bs x)+O(h^{4}) \right\}, \\[1ex]
&\E_{p}\left[ \left(\frac{X_{i}-x_{i}}{h}\right)^2 \bs k^2_{\bs x,\bs h}(X) \right]
= h^{-\mathsf d}\left\{  G_{22}(k)f_{X}(\bs x)+\frac{1}{2}h^2 G_{24}(k)\sum^{\mathsf d}_{i=1}\frac{\partial^{2}}{\partial\bs x_{i}^{2}}f_{X}(\bs x)+O(h^4) \right\}, \\[1ex]
&\E_{p}\left[ \bs k^4_{\bs x,\bs h}(X) \right]
= h^{-3\mathsf d}\left\{ f_{X}(\bs x)G_{40}(k)+\frac{1}{2}h^{2}G_{42}(k)\sum^{\mathsf d}_{i=1}\frac{\partial^2}{\partial\bs x_{i}^2}f_{X}(\bs x)
    +O(h^{4}) \right\}, \\[1ex]
&\E_{p}\left[\left(\frac{X_{i}-x_{i}}{h}\right)^4 \bs k^4_{\bs x,\bs h}(X) \right]
= h^{-3\mathsf d}\left\{ f_{X}(\bs x)G_{44}(k) +\frac{1}{2} h^{2}G_{46}(k)\sum^{\mathsf d}_{i=1}\frac{\partial^2}{\partial\bs x_{i}^2}f_{X}(\bs x)+O(h^{4}) \right\}, \\[1ex]
&\E_{p}\left[ \bs k_{\bs x,\bs h}(X)\bs k_{\bs y,\bs h}(Y) \right]
= f_{X}(\bs x)f_{Y|X}(\bs y,\bs x)+\frac{1}{2} h^2 G_{12}(k)f_{X}(\bs x)\bigg[ \sum^{\mathsf d}_{i=1}\frac{\partial^{2}}{\partial\bs y_{i}^{2}}f_{Y|X}(\bs y,\bs x)
     +\sum^{\mathsf d}_{i=1}\frac{\partial^{2}}{\partial\bs x_{i}^{2}}f_{Y|X}(\bs y,\bs x) \bigg] +O(h^{4}), \\[1ex]
&\E_{p}\left[ \bs k^{2}_{x,h}(X)\bs k^{2}_{y,h}(Y) \right]
= h^{-2\mathsf d} \Bigg\{ G^2_{20}(k)f_{X}(\bs x)f_{Y|X}(\bs y,\bs x)+\frac{1}{2}h^2 G_{20}(k)G_{22}(k)f_{X}(\bs x) \bigg[ \sum^{\mathsf d}_{i=1}\frac{\partial^{2}}{\partial\bs y_{i}^{2}}f_{Y|X}(\bs y,\bs x)\\
    &\quad +\sum^{\mathsf d}_{i=1}\frac{\partial^{2}}{\partial\bs x_{i}^{2}}f_{Y|X}(\bs y,\bs x) \bigg] +O(h^4) \Bigg\}, 
\end{align*}
\begin{align*}
&\E_{p}\left[ \left(\frac{X_{i}-x_{i}}{h}\right)^2 \bs k^{2}_{x,h}(X)\bs k^{2}_{y,h}(Y) \right]
= h^{-2} \left\{ G_{20}(k)G_{22}(k)f_{X}(\bs x)f_{Y|X}(\bs y,\bs x)
     +\frac{1}{2}h^{2}G^{2}_{22}(k)f_{X}(\bs x)\sum^{\mathsf d}_{i=1} \frac{\partial^{2}}{\partial\bs y_{i}^{2}}f_{Y|X}(\bs y,\bs x) \right. \\
&\quad \left. + \frac{1}{2}h^{2}G_{20}(k)G_{24}(k)\sum^{\mathsf d}_{i=1} \frac{\partial^{2}}{\partial\bs x_{i}^{2}}f_{Y|X}(\bs y,\bs x)+O(h^{3}) \right\}.
\end{align*}
where $G_{ji}(k)=\int \nu^{i}k^{j}(\nu)d\nu$, $j\in \{1,\dots,4\}$ and $i\in \{0,\dots,6\}$.
\end{lemma}
\textcolor{black}{In Lemma 3.3, $\mathbb{E}_p[\cdot]$ denotes the expectation under the true data-generating process, evaluated iteratively as
$\mathbb{E}_p[\cdot] = \mathbb{E}_{X} \Bigl[ \mathbb{E}_{Y \mid X} \bigl[ \cdot \mid X\bigr] \Bigr]
$ (i.e., using tower rule \cite{hyndman1996estimating}).
Expectations involving only $X$ are taken with respect to its marginal density $f_{X}$, while expectations involving both $X$ and $Y$ are taken with respect to the conditional distribution $f_{Y \mid X}$ followed by averaging with respect to $f_{X}$ which is assumed to be a uniform distribution over $D_X$.
}

\textcolor{black}{We next establish concentration bounds for the kernel-based estimators constructed from the sampled data. 
To this end, we consider the randomness induced by the data generation process, where $\hat{ X}_i$ are i.i.d. samples drawn according to the uniform density function $f_X$, and $\hat{ Y}_i$ are generated according to the conditional density function $f_{Y|X}$.
The uniform density function $f_X$ characterizes the data-generating process and serves as the reference distribution for the kernel-based estimators and their concentration analysis.}


Then, we can have the following results. 

\begin{lemma}\label{multi_lemma_midd}
    Let $\textcolor{black}{X=( X_{1},\ldots,X_{\mathsf d})}$ and $\textcolor{black}{Y=( Y_{1},\ldots, Y_{\mathsf d})}$ be $\mathsf d$-dimensional random variables with a CoDF $f_{Y|X}(\bs y,\bs x)$ for all $(\bs x,\bs y)\in D_{\bs X}\times D_{\bs Y}$, and $X$ is uniformly distributed over $D_{\bs X}$ with density $f_{X}(\bs x)$. Suppose that $k:\mathbb{R}\to \mathbb{R}$ is the Gaussian kernel function with bandwidth $h.$ In addition, samples $\{\textcolor{black}{\hat{ X}_{i}},~i=1,\ldots,n \}$ are selected uniformly from $D_{\bs X}$ with density function $f_{X}(\bs x)$, and for each $\textcolor{black}{\hat{ X}_{i}}$, sample $\textcolor{black}{\hat{Y}_{i}}$ is generated from $\textcolor{black}{(Y|\hat{ X}_{i})}$ with the CoDF $f_{Y|X}(\bs y,\bs x)$, $i\in \{1,\ldots,n \}$, where $\textcolor{black}{\hat{ X}_{i}=(\hat{ X}_{i1},\ldots,\hat{ X}_{i\,\mathsf d})}$ and $\textcolor{black}{\hat{ Y}_{i}=(\hat{ Y}_{i1},\ldots,\hat{ Y}_{i\,\mathsf d})}$. Then, for all $\tau_{\bs x},\tau_{\bs x,\bs y},\tau_{\bs x,d},\tau_{\bs x,\bs y, d}>0$, for any $j\in \{ 1,\ldots,\mathsf d \}$,
    we have 
\begin{align}
    &P\left(\left| \bar{f}(\bs x)-\hat{f}_{X}(\bs x) \right|\leq \epsilon_{\bs x} \right) 
    \ge 1 - 2\exp(-\tau_{\bs x}), \label{mult_mid_bernstain_1} \\
    &P\left(\left| \hat{f}_{YX}(\bs y,\bs x)\Bar{f}(\bs x)-\Bar{f}(\bs y,\bs x)\hat{f}_{X}(\bs x) \right| 
    \lesssim \textcolor{black}{\bar{\epsilon}}_{\bs x,\bs y} \right) 
    \ge 1 - 2\exp(-\tau_{\bs x,\bs y}), \label{mult_mid_bernstain_2} \\
    &P\left( \left| \frac{\partial}{\partial \textcolor{black}{x}_{j}}\bar{f}(\bs x) - \frac{\partial}{\partial \textcolor{black}{x}_{j}}\hat{f}_{X}(\bs x) \right|
    \leq \epsilon_{\bs x,d} \right) 
    \ge 1 - 2\exp(-\tau_{\bs x,d}), \text{ and }\label{mult_mid_berstain_3} \\
    &P\left( \left| \frac{\partial}{\partial\textcolor{black}{x}_{j}}\hat{f}_{YX}(\bs y,\bs x) - \frac{\partial}{\partial\textcolor{black}{x}_{j}}\bar{f}(\bs y,\bs x) \right|
    \lesssim \epsilon_{\bs x,\bs y,d} \right) 
    \ge 1 - 2\exp(-\tau_{\bs x,\bs y,d}), \label{mult_mid_berstain_4}
\end{align}
with the constants
\begin{align}
    &\epsilon_{\bs x} = \frac{\tau_{\bs x} R_{\bs x}}{3n} + \sqrt{ \frac{\tau_{\bs x}^2 R_{\bs x}^2}{9n^2} + \frac{2\tau_{\bs x}\delta_{\bs x}^2}{n} }, \label{epsi_mult_mid_bernstain_1} \\
    &\bar{\epsilon}_{\bs x,\bs y} = \frac{\tau_{\bs x,\bs y} \Bar{R}_{\bs x,\bs y}}{3n} + \sqrt{ \frac{\tau_{\bs x,\bs y}^2 \Bar{R}_{\bs x,\bs y}^2}{9n^2} + \frac{2\tau_{\bs x,\bs y}\bar{\delta}_{\bs x,\bs y}^2}{n} }, \label{epsi_mult_mid_bernstain_2} \\
    &\epsilon_{\bs x,d} = \frac{\tau_{\bs x,d} R_{\bs x,d}}{3n} + \sqrt{ \frac{\tau_{\bs x,d}^2 R_{\bs x,d}^2}{9n^2} + \frac{2\tau_{\bs x,d}\delta_{\bs x,d}^2}{n} }, \label{epsi_mult_mid_berstain_3} \\
    &\epsilon_{\bs x,\bs y,d} = \frac{\tau_{\bs x,\bs y,d} R_{\bs x,\bs y,d}}{3n} + \sqrt{ \frac{\tau_{\bs x,\bs y,d}^2 R_{\bs x,\bs y,d}^2}{9n^2} + \frac{2\tau_{\bs x,\bs y,d}\bar{\delta}_{\bs x,\bs y,d}^2}{n} }, \label{epsi_mult_mid_berstain_4}
\end{align}
where $\delta^2_{\bs x}=h^{-\mathsf d}  G_{20}(k)f_{X}(\bs x)-f^{2}_{X}(\bs x),$ $R_{\bs x}=\max_{u} h^{-\mathsf d}k^{\mathsf d}(u)+\E_{p}\bs k_{\bs x,h},$ $\bar{\delta}^2_{\bs x,\bs y}=h^{-2\mathsf d} G^2_{20}(k)f^{3}_{X}(\bs x)f_{Y|X}(\bs y,\bs x)$, 
\begin{align*}
    \bar{R}_{\bs x,\bs y}= 2\max_{u} h^{-2\mathsf d} k^{2\mathsf d}(u)\E_{p}\bs k_{\bs x,h},~ \delta^2_{\bs x,d}=h^{-(\mathsf d+2)}G_{22}(k)f_{X}(\bs x),~R_{\bs x,d}=h^{-(\mathsf d+1)} \max_{u}uk^{\mathsf d}(u),
\end{align*}    
\begin{align*}
    \Bar{\delta}^2_{\bs x,\bs y,d}=h^{-4}G_{20}(k)G_{22}(k)f_{X}(\bs x)C_{f},~R_{\bs x,\bs y,d}=h^{-2\mathsf d+1} \max_{u}u k^{2\mathsf d}(u)  +h^{-\frac{2\mathsf d+1}{2}}G^{\frac{1}{2}}_{22}(k)f^{\frac{1}{2}}_{X}(\bs x)h^{-\mathsf d}\max_{u}k^{\mathsf d}(u).
\end{align*}
\end{lemma}

\begin{proof}
     Here, we only give the proof for \textcolor{black}{inequality}~\eqref{mult_mid_bernstain_2}. The rest of the results have similar proofs. Rearranging $\hat{f}_{YX}(\bs y,\bs x)\Bar{f}(\bs x)-\Bar{f}(\bs y,\bs x)\hat{f}_{X}(\bs x)$ we can have 
    \begin{align}\label{mult_part_1}
        \hat{f}(\bs y,\bs x)\Bar{f}(\bs x)\!-\!\Bar{f}(\bs y,\bs x)\hat{f}(\bs x)\!=\!(\hat{f}(\bs y,\bs x)-\Bar{f}(\bs y,\bs x))\Bar{f}(\bs x) +\Bar{f}(\bs y,\bs x)(\Bar{f}(\bs x)-\hat{f}(\bs x)). 
    \end{align}
    We provide the probabilistic asymptotic bound for the right hand side of \textcolor{black}{equation}~\eqref{mult_part_1}.
    To simplify our discussion, we need to define a random variable in the following form 
    \begin{align*}
    g_{\bs x,\bs y}(\textcolor{black}{ X, Y}):=&\left[ \bs k_{\bs x,\bs y,h}(\textcolor{black}{ X,Y})-\E_{p} \bs k_{\bs x,\bs y,h}(\textcolor{black}{ X,Y})\right]\E_{p} \bs k_{\bs x,h}(\textcolor{black}{ X})
    +\left[ \E_{p} \bs k_{\bs x,h}(\textcolor{black}{ X})-\bs k_{\bs x,h}(\textcolor{black}{ X}) \right]\E_{p}\bs k_{\bs x,\bs y,h}(\textcolor{black}{ X,Y}).
    \end{align*}
    In fact, \textcolor{black}{equation}~\eqref{mult_part_1} can be treated as the expectation of random variable $g_{\bs x,\bs y}(\textcolor{black}{ X,Y})$ under samples $\hat{X}_{i}$ and $\hat{Y}_{i},$ $i\in\{1,\ldots,n\}$. Then, according to Bernstein inequality, we need to compute the bound and variance of $g_{\bs x,\bs y}(\textcolor{black}{ X,Y})$ for obtaining the bound of \textcolor{black}{equation}~\eqref{mult_part_1}.

First of all, we give the upper bound of $g_{\bs x,\bs y}(\textcolor{black}{\hat{X}_{i}},\textcolor{black}{\hat{Y}_{i}})$, as follows:
\begin{align*}
    |g_{\bs x,\bs y}(\textcolor{black}{\hat{X}_{i}},\textcolor{black}{\hat{Y}_{i}})|=&| \bs k_{\bs x,\bs y,h}(\textcolor{black}{\hat{X}_{i}},\textcolor{black}{\hat{Y}_{i}})\E_{p}\bs k_{\bs x,h}-\bs k_{\bs x,h}(\textcolor{black}{\hat{X}_{i}})\E_{p} \bs k_{\bs x,\bs y,h} |\\
    \leq& | \max_{(\textcolor{black}{X},\textcolor{black}{Y})\in D_X\times D_Y}\bs k_{\bs x,\bs y,h}(\textcolor{black}{X},\textcolor{black}{Y})\E_{p}\bs k_{\bs x,h}  |+| \max_{\textcolor{black}{X}\in D_X}k_{\bs x,h}\E_{p}k_{\bs x,\bs y,h}|\\
    =&\max_{(\textcolor{black}{X},\textcolor{black}{Y})\in D_X\times D_Y}h^{-2\mathsf d}\prod^{\mathsf d}_{i=1}k\left((x_{i}- X_{i})/h\right) k\left((y_{i}- Y_{i})/h\right)\\
    &+\max_{\textcolor{black}{X}\in D_X}h^{-\mathsf d}  \prod^{\mathsf d}_{i=1}k\left((x_{i}- X_{i})/h\right) \E_{p}\bs k_{\bs x,\bs y,h}\\
    =&\max_{u}  h^{-2\mathsf d} k^{2\mathsf d}(u)\E_{p}\bs k_{\bs x,h}+\max_{u}h^{-\mathsf d} k^{\mathsf d}(u){\color{black}\E_{p}\bs k_{\bs x,\bs y,h}}\\
    \leq&2\max_{u} h^{-2\mathsf d} k^{2\mathsf d}(u)\E_{p}\bs k_{\bs x,h}.
\end{align*}
Thus, the upper bound of $g_{\bs x,\bs y}(\textcolor{black}{\hat{X}_{i}},\textcolor{black}{\hat{Y}_{i}})$ is $ \Bar{R}_{\bs x,\bs y}:=2\max_{u} h^{-2\mathsf d} k^{2\mathsf d}(u)\E_{p}\bs k_{\bs x,h},$
such that for all $(\bs x,\bs y)\in \mathbb R^{2\mathsf d},$ $|g_{\bs x,\bs y}(\textcolor{black}{\hat{X}_{i}},\textcolor{black}{\hat{Y}_{i}})|\leq \Bar{R}_{\bs x,\bs y}$.

Secondly, we need to give the upper bound of the variance of $g_{\bs x,\bs y}(\textcolor{black}{\hat{X}_{i}},\textcolor{black}{\hat{Y}_{i}}).$
%
Since
\begin{align*}
    \delta^2_{\bs x,\bs y}:=&\Var_{p}[ g_{\bs x,\bs y}(\textcolor{black}{\hat{X}_{i}},\textcolor{black}{\hat{Y}_{i}})]\\
    =&\E_{p}\left[\left[ k_{\bs x,\bs y,h}(\textcolor{black}{\hat{X}_{i}},\textcolor{black}{\hat{Y}_{i}})\E_{p}k_{\bs x,h}-k_{\bs x,h}(\textcolor{black}{\hat{X}_{i}})\E_{p}k_{\bs x,\bs y,h} \right]^2\right]
    \leq\E_{p}k^2_{\bs x,\bs y,h}(\textcolor{black}{\hat{X}_{i}},\textcolor{black}{\hat{Y}_{i}})\E^2_{p}k_{\bs x,h},
\end{align*}
we can have 
$
    \delta^2_{\bs x,\bs y}\lesssim h^{-2\mathsf d} G^2_{20}(k)f^{3}_{X}(\bs x)f_{Y|X}(\bs y,\bs x)
,$ based on Assumption~\ref{asstwod_1} and Lemma.~\ref{multi_expect}.
Thus, we can obtain $\bar{\delta}^2_{\bs x,\bs y}=h^{-2\mathsf d} G^2_{20}(k)f^{3}_{X}(\bs x)f_{Y|X}(\bs y,\bs x).$

According to Bernstein inequality, 
we can obtain the confidence of the upper bound of $\hat{f}(\bs y,\bs x)\Bar{f}(\bs x)-\Bar{f}(\bs y,\bs x)\hat{f}(\bs x)$, as follows.
For any  $n$ and $\tau_{\bs x,\bs y}$, there exists $\epsilon_{\bs x,\bs y}$ such that 

\begin{align*}
    P\left( \left|  (\hat{f}(\bs y,\bs x)-\Bar{f}(\bs y,\bs x))\Bar{f}(\bs x) +\Bar{f}(\bs y,\bs x)(\Bar{f}(\bs x)-\hat{f}(\bs x)) \right|\leq \epsilon_{x,y}\right)\ge 1-2\exp(-\tau_{\bs x,\bs y}),
\end{align*}
where 
\begin{align*}
    \epsilon_{\bs x,\bs y}=&\max\left\{\frac{\tau_{\bs x,\bs y} \Bar{R}_{\bs x,\bs y}}{3n}+\sqrt{ \frac{\tau_{\bs x,\bs y}^{2} \Bar{R}_{\bs x,\bs y}^2}{9n^2}+\frac{2\tau_{\bs x,\bs y}\delta_{\bs x,\bs y}^2}{n}  }, \frac{\tau_{\bs x,\bs y} \Bar{R}_{\bs x,\bs y}}{3n}-\sqrt{ \frac{\tau_{\bs x,\bs y}^{2}\Bar{R}_{\bs x,\bs y}^2}{9n^2}+\frac{2\tau_{\bs x,\bs y}\delta_{\bs x,\bs y}^2}{n} }  \right\}\\
    =&\frac{\tau_{\bs x,\bs y} \Bar{R}_{\bs x,\bs y}}{3n}+\sqrt{ \frac{\tau_{\bs x,\bs y}^{2} \Bar{R}_{\bs x,\bs y}^2}{9n^2}+\frac{2\tau_{\bs x,\bs y}\delta_{\bs x,\bs y}^2}{n}  }.
\end{align*}
Since $\delta^2_{\bs x,\bs y}\lesssim \bar{\delta}^2_{\bs x,\bs y},$ it still holds that 
\begin{align*}
    P\left( \left|  (\hat{f}(\bs y,\bs x)-\Bar{f}(\bs y,\bs x))\Bar{f}(\bs x) +\Bar{f}(\bs y,\bs x)(\Bar{f}(\bs x)-\hat{f}(\bs x)) \right|\lesssim \bar{\epsilon}_{x,y}\right)\ge 1-2\exp(-\tau_{\bs x,\bs y}),
\end{align*}
where $\bar{\epsilon}_{\bs x,\bs y}$ is provided in \eqref{epsi_mult_mid_bernstain_2} and is obtained by replacing $\delta_{\bs x,\bs y}$ with $\bar{\delta}_{\bs x,\bs y}$ in the above expression for $\epsilon_{\bs x,\bs y}$.
\end{proof}

Note that selecting the same values $\tau_{\bs x}=\tau_{\bs x,\bs y}=\tau_{\bs x,d}=\tau_{\bs x,\bs y,d}$ will result in having the same probability for the inequalities in~\eqref{mult_mid_bernstain_1},~\eqref{mult_mid_bernstain_2},~\eqref{mult_mid_berstain_3} and \eqref{mult_mid_berstain_4}.
Also note that if Assumption~\ref{asstwod_1} is not satisfied, we can have a weaker result shown in Appendix~\ref{supp_main}, in the sense that the values of $\textcolor{black}{\bar{\epsilon}}^{\textcolor{black}{c}}_{\bs x,\bs y}$ and $\epsilon^{\textcolor{black}{c}}_{\bs x,\bs y,d}$ are larger than $\textcolor{black}{\bar{\epsilon}}_{\bs x,\bs y}$ and $\epsilon_{\bs x,\bs y,d}$ in \textcolor{black}{inequality}~\eqref{mult_mid_bernstain_2} and \eqref{mult_mid_berstain_4}, respectively. 


\subsection{Asymptotic Convergence of Transition Probabilities and the Estimated LC}

Next we formulate a confidence for the asymptotic upper bound on the discrepancy
between the true transition probabilities and their non-parametric estimates.
\textcolor{black}{For the following analysis,  we need to assume that, for every fixed $\bs x\in D_X$ and every compact set $A\subseteq D_Y$,
\begin{equation}
\label{eq:lip_residue}
    R_n(\bs x,\bs y) = \left|
\bigl(\hat f_{YX}(\bs y,\bs x)-\bar f(\bs y,\bs x)\bigr)
\bar f(\bs x)
+
\bar f(\bs y,\bs x)
\bigl(\bar f(\bs x)-\hat f_X(\bs x)\bigr)
\right|
\end{equation}
is Lipschitz continuous with respect to $\bs y$ on $A$, with a deterministic Lipschitz
constant independent of $n$. 
This condition is mild in the present setting.
Under Assumption~\ref{asstwod_1}, the underlying density functions have the required smoothness, while $\hat f_{YX}(\bs y,\bs x)$ is constructed using a sufficiently smooth kernel, such as the Gaussian kernel considered below.
Hence, on every compact set $A$, the Lipschitz continuity of
$R_n(\bs x,\bs y)$ with respect to $\bs y$ is natural.
This condition enables a uniform concentration bound over
$\bs y\in A$ via a finite-covering argument; see
Proposition~\ref{prop:uniform_y_transition}.}


\begin{theorem}\label{multi_Bound_TranProbi}
   Suppose $\textcolor{black}{ X=( X_{1},\ldots, X_{\mathsf d})}$ and $\textcolor{black}{Y=( Y_{1},\ldots, Y_{\mathsf d})}$ are $\mathsf d$-dimensional random variables, which have a CoDF $f_{Y|X}(\bs y,\bs x)$ for all $(\bs x,\bs y)\in D_{\bs X}\times D_{\bs Y}$, and $X$ is uniformly distributed over $D_{\bs X}$ with density $f_{X}(\bs x)$. Suppose that $k:\mathbb{R}\to \mathbb{R}$ is the Gaussian kernel function with bandwidth $ h$.
   In addition, samples $\{\textcolor{black}{\hat{ X}_{i}},~i=1,\ldots,n \}$ are selected uniformly from $D_{\bs X}$ with density function $f_{X}(\bs x)$, and for each $\textcolor{black}{\hat{ X}_{i}}$, sample $\textcolor{black}{\hat{ Y}_{i}}$ is generated from $(\textcolor{black}{Y}|\textcolor{black}{\hat{ X}_{i}})$ with the CoDF $f_{Y|X}(\bs y,\bs x)$, $i\in \{1,\ldots,n \}$, where $\textcolor{black}{\hat{ X}_{i}=(\hat{ X}_{i1},\ldots,\hat{ X}_{i\,\mathsf d})}$ and $\textcolor{black}{\hat{ Y}_{i}=(\hat{ Y}_{i1},\ldots,\hat{ Y}_{i\,\mathsf d})}$.
   \textcolor{black}{Suppose further that the Lipschitz condition on
$R_n(\bs x,\bs y)$ in \eqref{eq:lip_residue} holds.}
   Then, for all $\tau_{p}>0$, any $ \bs x \in D_{\bs X}$, and any set $A\subseteq D_{\bs Y},$
   we have 
\begin{align}\label{ineq_multi_Bound_TranProbi}
\left| \int_{A} \hat{f}_{Y|X}(\bs y,\bs x)d\bs y \!-\! \int_{A} f_{Y|X}(\bs y,\bs x)d\bs y  \right|\!\lesssim\! \epsilon_{p} ,
\end{align}
with a probability at least $1-2\exp{(-\tau_{p})},$
where 
\begin{align*}
    \epsilon_{p}=&\int_{A}  \frac{1}{\hat{f}(\bs x)\Bar{f}(\bs x)} \left[ \frac{\tau_{\bs x,\bs y} \Bar{R}_{\bs x,\bs y}}{3n}+\sqrt{ \frac{\tau_{\bs x,\bs y}^{2} \Bar{R}_{\bs x,\bs y}^2}{9n^2}+\frac{2\tau_{\bs x,\bs y}\bar{\delta}_{\bs x,\bs y}^2}{n}  }\right] d\bs y\\
    &+\int_{A}  \bigg\{ \frac{1}{2} h^{2}G_{12}(k)\bigg| \sum^{\mathsf d}_{i=1}\frac{\partial^{2}}{\partial y_{i}^{2}}f_{Y|X}(\bs y,\bs x)
    +\sum^{\mathsf d}_{i=1}\frac{\partial^{2}}{\partial x_{i}^{2}}f_{Y|X}(\bs y,\bs x)  \bigg|+O(h^{4}) \bigg\} dy,
\end{align*}
with
$ \bar{R}_{\bs x,\bs y}= 2\max_{u} h^{-2\mathsf d} k^{2\mathsf d}(u)\E_{p}\bs k_{\bs x,h}$ and $\bar{\delta}^2_{\bs x,\bs y}=h^{-2\mathsf d} G^2_{20}(k)f^{3}_{X}(\bs x)f_{Y|X}(\bs y,\bs x).$
\end{theorem}

\begin{proof}
   Since 
   \begin{align*}
    \hat{f}_{Y|X}(\bs y,\bs x)-f_{Y|X}(\bs y,\bs x)
    &= \frac{\hat{f}_{YX}(\bs y,\bs x)}{\hat{f}_{X}(\bs x)}-\frac{\Bar{f}(\bs y,\bs x)}{\hat{f}_{X}(\bs x)}+\frac{\Bar{f}(\bs y,\bs x)}{\hat{f}_{X}(\bs x)}-\frac{\Bar{f}(\bs y,\bs x)}{\Bar{f}(\bs x)}+\frac{\Bar{f}(\bs y,\bs x)}{\Bar{f}(\bs x)}-f_{Y|X}(\bs y,\bs x)\\
    &=\frac{1}{\hat{f}_{X}(\bs x)\Bar{f}(\bs x)}\left[  (\hat{f}_{YX}(\bs y,\bs x)-\Bar{f}(\bs y,\bs x))\Bar{f}(\bs x) +\Bar{f}(\bs y,\bs x)(\Bar{f}(\bs x)-\hat{f}_{X}(\bs x))\right]
    +\left( \frac{\Bar{f}(\bs y,\bs x)}{\Bar{f}(\bs x)}-f_{Y|X}(\bs y,\bs x) \right),
\end{align*}
   we can obtain
   \begin{align}\label{main_ineq}
       D(\bs x):=&\int_{A} \hat{f}_{Y|X}(\bs y,\bs x)d\bs y - \int_{A} f_{Y|X}(\bs y,\bs x)d\bs y
    =\int_{A} \hat{f}_{Y|X}(\bs y,\bs x)-f_{Y|X}(\bs y,\bs x) d\bs y\nonumber\\ 
    =&\int_{A}\Bigg[\frac{1}{\hat{f}_{X}(\bs x)\Bar{f}(\bs x)}\left[  (\hat{f}_{YX}(\bs y,\bs x)-\Bar{f}(\bs y,\bs x))\Bar{f}(\bs x) +\Bar{f}(\bs y,\bs x)(\Bar{f}(\bs x)-\hat{f}_{X}(\bs x))\right]
    +\left( \frac{\Bar{f}(\bs y,\bs x)}{\Bar{f}(\bs x)}-f_{Y|X}(\bs y,\bs x) \right) \Bigg] d\bs y\nonumber\\
    \leq&\int_{A} \frac{1}{\hat{f}_{X}(\bs x)\Bar{f}(\bs x)}\left|  (\hat{f}_{YX}(\bs y,\bs x)-\Bar{f}(\bs y,\bs x))\Bar{f}(\bs x) +\Bar{f}(\bs y,\bs x)(\Bar{f}(\bs x)-\hat{f}_{X}(\bs x))\right| d\bs y
    +\int_{A} \left| \frac{\Bar{f}(\bs y,\bs x)}{\Bar{f}(\bs x)}-f_{Y|X}(\bs y,\bs x) \right|d\bs y.
   \end{align}

   Thus, our next tasks include calculating the upper bounds for $ \left| \frac{\Bar{f}(\bs y,\bs x)}{\Bar{f}(\bs x)}-f_{Y|X}(\bs y,\bs x) \right|$ and 
   $$\frac{1}{\hat{f}_{X}(\bs x)\Bar{f}(\bs x)}\left|  (\hat{f}_{YX}(\bs y,\bs x)-\Bar{f}(\bs y,\bs x))\Bar{f}(\bs x) +\Bar{f}(\bs y,\bs x)(\Bar{f}(\bs x)-\hat{f}_{X}(\bs x))\right|.$$

   For $\left| \frac{\Bar{f}(\bs y,\bs x)}{\Bar{f}(\bs x)}-f_{Y|X}(\bs y,\bs x)\right|,$
   it is easy to obtain
   \begin{align}\label{multi_sec_term}
     &\frac{\Bar{f}(\bs y,\bs x)}{\Bar{f}(\bs x)}-f_{Y|X}(\bs y,\bs x)
     =\frac{1}{2} h^2 G_{12}(K)\bigg[ \sum^{\mathsf d}_{i=1}\frac{\partial^{2}}{\partial y_{i}^{2}}f_{Y|X}(\bs y,\bs x)
    +\sum^{\mathsf d}_{i=1}\frac{\partial^{2}}{\partial x_{i}^{2}}f_{Y|X}(\bs y,\bs x) \bigg]
    +O(h^{4}),
\end{align}
since $\frac{\Bar{f}(\bs y,\bs x)}{\Bar{f}(\bs x)}=\frac{\E \left[  k_{\bs x,h}(\textcolor{black}{X})k_{\bs y,h}(\textcolor{black}{Y}) \right]}{\E\left[ k_{\bs x,h}(\textcolor{black}{X}) \right]}$, where $\E \left[  k_{\bs x,h}(\textcolor{black}{X})k_{\bs y,h}(\textcolor{black}{Y}) \right]$ and $\E\left[ k_{\bs x,h}(\textcolor{black}{X}) \right]$ can be found in Lemma \ref{multi_expect}.

For the upper bound of $\frac{1}{\hat{f}_{X}(\bs x)\Bar{f}(\bs x)}\left|  (\hat{f}_{YX}(\bs y,\bs x)-\Bar{f}(\bs y,\bs x))\Bar{f}(\bs x) +\Bar{f}(\bs y,\bs x)(\Bar{f}(\bs x)-\hat{f}_{X}(\bs x))\right|,$ based on Lemma \ref{multi_lemma_midd}, we can have that for any $\tau_{\bs x,\bs y}>0$,
\begin{align}\label{multi_confi_MiddSecond}
    &P\left( \frac{1}{\hat{f}_{X}(\bs x)\Bar{f}(\bs x)} \left|  (\hat{f}_{YX}(\bs y,\bs x)-\Bar{f}(\bs y,\bs x))\Bar{f}(\bs x) +\Bar{f}(\bs y,\bs x)(\Bar{f}(\bs x)-\hat{f}_{X}(\bs x))\right|\lesssim \frac{1}{\hat{f}_{X}(\bs x)\Bar{f}(\bs x)}\bar{\epsilon}_{\bs x,\bs y}   \right)
    \ge 1-2\exp{(-\tau_{\bs x,\bs y})},
\end{align}
where 
$\bar{\epsilon}_{\bs x,\bs y}=\frac{\tau_{\bs x,\bs y} \Bar{R}_{\bs x,\bs y}}{3n}+\sqrt{ \frac{\tau_{\bs x,\bs y}^{2} \Bar{R}_{\bs x,\bs y}^2}{9n^2}+\frac{2\tau_{\bs x,\bs y}\bar{\delta}_{\bs x,\bs y}^2}{n}  },$ $ \bar{R}_{\bs x,\bs y}= 2\max_{u} h^{-2\mathsf d} k^{2\mathsf d}(u)\E_{p}\bs k_{\bs x,h},$ and 
$$\bar{\delta}^2_{\bs x,\bs y}=h^{-2\mathsf d} G^2_{20}(k)f^{3}_{X}(\bs x)f_{Y|X}(\bs y,\bs x).$$

\textcolor{black}{By Proposition~\ref{prop:uniform_y_transition},
the pointwise concentration inequality
\eqref{multi_confi_MiddSecond}
holds uniformly over
$\bs y\in A$
with probability at least
$1-2\exp(-\tau_p)$.}
By substituting the results from \textcolor{black}{equation}~\eqref{multi_sec_term} and \textcolor{black}{inequality}~\eqref{multi_confi_MiddSecond} into the inequality \eqref{main_ineq}, we can have 
\begin{align*}
    D(\bs x)\lesssim & \int_{A}  \frac{1}{\hat{f}(\bs x)\Bar{f}(\bs x)} \left[ \frac{\tau_{\bs x,\bs y} \Bar{R}_{\bs x,\bs y}}{3n}+\sqrt{ \frac{\tau_{\bs x,\bs y}^{2} \Bar{R}_{\bs x,\bs y}^2}{9n^2}+\frac{2\tau_{\bs x,\bs y}\bar{\delta}_{\bs x,\bs y}^2}{n}  }\right] d\bs y\\
    &+\int_{A}  \bigg\{ \frac{1}{2} h^{2}G_{12}(k)\bigg| \sum^{\mathsf d}_{i=1}\frac{\partial^{2}}{\partial y_{i}^{2}}f_{Y|X}(\bs y,\bs x)
    +\sum^{\mathsf d}_{i=1}\frac{\partial^{2}}{\partial x_{i}^{2}}f_{Y|X}(\bs y,\bs x)  \bigg|+O(h^{4}) \bigg\} dy
\end{align*}
with a probability at least $1-2\exp{(-\tau_{p})}$, where $\tau_{p}=\tau_{\bs x,\bs y}$. 
\end{proof}

\begin{remark}
    The above theorem can provide a confidence for the asymptotic upper bound on the discrepancy between the true transition probabilities and their non-parametric estimates, provided that coarse bounds on the higher-order derivatives of the conditional density function are available.
    In addition, if $h$ is the order of $O(n^{-1/(6+2\mathsf d)})$ to guarantee the smoothness of the transition kernel (see the line of reasoning in Section~\ref{b_lc_estima} below), then the convergence rate of $\epsilon_{p}$ is $O(n^{-1/(3+ \mathsf d)})$.
\end{remark}

\textcolor{black}{
For the running example in \eqref{run_examp}, Theorem~\ref{multi_Bound_TranProbi} shows that the estimated probability of transitioning to any set $A$ from any state $\bs{x}$, obtained via the non-parametric estimator, converges asymptotically to the corresponding transition probability in the original system. This result is useful when constructing an IMDP abstraction for this system.
}

\textcolor{black}{
\begin{remark}
\label{rem:uniform_transition_probability}
Suppose that the transition-probability estimation error $
\hat P_{i j,a}(\bs x)
-
P_{i j,a}(\bs x)$
is uniformly Lipschitz continuous on every compact abstract state
$q_i$, with a Lipschitz constant $L_{\mathrm{tr}}$ independent of
$n$.
This condition is mild in the present setting.
Under Assumption \ref{asstwod_1}, the true transition probability
$P_{i j,a}(\bs x)$ is Lipschitz continuous.
Moreover, $\hat P_{i j,a}(\bs x)$ is constructed using a Lipschitz continuous kernel (e.g., the Gaussian or Epanechnikov kernel), which naturally yields a Lipschitz continuous estimator on every compact partition.
For any $\eta>0$, let
$\mathcal N_\eta$
be a finite $\eta$-cover of each compact abstract state
$q_i$, and define 
$\bar{\epsilon}_{p}\!\left(\tau_p+O(\log n)\right)
:=
\max_{1\leq j\leq n_{Q}}\max_{
\bs x\in\mathcal N_\eta
}
\epsilon_{p,j} (\bs x,\tau_p+O(\log n)),$
where $\epsilon_{p,j}$ denotes the pointwise error bound for the
transition probability from $q_i$ to $q_j$ in
Theorem~\ref{multi_Bound_TranProbi} evaluated at $\bs x\in q_i$, and the
$O(\log n)$ term is introduced by the covering number through the
standard union-bound argument.
Then, the pointwise bound in
Theorem~\ref{multi_Bound_TranProbi} extends simultaneously over all
abstract state $q_i$ as
\begin{align*}
\max_{\substack{
1\leq j\leq n_Q
}}
\sup_{\bs x\in q_i}
\left|
\hat P_{i j,a}(\bs x)
-
P_{i j,a}(\bs x)
\right|
\lesssim
\bar\epsilon_p\bigl(\tau_p+O(\log n)\bigr)
\end{align*}
with probability at least $1-2\exp(-\tau_p)$.
The result follows by applying Theorem~\ref{multi_Bound_TranProbi}
on a finite cover of each compact abstract state and using a union
bound. A detailed statement and proof are provided in Proposition \ref{prop:uniform_transition_probability} in 
Appendix~\ref{supp_main}.
For ease of notation, throughout the remainder of the paper, we write
$\bar{\epsilon}*_{p}
:=
\bar{\epsilon}_{p}
\!\left(\tau_p+O(\log n)\right).$
\end{remark}
}

Next theorem provides a confidence for the asymptotic upper bound on the discrepancy
between partial derivatives of the CoDF and their non-parametric estimates, which are needed in the analysis of LC estimation.

\begin{theorem}\label{multi_Bound_FirstDeriva}
Suppose $\textcolor{black}{ X=( X_{1},\ldots, X_{\mathsf d})}$ and $\textcolor{black}{ Y=( Y_{1},\ldots, Y_{\mathsf d})}$ are $\mathsf d$-dimensional random variables, which have a CoDF $f_{Y|X}(\bs y,\bs x)$ for all $(\bs x,\bs y)\in D_{\bs X}\times D_{\bs Y}$, and $X$ is uniformly distributed over $D_{\bs X}$ with density $f_{X}(\bs x)$.
Suppose that $k:\mathbb{R}\to \mathbb{R}$ is the Gaussian kernel function with bandwidth $ h$.
In addition, samples $\{\textcolor{black}{\hat{ X}_{i}},~i=1,\ldots,n \}$ are selected uniformly from $D_{\bs X}$ with density function $f_{X}(\bs x)$, and for each $\textcolor{black}{\hat{ X}_{i}}$, sample $\textcolor{black}{\hat{Y}_{i}}$ is generated from $(\textcolor{black}{ Y|\hat{ X}_{i}})$ with the CoDF $f_{Y|X}(\bs y,\bs x)$, $i\in \{1,\ldots,n \}$, where $\textcolor{black}{\hat{ X}_{i}=(\hat{ X}_{i1},\ldots,\hat{ X}_{i\,\mathsf d})}$ and $\textcolor{black}{\hat{ Y}_{i}=(\hat{ Y}_{i1},\ldots,\hat{ Y}_{i\,\mathsf d})}$.
Then, for all $\tau_{D}>0$ and $ (\bs x,\bs y)\in D_{\bs X}\times D_{\bs Y}$,
   we have 
   \begin{align*}
   &\left| \left|\frac{\partial}{\partial \textcolor{black}{x}_{j}}\left[\hat{f}_{YX}(\bs y,\bs x)\hat{f}_{X}^{-1}(\bs x) \right]\right|-\left|\frac{\partial}{\partial \textcolor{black}{x}_{j}}f_{Y|X}(\bs y,\bs x)\right|\right|\lesssim \epsilon_d, \text{with a probability at least $1-2\exp(-\tau_{D}),$ \textcolor{black}{$\forall j\in\{1,\ldots,\mathsf d\}$} }
   \end{align*}
   where 
   \begin{align}\label{pac_bound_lc}
       \epsilon_d= &\hat{f}_{X}^{-1}(\bs x)\epsilon_{\bs x,\bs y,d}+\hat{f}_{YX}(\bs y,\bs x)\hat{f}_{X}^{-2}(\bs x)\epsilon_{\bs x,d}+\hat{f}_{X}^{-1}(\bs x)\bar{f}^{-1}(\bs x)\epsilon_{\bs x}\nonumber\\
    &+\frac{1}{2} h^2 G_{12}(k)\bigg| \sum^{\mathsf d}_{i=1}\frac{\partial^{3}}{\partial\textcolor{black}{x}_{j} \partial\textcolor{black}{y}_{i}^{2}}f_{Y|X}(\bs y,\bs x)
    +\sum^{\mathsf d}_{i=1}\frac{\partial^{3}}{\partial\textcolor{black}{x}_{i}^{2}\partial \textcolor{black}{x}_{j}}f_{Y|X}(\bs y,\bs x) \bigg|.
   \end{align}
\end{theorem}
\begin{proof}
Through the product rule of \textcolor{black}{differentiation}, we can write  
    \begin{align}\label{differece_derivative}
    \left|\frac{\partial}{\partial\textcolor{black}{x}_{j}}\left[\hat{f}_{YX}(\bs y,\bs x)\hat{f}_{X}^{-1}(\bs x) \right]-\frac{\partial}{\partial\textcolor{black}{x}_{j}}f_{Y|X}(\bs y,\bs x) \right|
    =&\bigg| \hat{f}_{X}^{-1}(\bs x)\frac{\partial}{\partial\textcolor{black}{x}_{j}}\hat{f}_{YX}(\bs y,\bs x)-\hat{f}_{YX}(\bs y,\bs x)\hat{f}_{X}^{-2}(\bs x)\frac{\partial}{\partial\textcolor{black}{x}_{j}}\hat{f}_{X}(\bs x)\nonumber\\
    &-\frac{\partial}{\partial\textcolor{black}{x}_{j}}f_{Y|X}(\bs y,\bs x) \bigg|.
    \end{align}
    Next we do rearrangements for $\hat{f}_{X}^{-1}(\bs x)\frac{\partial}{\partial\textcolor{black}{x}_{j}}\hat{f}_{YX}(\bs y,\bs x)-\hat{f}_{YX}(\bs y,\bs x)\hat{f}_{X}^{-2}(\bs x)\frac{\partial}{\partial\textcolor{black}{x}_{j}}\hat{f}_{X}(\bs x)$, as follows:
\begin{align}\label{multi_ineq_LC_firsTer}
    &\hat{f}_{X}^{-1}(\bs x)\frac{\partial}{\partial\textcolor{black}{x}_{j}}\hat{f}_{YX}(\bs y ,\bs x)-\hat{f}_{YX}(\bs y,\bs x)\hat{f}_{X}^{-2}(\bs x)\frac{\partial}{\partial\textcolor{black}{x}_{j}}\hat{f}_{X}(\bs x)\nonumber\\
=&\hat{f}_{X}^{-1}(\bs x)\left[ \frac{\partial}{\partial\textcolor{black}{x}_{j}}\hat{f}_{YX}(\bs y,\bs x)-\frac{\partial}{\partial\textcolor{black}{x}_{j}}\bar{f}(\bs y,\bs x)  \right]+\hat{f}_{YX}(\bs y,\bs x)\hat{f}_{X}^{-2}(\bs x)\left[  \frac{\partial}{\partial\textcolor{black}{x}_{j}}\bar{f}(\bs x)-\frac{\partial}{\partial\textcolor{black}{x}_{j}}\hat{f}_{X}(\bs x) \right]\nonumber\\
&+\left[ \hat{f}_{X}^{-1}(\bs x)-\bar{f}^{-1}(\bs x)  \right]\frac{\partial}{\partial\textcolor{black}{x}_{j}}\bar{f}(\bs y,\bs x)+\left[ \Bar{f}(\bs y,\bs x)-\hat{f}_{YX}(\bs y,\bs x) \right]\hat{f}_{X}^{-2}(\bs x)\frac{\partial}{\partial\textcolor{black}{x}_{j}}\bar{f}(\bs x)\nonumber\\
&+\Bar{f}^{-1}(\bs x)\frac{\partial}{\partial\textcolor{black}{x}_{j}}\bar{f}(\bs y,\bs x).
\end{align}

Since $f_{X}(\bs x)$ is the density function of the uniform distribution, i.e., $\frac{\partial}{\partial\textcolor{black}{x}_{j}}\bar{f}(\bs x)=0,$  
we have 
\begin{align}\label{multi_Deriv_MainIneq}
    &\left| \frac{\partial}{\partial\textcolor{black}{x}_{j}}\left[\hat{f}_{YX}(\bs y,\bs x)\hat{f}_{X}^{-1}(\bs x) \right]-\frac{\partial}{\partial\textcolor{black}{x}_{j}}f_{Y|X}(\bs y,\bs x)\right|\nonumber\\
    \leq&\hat{f}_{X}^{-1}(\bs x)\left| \frac{\partial}{\partial\textcolor{black}{x}_{j}}\hat{f}_{YX}(\bs y,\bs x)-\frac{\partial}{\partial\textcolor{black}{x}_{j}}\bar{f}(\bs y,\bs x)  \right|+\hat{f}_{YX}(\bs y,\bs x)\hat{f}_{X}^{-2}(\bs x)\left|  \frac{\partial}{\partial\textcolor{black}{x}_{j}}\bar{f}(\bs x)-\frac{\partial}{\partial\textcolor{black}{x}_{j}}\hat{f}_{X}(\bs x) \right|\nonumber\\
    &+ \hat{f}_{X}^{-1}(\bs x)\bar{f}^{-1}(\bs x)\left| \hat{f}_{X}(\bs x)-\bar{f}(\bs x)  \right|\frac{\partial}{\partial\textcolor{black}{x}_{j}}\bar{f}(\bs y,\bs x)+\left|\Bar{f}^{-1}(\bs x)\frac{\partial}{\partial\textcolor{black}{x}_{j}}\bar{f}(\bs y,\bs x)-\frac{\partial}{\partial\textcolor{black}{x}_{j}}f_{Y|X}(\bs y,\bs x)\right|.
\end{align}

According to Lemma \ref{multi_expect}, we can have 
\begin{align}\label{multi_deriv_inequ}
    &\Bar{f}^{-1}(\bs x)\frac{\partial}{\partial\textcolor{black}{x}_{j}}\bar{f}(\bs y,\bs x)-\frac{\partial}{\partial\textcolor{black}{x}_{j}}f_{Y|X}(\bs y,\bs x)\nonumber\\
    =&\frac{1}{2} h^2 G_{12}(k)\bigg[ \sum^{\mathsf d}_{i=1}\frac{\partial^{3}}{\partial\textcolor{black}{x}_{j} \partial\textcolor{black}{y}_{i}^{2}}f_{Y|X}(\bs y,\bs x)
    +\sum^{\mathsf d}_{i=1}\frac{\partial^{3}}{\partial\textcolor{black}{x}_{i}^{2}\partial \textcolor{black}{x}_{j}}f_{Y|X}(\bs y,\bs x) \bigg]
    +O(h^{4}).
\end{align}

By plugging \textcolor{black}{equation}~\eqref{multi_deriv_inequ} into \textcolor{black}{inequality}~\eqref{multi_Deriv_MainIneq} and applying Lemma \ref{multi_lemma_midd}, for $\tau_{D}:=\tau_{\bs x}=\tau_{\bs x,d}=\tau_{\bs x,\bs y,d},$ we can deduce that 
\begin{align*}
    &\left| \frac{\partial}{\partial\textcolor{black}{x}_{j}}\left[\hat{f}_{YX}(\bs y,\bs x)\hat{f}_{X}^{-1}(x) \right]-\frac{\partial}{\partial\textcolor{black}{x}_{j}}f_{Y|X}(\bs y,\bs x)\right|\\
    \lesssim & \hat{f}_{X}^{-1}(\bs x)\epsilon_{\bs x,\bs y,d}+\hat{f}_{YX}(\bs y,\bs x)\hat{f}_{X}^{-2}(\bs x)\epsilon_{\bs x,d}+\hat{f}_{X}^{-1}(\bs x)\bar{f}^{-1}(\bs x)\epsilon_{\bs x}\\
    &+\frac{1}{2} h^2 G_{12}(k)\bigg| \sum^{\mathsf d}_{i=1}\frac{\partial^{3}}{\partial\textcolor{black}{x}_{j} \textcolor{black}{y}_{i}^{2}}f_{Y|X}(\bs y,\bs x)
    +\sum^{\mathsf d}_{i=1}\frac{\partial^{3}}{\partial\textcolor{black}{x}_{i}^{2}\partial \textcolor{black}{x}_{j}}f_{Y|X}(\bs y,\bs x) \bigg|,
\end{align*}
with a probability at least $1-2\exp(-\tau_{D}).$

\end{proof}

\begin{remark}
    If $h=n^{-1/(6+2\mathsf d)}$, the convergence rate of $\epsilon_{d}$ can be $O(n^{-1/(3+ \mathsf d)})$.
\end{remark}

\textcolor{black}{
\begin{remark}
\label{rem:uniform_first_derivative}
Suppose that the first-derivative estimation error is uniformly
Lipschitz continuous on the compact domain
$D_{\bs X}\times D_{\bs Y}$, with a Lipschitz constant independent of
$n$.
This assumption is also mild. Similar to
Remark~\ref{rem:uniform_transition_probability},
Assumption~\ref{asstwod_1} implies the Lipschitz continuity of
$\frac{\partial}{\partial x_j}f_{Y|X}(\bs y,\bs x)$, while
$\frac{\partial}{\partial x_j}
\!\left[
\hat f_{YX}(\bs y,\bs x)\hat f_X^{-1}(\bs x)
\right]$
is induced by a sufficiently smooth kernel estimator.
For any $\delta>0$, let
$\mathcal N_\delta$
be a finite $\delta$-cover of
$D_{\bs X}\times D_{\bs Y}$, and let $\tau_\delta
:=
\tau_D+\log(\mathsf dN_\delta),$
where $N_\delta=|\mathcal N_\delta|$.
Define $\bar\epsilon_d(\tau_\delta)
:=
\max_{\substack{
1\le j\le\mathsf d\\
(\bs x,\bs y)\in\mathcal N_\delta
}}
\epsilon_{d}(\bs x,\bs y,j;\tau_\delta),$
where $\epsilon_{d}$ denotes the bound given by
\eqref{pac_bound_lc}
evaluated at
$(\bs x,\bs y)$
for the $j$-th derivative.
Then, Theorem~\ref{multi_Bound_FirstDeriva} admits the
uniform extension
\begin{align*}
    \max_{1\le j\le \mathsf d}
\sup_{(\bs x,\bs y)\in D_{\bs X}\times D_{\bs Y}}
\left|
\left|\frac{\partial}{\partial x_j}
\!\left[
\hat f_{YX}(\bs y,\bs x)\hat f_X^{-1}(\bs x)
\right]\right|
-
\left|\frac{\partial}{\partial x_j}f_{Y|X}(\bs y,\bs x)\right|
\right|
\lesssim
\bar\epsilon_d(\tau_D+O(\log n)),
\end{align*}
with probability at least
$1-2\exp(-\tau_D)$.
The proof follows the same covering and union-bound argument as
Proposition~\ref{prop:uniform_transition_probability} and is omitted.
For ease of notation, throughout the remainder of the paper, we
write $\bar{\epsilon}^*_d
:=
\bar{\epsilon}_d(\tau_D+O(\log n)).$
\end{remark}
}

A corollary of the above theorem is a bound on the discrepancy between the LC of the CoDF and that of its non-parametric estimate, as stated next. 

\begin{corollary}
\label{LC_confide_estim}
Let $\textcolor{black}{ X=( X_{1},\ldots, X_{\mathsf d})}$ and $\textcolor{black}{ Y=( Y_{1},\ldots, Y_{\mathsf d})}$ \textcolor{black}{be} $\mathsf d$-dimensional random variables, which have a CoDF $f_{Y|X}(\bs y,\bs x)$ for all $(\bs x,\bs y)\in D_{\bs X}\times D_{\bs Y}$, and $X$ is uniformly distributed over $D_{\bs X}$ with density $f_{X}(\bs x)$.
Suppose that $k:\mathbb{R}\to \mathbb{R}$ is the Gaussian kernel function with bandwidth $h$.
In addition, samples $\{\textcolor{black}{\hat{X}_{i}},~i=1,\ldots,n \}$ are selected uniformly from $D_{\bs X}$ with density function $f_{X}(\bs x)$, and for each $\textcolor{black}{\hat{X}_{i}}$, sample $\textcolor{black}{\hat{Y}_{i}}$ is generated from $(\textcolor{black}{Y|\hat{ X}_{i}})$ with the CoDF $f_{Y|X}(\bs y,\bs x)$, $i\in \{1,\ldots,n \}$, where $\textcolor{black}{\hat{ X}_{i}=(\hat{ X}_{i1},\ldots,\hat{ X}_{i\,\mathsf d})}$ and $\textcolor{black}{\hat{ Y}_{i}=(\hat{ Y}_{i1},\ldots,\hat{ Y}_{i\,\mathsf d})}$.
\textcolor{black}{Suppose further that the Lipschitz condition in
Remark~\ref{rem:uniform_first_derivative} holds.}
Then, for any $\tau_{\textcolor{black}{D}}$, there exists $\textcolor{black}{\bar{\epsilon^*}_d}$ such that
\begin{align*}\label{uniform_tranPro_error_first_der}
    \left|\hat{L}-L\right|\lesssim \textcolor{black}{\bar{\epsilon}^*_d}, \text{ with a probability at least $1-2\exp{(-\tau_{{\textcolor{black}{D}}})}$,}
\end{align*}
where $\hat{L}=\max\{ \hat{L}_{1},\ldots,\hat{L}_{\mathsf d} \}$ and $L=\max\{ L_{1},\ldots,L_{\mathsf d} \}$. Assuming that $\hat{L}_{i}=\max_{(\bs x,\bs y)\in D_{\bs X}\times D_{\bs Y}}\left| \frac{\partial}{\partial\textcolor{black}{x}_{i}}\hat{f}_{Y|X}(\bs y,\bs x) \right|$ and 
$L_{i}=\max_{(\bs x,\bs y)\in D_{\bs X}\times D_{\bs Y}}\left| \frac{\partial}{\partial\textcolor{black}{x}_{i}}f_{Y|X}(\bs y,\bs x) \right|,$ $ i\in \{ 1,\ldots,\mathsf d \}.$
\end{corollary}
\begin{proof}
    It is known that there exist $(\bs x^{*},\bs y^{*}), (\bar{\bs x},\bar{\bs y})\in D_{\bs X}\times D_{\bs Y}$ and $i,j\in \{1,\ldots,\mathsf d\}$ such that $\hat{L}=\left| \frac{\partial}{\partial\textcolor{black}{x}_{i}}\hat{f}_{Y|X}(\bs y^{*},\bs x^{*}) \right|$ and $L=\left| \frac{\partial}{\partial\textcolor{black}{x}_{j}}f_{Y|X}(\bar{\bs y},\bar{\bs x}) \right|.$ 
    According to Theorem \ref{multi_Bound_FirstDeriva} and \textcolor{black}{Remark \ref{rem:uniform_first_derivative}}, there exists a $\textcolor{black}{\bar{\epsilon}^*_d}>0$ and $n^{*}$ such that $\Big| \frac{\partial}{\partial\textcolor{black}{x}_{i}}\hat{f}_{Y|X}(\bs y^{*},\bs x^{*})-\frac{\partial}{\partial\textcolor{black}{x}_{i}}f_{Y|X}(\bs y^{*},\bs x^{*})\Big|\lesssim \textcolor{black}{\bar{\epsilon}^*_d},$ with a probability at least $1-2\exp(-\tau_{\textcolor{black}{D}}).$
    If $\hat{L}-L\ge 0,$ we can have 
    \begin{align*}
        &\left| \frac{\partial}{\partial\textcolor{black}{x}_{i}}\hat{f}_{Y|X}(\bs y^{*},\bs x^{*}) \right|-\left| \frac{\partial}{\partial\textcolor{black}{x}_{j}}f_{Y|X}(\bar{\bs y},\bar{\bs x}) \right|\\
        =&\left| \frac{\partial}{\partial\textcolor{black}{x}_{i}}\hat{f}_{Y|X}(\bs y^{*},\bs x^{*}) \right|-\left| \frac{\partial}{\partial\textcolor{black}{x}_{i}}f_{Y|X}(\bs y^{*},\bs x^{*}) \right|+\left| \frac{\partial}{\partial\textcolor{black}{x}_{i}}f_{Y|X}(\bs y^{*},\bs x^{*}) \right| -\left| \frac{\partial}{\partial\textcolor{black}{x}_{j}}f_{Y|X}(\bar{\bs y},\bar{\bs x}) \right|,
    \end{align*}
    implying that $\left| \frac{\partial}{\partial\textcolor{black}{x}_{i}}\hat{f}_{Y|X}(\bs y^{*},\bs x^{*}) \right|-\left| \frac{\partial}{\partial\textcolor{black}{x}_{j}}f_{Y|X}(\bs y^{\textcolor{black}{*}},\bs x^{\textcolor{black}{*}}) \right|-\left| \frac{\partial}{\partial\textcolor{black}{x}_{i}}f_{Y|X}(\bs y^{*},\bs x^{*}) \right| +\left| \frac{\partial}{\partial\textcolor{black}{x}_{j}}f_{Y|X}(\bar{\bs y},\bar{\bs x}) \right|\lesssim \textcolor{black}{\bar{\epsilon}^*_d}$.
Since $\left| \frac{\partial}{\partial\textcolor{black}{x}_{i}}f_{Y|X}(\bs y^{*},\bs x^{*}) \right|-\left| \frac{\partial}{\partial\textcolor{black}{x}_{j}}f_{Y|X}(\bar{\bs y},\bar{\bs x}) \right|\leq 0,$ we then can obtain 
$
\left| \frac{\partial}{\partial\textcolor{black}{x}_{i}}\hat{f}_{Y|X}(\bs y^{*},\bs x^{*}) \right|-\left| \frac{\partial}{\partial\textcolor{black}{x}_{j}}f_{Y|X}(\bar{\bs y},\bar{\bs x}) \right|\lesssim \textcolor{black}{\bar{\epsilon}^*_d}
$ with a probability at least $1-2\exp(-\tau_{\textcolor{black}{D}}).$
Similarly, we can guarantee that if $\hat{L}-L< 0,$ $\left| \frac{\partial}{\partial\textcolor{black}{x}_{j}}f_{Y|X}(\bar{\bs y},\bar{\bs x}) \right|-\left| \frac{\partial}{\partial\textcolor{black}{x}_{i}}\hat{f}_{Y|X}(\bs y^{*},\bs x^{*}) \right|\lesssim \textcolor{black}{\bar{\epsilon}^*_d},$
with a probability at least $1-2\exp(-\tau_{\textcolor{black}{D}}).$

\end{proof}

\textcolor{black}{For the running example \eqref{run_examp}, the LC is the maximum absolute value of the partial derivatives of its CoDF in \eqref{run_examp_codf}. This LC is estimated via the partial derivative of its kernel-based non-parametric estimator \eqref{condi_densityestima} computed from the samples. By Corollary~\ref{LC_confide_estim}, the resulting estimate $\hat{L}$ converges asymptotically to the true LC $L$ with high probability.}

\begin{remark}
In the above, we showed through $\epsilon_d$ the expression on the upper bound of the difference between the estimated and original Lipschitz constant of the density function of the transition kernel.
A precise bound can be obtained by proving a closeness guarantee for the estimates of the third-derivative terms in \textcolor{black}{equation}~\eqref{pac_bound_lc}.
In addition, a coarse asymptotic upper bound can be derived when loose bounds on the higher-order derivatives of the conditional density function are available.
\end{remark}


\subsection{LC Estimation Algorithm and Closeness Guarantee}




Building on Corollary~\ref{LC_confide_estim}, we propose a novel algorithm using NPE to estimate the LC of the CoDF of the stochastic system $\Sigma_{ss}$, and we provide a probabilistic asymptotic bound for the estimation of the LC.
This guarantee is crucial for designing a partitioning strategy that ensures the probabilities of satisfying the \textcolor{black}{specifications} 
on $\Sigma_{ss}$ remain close to those on its finite abstraction. The partitioning strategy will be discussed in Section~\ref{sec:IMDP_abstraction_NPE}. In this section, we focus on deriving this new bound.


We propose a method to estimate the LC of a CoDF $f_{Y|X}$ on a given domain $D_X\times D_Y$ using \textcolor{black}{equation}~\eqref{condi_densityestima}. The estimation method is presented in Algorithm~\ref{algo_enviro}, which is based on taking samples of $X$ uniformly on $D_X$, then taking samples of $Y$ from $(Y|X)$ associated with samples of $X$, constructing $\hat f_{Y|X}$ according to \textcolor{black}{equation}~\eqref{condi_densityestima},
taking partial derivative of $\hat f_{Y|X}$, and finally computing the maximum absolute values of derivatives on the domain $D_X\times D_Y$ and across dimensions.
The algorithm also iterates over these steps and compute\textcolor{black}{s} the empirical mean of the results in Step~\ref{step_emp}. Note that the max operator in Step~\ref{step_max} corresponds to using infinity norm in the definition of the LC. Other norms could be used similarly.

\begin{algorithm}
\caption{(\textbf{LC\_Estimation}) Estimating the Lipschitz constant of $f_{Y|X}(\bs{y},
\bs{x})$}
\label{algo_enviro}
      
    \begin{algorithmic}[1]
    \REQUIRE  Domain $D_X\times D_Y$, sample generators of $(Y|X)$, number of iterations $m$
    \FOR{$\mu=1:m$}
    \STATE\label{step_H} Select bandwidths $H_{\mathsf x}$ and $H_{\mathsf y}$ and kernel $K(\cdot)$
    \STATE\label{step_sample} Select samples $ \{\hat{X}_i, i=1,\ldots,n\}$ uniformly from $D_X$
    \STATE For each $\hat X_i$, generate a sample $ \hat{Y}_i\in D_Y$ from $(Y|X_i)$, $i\in\{1,2,\ldots,n\}$
    \STATE Construct $\hat{f}_{Y|X}(\bs{y},
     \bs{x})$ using \textcolor{black}{equation}~\eqref{condi_densityestima}, samples $(\hat Y_i,\hat X_i)$, and kernel $K(\cdot)$
    \STATE For each $j\in\{1,\ldots,\mathsf d\}$, compute $\hat L_{\mu j}$ as
    \begin{equation}
    \label{eq:deriv_hat}
    \hat L_{\mu j} := \max_{(\bs{x},\bs{y})\in D_X\times D_Y}\left|\frac{\partial}{\partial x_{j}}\hat{f}_{Y|X}(\bs{y},\bs{x})\right|
    \end{equation}
    \ENDFOR
    \STATE\label{step_emp}Compute the empirical means $\hat L_j:=\frac{1}{m}\sum^{m}_{\mu=1}\hat L_{\mu j}$
    \STATE\label{step_max} Compute $\hat L=\max\{\hat L_{1},\hat L_{2},\ldots,\hat L_{\mathsf d}\}$
    \ENSURE Estimated LC $\hat L$
    \end{algorithmic}
\end{algorithm}

\textcolor{black}{
For the running example \eqref{run_examp}, the random variable $Y\in D_{Y}$ can be treated as the next state $X_{k+1}\in D_{X}$ of \textcolor{black}{equation}~\eqref{run_examp} under a fixed action $a$. Thus, $f_{Y|X}$ is the conditional density function $f_{X_{k+1}|X_{k}}$ in \textcolor{black}{equation}~\eqref{run_examp_codf}, and its LC can be represented as
\begin{equation*}
\max_{(\bs{x}(k),\bs{x}(k+1))\in D_X\times D_X}|\frac{\partial}{\partial x_{i}(k)}f_w(x(k+1)-f_r(x(k))-a(k))|.
\end{equation*}
Algorithm~\ref{algo_enviro} estimates the LC of this conditional density function. The asymptotic upper bound of this estimation is provided by Theorem~\ref{LC_main} in the sequel.
}


In the rest of this section, we formulate bounds on the bias and variance of the estimator $\hat L_\mu$ and discuss how to select the bandwidths $H_{\mathsf x}$ and $H_{\mathsf y}$ (cf. Step~\ref{step_H} of Algorithm~\ref{algo_enviro}) to tune the asymptotic bias of the estimation.
The total variance of the estimation can be reduced by increasing the iteration number $m$ of the algorithm.
Our theoretical results are established for the uniform distribution in Step~\ref{step_sample} of the algorithm. Similar results can be obtained for other distributions.

\subsubsection{Bias of LC Estimation}\label{b_lc_estima}

Consider the original and estimated LC across each dimension for a fixed $\mu$ defined as
\begin{align*}
L_{i}=\max_{(\bs{x},\bs{y})\in D_X\times D_Y}|\frac{\partial}{\partial x_{i}}f_{Y|X}(\bs{y},\bs{x})|,\qquad
\hat{L}_{i}=\max_{(\bs{x},\bs{y})\in D_X\times D_Y}|\frac{\partial}{\partial x_{i}}\hat{f}_{Y|X}(\bs{y},\bs{x})|,
\end{align*}
for $i\in \{1,\ldots,\mathsf d\}$ with $D_X$, $D_Y\subset \mathbb R^{\mathsf d}$.
Thus, the original and estimated LC are $L:=\max_{i=1,\ldots,\mathsf d}\{L_i\}$ and $\hat{L}:=\max_{i=1,\ldots, \mathsf d}\{\hat L_i\}$.

\begin{lemma}
    Suppose that $f_{Y|X}(\bs y,\bs x)$ satisfies {\normalfont Assumption~\ref{asstwod_1}}.
For any $(\bs x,\bs y)\in D_X\times D_Y$, $h>0$, we have that for large $n$ if $nh^{2\mathsf d+2}\to +\infty$ and $h\to 0$ as $n\to +\infty$, then
\begin{align*}
    \Var[& \frac{\partial}{\partial x_{i}} \hat{f}_{ Y|X }(\bs{y},\bs{x}) ]
    \lesssim \frac{1}{nh^{2\mathsf d+2}}G^{2\mathsf d-1\!}_{20}(k)\Vol(D_{X})C_{f},
\end{align*}
\begin{align*}
    \left|\E[\frac{\partial}{\partial x_{i}} \hat{f}_{ Y|X }(\bs{y},\bs{x})]-\frac{\partial}{\partial x_{i}} f_{ Y|X }(\bs{y},\bs{x})\right|\lesssim 
    \frac{1}{2}h^{2} \bigg| \sum^{\mathsf d}_{j=1}C_{ij}
    +\sum^{\mathsf d}_{\varsigma\ne i} \bar{C}_{i\varsigma}\bigg|,~i\in\{1,\ldots,\mathsf d\}
\end{align*}
where $\Vol(\cdot)$ \textcolor{black}{denotes} 
the volume (Lebesgue measure) of a set.
\end{lemma}

\begin{theorem}[\textbf{Bias of the Estimation}]\label{LC_main}
    Suppose that $f_{Y|X}(\bs y,\bs x)$ satisfies {\normalfont Assumption~\ref{asstwod_1}}.
For any $(\bs x,\bs y)\in D_X\times D_Y$, $h>0$, we have that for large $n$ if $nh^{2\mathsf d+2}\to +\infty$ and $h\to 0$ as $n\to +\infty$, then
\begin{equation*}
    \label{main_twodimen}
     \left|\E[\hat{L}_{i}] - L_{i}\right|\lesssim \epsilon_{b,i}^{\frac{1}{2}}, 
\end{equation*}
where $\epsilon_{b,i}\!:=\frac{1}{nh^{2\mathsf d+2}}G^{2\mathsf d-1\!}_{20}(k)\Vol(D_{X})C_{f}+\frac{1}{4}h^{4} \bigg( \sum^{\mathsf d}_{j=1}C_{ij}
    +\sum^{\mathsf d}_{\varsigma\ne i} \bar{C}_{i\varsigma}\bigg)^2,~i\in \{1,\ldots,\mathsf d\}.$
\end{theorem}

\begin{remark}
    The bandwidths should be selected appropriately such that $\epsilon_{b,i}\to 0$ for large $n$. The best convergence rate for $\epsilon_{b,i}$ is $O(n^{-\frac{2}{3+d}})$, where the bandwidth $h$ should be the order of $O(n^{-\frac{1}{6+2d}})$. Hence, $\epsilon^{\frac{1}{2}}_{b,i}$ has the best convergence rate of $O(n^{-\frac{1}{3+d}})$.
\end{remark}


Now that we have established how to estimate transition probabilities and the LC with asymptotic convergence rates, we discuss in the next section the construction of the IMDP abstraction with closeness guarantees on the verification and synthesis.

\section{Construction of IMDP Using Non-parametric Estimation}
\label{sec:IMDP_abstraction_NPE}

In this section, we use the non-parametric estimation to construct an IMDP as a finite abstraction of the system $\Sigma_{ss}=(\mathcal S, U,w,f)$ with formal asymptotic closeness guarantees.
First consider an IMDP $\bar\Sigma_{ss}=(Q, S_{\mathfrak a}, P_{lo},P_{up}, AP, L )$
Next we build an IMDP $\hat \Sigma_{ss}=(Q, S_{\mathfrak a}, P_{lo},P_{up}, AP, L )$ \textcolor{black}{whose} 
transition probabilities $P_{ij,a}(\bs x)$ are estimated through the non-parametric estimation based on the samples $\left\{\left(\textcolor{black}{\hat{X}_{i},\hat{ Y}_{i}} \right),~i=1,\ldots,n \right\}$.
Let $ \hat{P}_{ij,a} (\bs x)$ represent the probability estimated using non-parametric methods. 
The estimated upper and lower bounds of the transition probability from $q_{i}$ to $q_j$, can be represented as  
\begin{align}
    &\hat{P}_{lo,a}(q_{i},q_{j})=\min_{\bs x\in q_{i}}  \hat{P}_{ij,a} (\bs x)=\min_{x\in q_{i}} \int_{q_{j}} \hat{f}^{a}_{Y|X}(\bs y,\bs x) d\bs y, \label{tran_prob_optimal_1}\\
    &\hat{P}_{up,a}(q_{i},q_{j})=\max_{\bs x\in q_{i}}  \hat{P}_{ij,a} (\bs x)=\max_{x\in q_{i}} \int_{q_{j}} \hat{f}^{a}_{Y|X}(\bs y,\bs x) d\bs y,\label{tran_prob_optimal_2}
\end{align}
where $\hat{f}^{a}_{Y|X}$ is the estimator of CoDF \textcolor{black}{ under action $a$}, shown in \textcolor{black}{the equation}~\eqref{condi_densityestima}. We refer to $\hat \Sigma_{ss}$ as the NPE abstraction of $\Sigma_{ss}$.

\textcolor{black}{
In the context of the running example \eqref{run_examp}, $q_i,q_j\in Q$ indicate subsets of the two-dimensional state space. The bounds $\hat{P}_{lo,a}(q_{i},q_{j})$ and $\hat{P}_{up,a}(q_{i},q_{j})$ provide estimates of the pessimistic and optimistic true transition probabilities from any state $\bs x \in q_i$ to the set $q_j$ under action $a$. Specifically:
\begin{itemize}
    \item[1.] $\hat{P}_{lo,a}(q_i,q_j)$ is the estimate for the lowest probability with respect to $\bs x\in q_i$ of
    $$P_{ij,a} (\bs x):=\int f_w(\zeta)\mathbf 1 (\zeta\in (q_j-f_r(\bs x)-a))d\zeta.$$
    %
    \item[2.] $\hat{P}_{up,a}(q_i,q_j)$ is the estimate for the highest probability with respect to $\bs x\in q_i$ of the above quantity $P_{ij,a} (\bs x)$.
\end{itemize}
These interval estimates $  [\hat{P}_{lo,a}(q_i,q_j), \hat{P}_{up,a}(q_i,q_j)]  $ are then used to construct the IMDP abstraction $ \hat{\Sigma}_{ss}$.
}

\subsection{Closeness Guarantees for the IMDP Verification}\label{clos_IMDP}

Here, we establish theoretical guarantees for the closeness between the verification outcomes of the NPE abstraction $\hat \Sigma_{ss}$ and those of the original stochastic system. These guarantees provide a foundation for reliable verification and synthesis based on non-parametric estimation.
\textcolor{black}{The result in this section provides a theoretical correctness and convergence guarantee, showing that as the number of samples increases, the satisfaction probability computed through the proposed data-driven abstraction, converges to that of the original system with a tunable bounded error and with high probability. 
}

\begin{lemma}\label{bound_two_abstraction}
        For any specification $\psi$ that is formally verified via the satisfaction procedure in \eqref{upprob_path}, \textcolor{black}{Suppose further that the uniform Lipschitz condition in
Remark~\ref{rem:uniform_transition_probability} holds. Then,}
        we can have that for all $\tau_{p}$, there exists $ \textcolor{black}{\bar \epsilon^*}_{p}$, such that 
\begin{align*}
    \left| \hat{P}^{\bs T}_{up}(q) - \bar{P}^{\bs T}_{up}(q) \right|\lesssim \textcolor{black}{\bs T}n_{Q}\textcolor{black}{\bar{\epsilon}^*}_{p}\to 0,\text{ as }n\to +\infty,
\end{align*}
with a probability at least $1-2\exp{(-\tau_{p})},$ where $\bs T$ is the horizon of the specification $\psi$, $\hat{P}^{\bs T}_{up}(q)$ is the solution of equation~\eqref{upprob_path} computed for the IMDP $\hat \Sigma_{ss},$ and $\bar{P}^{\bs T}_{up}(q)$ is the solution of equation~\eqref{upprob_path} computed for the IMDP $\bar{\Sigma}_{ss}.$
 
        
\end{lemma}

\begin{proof}
According to Theorem~\ref{multi_Bound_TranProbi} \textcolor{black}{and the uniform extension stated in Remark~\ref{rem:uniform_transition_probability}, whose covering and union-bound argument is given in the proof of Proposition~\ref{prop:uniform_transition_probability}, 
together with an additional union bound over the finite collections of
source abstract states and actions,
for every
$\tau_p>0$,
}
\begin{align*}
    &\textcolor{black}{\max_{1 \leq j\leq n_{Q}}}\left|
    \max_{\bs x\in q_{i}}
    \hat{P}_{ij,a}(\bs x)
    -
    \max_{\bs x\in q_{i}}
    P_{ij,a}(\bs x)
    \right|
    \nonumber\\
    &\leq
    \textcolor{black}{\max_{1\leq j\leq n_{Q}}}\max_{\bs x\in q_{i}}
    \left|
    \hat{P}_{ij,a}(\bs x)
    -
    P_{ij,a}(\bs x)
    \right|
    \lesssim
    \textcolor{black}{\bar{\epsilon}^*}_{p}
    \to 0,
    \qquad n\to+\infty,
\end{align*}
with probability at least $1-2\exp(-\tau_{p})$.
\textcolor{black}{In the remainder of the proof, all inequalities are understood on this common high-probability event established by the preceding finite union-bound arguments; therefore, no additional union bound is required in the induction argument.}

\textcolor{black}{
We prove the result by induction on the horizon $\bs T$.
Assume that, for some $\bs T=\bs T^*$,
\[
\sup_{q}
\left|
\hat P^{\bs T^*}_{up}(q)
-
\bar P^{\bs T^*}_{up}(q)
\right|
\lesssim
\bs T^* n_Q\textcolor{black}{\bar{\epsilon}^*}_p.
\]
For horizon $\bs T=\bs T^*+1$, using $\left|
\max_{a} f_a-\max_{a} g_a
\right|
\le
\max_{a}
\left|
f_a-g_a
\right|,$
we obtain
\begin{align*}
&
\left|
\hat P^{\bs T^*+1}_{up}(q_i)
-
\bar P^{\bs T^*+1}_{up}(q_i)
\right|
\\
&\le
\max_{a}
\left|
\max_{\hat \theta^{a}_{q_i}}
\sum_{j=1}^{n_Q}
\hat \theta^{a}_{q_i}(q_j)
\hat P^{\bs T^*}_{up}(q_j)
-
\max_{\theta^{a}_{q_i}}
\sum_{j=1}^{n_Q}
\theta^{a}_{q_i}(q_j)
\bar P^{\bs T^*}_{up}(q_j)
\right|
\\
&=
\max_{a}
\left|
\max_{\hat \theta^{a}_{q_i}}\sum_{j=1}^{n_Q}
\hat \theta^{a}_{q_i}(q_j)
\hat P^{\bs T^*}_{up}(q_j)
-
\max_{ \theta^{a}_{q_i}}\sum_{j=1}^{n_Q}
\theta^{a}_{q_i}(q_j)
\hat P^{\bs T^*}_{up}(q_j)
+
\max_{ \theta^{a}_{q_i}}\sum_{j=1}^{n_Q}
 \theta^{a}_{q_i}(q_j)
\hat P^{\bs T^*}_{up}(q_j)
-
\max_{\theta^{a}_{q_i}}
\sum_{j=1}^{n_Q}
\theta^{a}_{q_i}(q_j)
\bar P^{\bs T^*}_{up}(q_j)
\right|
\\
&\le 
\max_{a} \left|\sum_{j=1}^{n_Q}
\hat \theta^{a,*}_{q_i}(q_j)
\hat P^{\bs T^*}_{up}(q_j)- 
\sum_{j=1}^{n_Q}
\theta^{a,*}_{q_i}(q_j)
\hat P^{\bs T^*}_{up}(q_j)\right|
+
\max_{a} \max_{ \theta^{a}_{q_i}}\sum_{j=1}^{n_Q}
\theta^{a}_{q_i}(q_j)\left|
\hat P^{\bs T^*}_{up}(q_j)-\bar P^{\bs T^*}_{up}(q_j)
\right|
\end{align*}
where 
$\hat\theta_{q_i}^{a,*}
\in
\arg\max_{\hat\theta\in\hat\Theta_{q_i}^{a}}
\sum_{j=1}^{n_Q}
\hat\theta(q_j)
\hat P_{up}^{\bs T^*}(q_j)$ and 
$\theta_{q_i}^{a,*}
\in
\arg\max_{\theta\in\Theta_{q_i}^{a}}
\sum_{j=1}^{n_Q}
\theta(q_j)
\hat P_{up}^{\bs T^*}(q_j),$ for each action $a.$ 
}

\textcolor{black}{
Since the transition distributions generated by the O-maximizing procedure
\cite{haddad2018interval}
are stochastic, $\sum_{j=1}^{n_Q}
\theta^{a}_{q_i}(q_j)=1,$ and $0\le
\hat P^{\bs T^*}_{up}(q_j)
\le1.$
Together with Theorem~\ref{multi_Bound_TranProbi}, 
this construction yields
$\sum^{n_{Q}}_{j=1}\left|
\hat\theta^{a,*}_{q_i}(q_j)
-
\theta^{a,*}_{q_i}(q_j)
\right|
\lesssim
n_{Q}\textcolor{black}{\bar{\epsilon}^*}_p.$
Hence,
\[
\max_{a} \left|\sum_{j=1}^{n_Q}
\hat \theta^{a,*}_{q_i}(q_j)
\hat P^{\bs T^*}_{up}(q_j)-
\sum_{j=1}^{n_Q}
\theta^{a,*}_{q_i}(q_j)
\hat P^{\bs T^*}_{up}(q_j)\right|
\lesssim
n_Q\textcolor{black}{\bar{\epsilon}^*}_p.
\]
Therefore,
\begin{align*}
&
\left|
\bar P^{\bs T^*+1}_{up}(q_i)
-
\hat P^{\bs T^*+1}_{up}(q_i)
\right|
\lesssim
\sup_q
\left|
\bar P^{\bs T^*}_{up}(q)
-
\hat P^{\bs T^*}_{up}(q)
\right|
+
n_Q\textcolor{black}{\bar{\epsilon}^*}_p
\lesssim
(\bs T^*+1)n_Q\textcolor{black}{\bar{\epsilon}^*}_p.
\end{align*}
The proof is completed by induction.
}

\end{proof}

\begin{remark}
    The above lemma gives a probabilistic convergence guarantee for a data-driven formal synthesis procedure between the finite abstraction  $\bar \Sigma_{ss}$ and its estimation $\hat\Sigma_{ss}$ from data. 
    The convergence rate is $O(n^{-1/(3+\mathsf d)})$, if $h$ is in the order of $O(n^{-1/(6+2\mathsf d)})$.
\end{remark}

Using Lemma~\ref{bound_two_abstraction} and Theorem~\ref{SA13_bound}, we establish the asymptotic closeness between satisfaction probabilities computed on DTSCS $\Sigma_{ss}$ and its finite abstraction based on non-parametric estimation, as detailed below.

\begin{theorem}\label{abstract_guarantee}
    Let  $\hat \Sigma_{ss}$ be
    the IMDP abstraction of DTSCS $\Sigma_{ss}$ with its transition probabilities are estimated through non-parametric estimation.
     For any specification $\psi$ through a certain strategy $\textcolor{black}{\varpi}\in \Pi$ satisfying procedure \eqref{upprob_path}, then we have that for all $\tau_{p}$, there exists $\epsilon_{p}$ such that 
\begin{align}\label{close_orgi_data_abstra}
    \left| P\left( \Sigma_{ss} \models \psi \right)-  P\left( \hat \Sigma_{ss} \models \psi \right) \right|\lesssim \epsilon_{0}+\epsilon^{*}_{p}, \text{ with }  \epsilon_{0}=\bs T\delta B_{L} \mathcal L \text{ and } \epsilon^{*}_{p}=\textcolor{black}{\bs T n_{Q}\textcolor{black}{\bar \epsilon^*}_{p}}, 
\end{align}
with a probability at least $1-2\exp{(-\tau_{p})},$ where $\bs T$ is the number of steps, $n_{Q}$ is the number of elements of $Q,$ $\delta$ is the state discretizations parameter, $B_{L}$ is the asymptotic upper bound of LC, and $\mathcal{L}$
is the Lebesgue measure of the specification set.     
\end{theorem}
\textcolor{black}{
\begin{proof}
Based on Lemma~\ref{bound_two_abstraction} and Theorem~\ref{SA13_bound},  we have 
\begin{align*}
    \left|P(\Sigma_{ss}\vDash \psi)-P(\hat{\Sigma}_{ss}\vDash \psi)\right|\leq& \left| P(\Sigma_{ss}\vDash \psi)- P(\bar{\Sigma}_{ss}\vDash \psi)+ P(\bar{\Sigma}_{ss}\vDash \psi)-P(\hat{\Sigma}_{ss}\vDash \psi) \right|\\
    \leq& \left| P(\Sigma_{ss}\vDash \psi)- P(\bar{\Sigma}_{ss}\vDash \psi)\right|+\left|P(\bar{\Sigma}_{ss}\vDash \psi)-P(\hat{\Sigma}_{ss}\vDash \psi) \right|\\
    \lesssim& \epsilon_0+\epsilon^*_p, \text{with $\epsilon_0=\bs T\delta B_L \mathcal{L}$,}
\end{align*}
with a probability at least $1-2\exp{(-\tau_{p})},$ where $P(\Sigma_{ss}\vDash \psi)$ is the probabilities that $\Sigma_{ss}$ satisfies the specification $\psi$ under strategy $\varpi$, 
$P(\bar{\Sigma}_{ss}\vDash \psi)=\bar P^{\bs T}_{up}(q)$, 
$P(\hat{\Sigma}_{ss}\vDash \psi)=\hat P^{\bs T}_{up}(q),~ q\in Q$, $\bs T$ is the number of steps, $n_Q$ is the number of entries of $Q$, $\delta$ is the state discretizations parameter, $B_L$ is the asymptotic upper bound of LC, and $\mathcal{L}$ is the Lebesgue measure of the specification set.
\end{proof}
}

The above theorem provides a quantitative guarantee that the satisfaction probability computed from a data-driven abstraction of an unknown DTSCS is close to the true satisfaction probability, with high probability.
    The closeness between $\Sigma_{ss}$ and its data-driven abstraction $\hat{\Sigma}_{ss}$ in probabilistic setting can asymptotically converge to the closeness between $\Sigma_{ss}$ and its finite abstraction $\bar{\Sigma}_{ss}$ with the order $O(kn_{Q}n^{-1/(3+\mathsf d)}).$
    Also, both errors in inequality~\eqref{close_orgi_data_abstra} can be systematically reduced by refining the partition parameter $\delta$ and collecting more data to reduce $\textcolor{black}{\bar \epsilon^*}_{p}$.
    Thus, the result ensures that the data-driven formal synthesis procedure is sound and asymptotically convergent, making it a reliable approach when the system model is unknown.

\textcolor{black}{
Note that the error bound given in the above theorem increases linearly with the number of time steps $\bs T$. This makes it suitable for the specifications that are time bounded (i.e., their satisfaction can be concluded using finite-horizon paths of the system). This includes the time-bounded until operator $  \mathcal{U}^{\le \bs T}$. Addressing unbounded until operator while using MDPs as the underlying abstract model requires additional structural assumptions on the dynamics of the system \cite{majumdar2020symbolic,soudjani2015quantitative}.
}

\subsection{Non-parametric Formal Verification Algorithm}

In this subsection, we provide the algorithm to perform the non-parametric formal verification and synthesis for the unknown DTSCS with the workflow shown in Fig.~\ref{verification_algorithm} and using Algorithm \ref{algo_enviro}, referred to as \emph{LC\_Estimation}. 
This algorithm requires the coarse values of the parameters in Assumption~\ref{asstwod_1}~(a), (c) and (d).


\begin{algorithm}
\caption{Non-parametric Formal Verification and Policy Synthesis}
\label{algo_npe_fv}
      
    \begin{algorithmic}[1]
    \REQUIRE  Domain $D_X\times D_Y$, data scale $n$, samples $\{\textcolor{black}{\hat{X}_{i}},i=1,\ldots,n \}$, sample generators of $(Y|X)$, number of iterations $m$, Specification $\psi$, parameters $C_{f},$ $C_{ij},$ and \textcolor{black}{$\bar{C}_{i\varsigma}$} $i,j,\textcolor{black}{\varsigma}\in \{1,\ldots,\mathsf d\}$, error $\epsilon_{0}$
    
    \STATE \textcolor{black}{$\hat{L}$} $\leftarrow$ $\emph{LC\_Estimation}(D_{X}\times D_{Y}, Y|X, m)$
    \STATE
\textcolor{black}{Compute
$B_L$
using
$\hat L$,
$C_f$,
$C_{ij}$,
$\bar C_{i\varsigma}$
according to
Theorem~\ref{LC_main}}
    \STATE Determine $\delta$ based on $B_{L}$ and $\epsilon_{0}$
    \STATE Partitioning the state space with discretization parameter $\delta$
    \STATE Generate the transition matrix $M^{up}_{tran}$ and $M^{lo}_{tran}$ using \textcolor{black}{equation}~\eqref{tran_prob_optimal_1} and \eqref{tran_prob_optimal_2}
    \STATE Perform the formal verification and synthesis using $\textcolor{black}{(M^{lo}_{tran},M^{up}_{tran})}$ \textcolor{black}{according to \eqref{lowprob_path}
and
\eqref{upprob_path}}
    \ENSURE Obtain the satisfaction probability of $\psi$ with a control policy
    \end{algorithmic}
\end{algorithm}

This algorithm estimates the upper bound $B_{L}$ of the LC as in Step \textcolor{black}{2} using Algorithm~\ref{algo_enviro}, and
has asymptotic convergence guarantee with the increasing of the data scale $n$.
In addition, it can give the asymptotic upper bound of the formal verification result if the rough values of all parameters in Assumption~\ref{asstwod_1} are known.


\section{Numerical Study}
\label{sec:case_studies}

In this section, we run our numerical studies on a MacBook Pro with an M4 chip and 24 GB of memory, using MATLAB 2024a. \textcolor{black}{First, we demonstrate the estimation of transition probabilities and the LC on two linear systems to demonstrate the asymptotic properties discussed in Section~\ref{main}. Then, we discuss verification and synthesis of policies on a nonlinear stochastic system.
Finally, we compare the proposed framework with an approach based on empirical estimation of transition probabilities.} 

\subsection{Estimation of Transition Probabilities and LC}
\begin{example}
\begin{figure}
    \centering
    \includegraphics[width=0.65\linewidth]{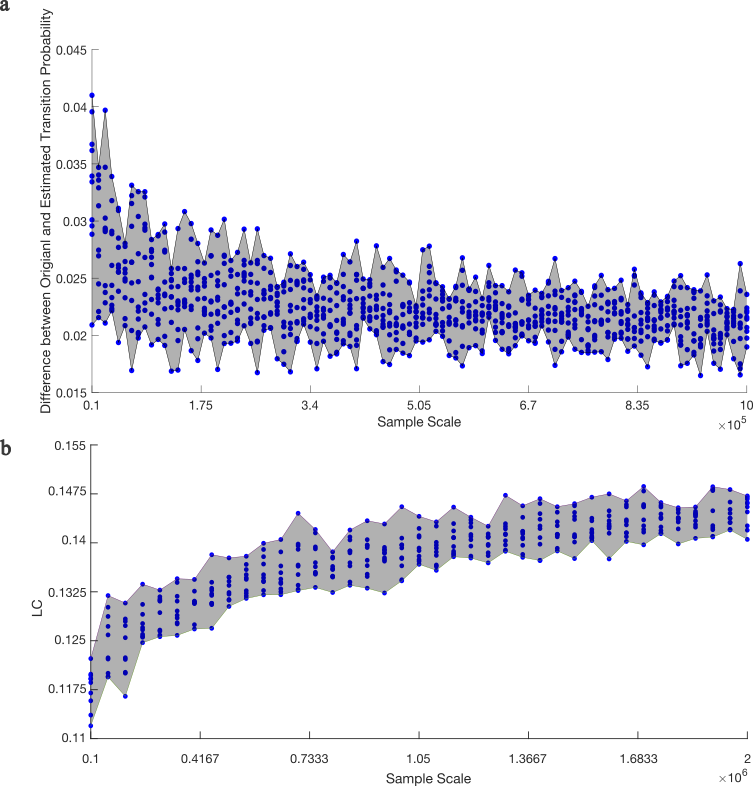}
    \caption{Panel (a) shows the difference between the original and estimated transition probability with varying data scale $n$.
    Panel (b) shows the LC with varying data scale $n$. }
    \label{compar_estim}
\end{figure}
Consider an unknown linear system
$
\textcolor{black}{\bs x}(k+1)=A\textcolor{black}{\bs x}(k)+\textcolor{black}{w}(k)
$
with noise $\textcolor{black}{w}(\cdot)\sim \mathcal{N}(\mu,\Sigma)$, where $
A=\begin{bmatrix}
2 & 0  \\
0 & 1 
\end{bmatrix}$, and $\textcolor{black}{w}$ has Gaussian distribution with mean $\mu=\begin{bmatrix}
0  \\
0\end{bmatrix}$ and variance $\Sigma=\begin{bmatrix}
1 & 0  \\
0 & 1 
\end{bmatrix}$.
Its CoDF is 
\begin{equation*}
   f_{X_{k+1}|X_{k}}(\textcolor{black}{\bs x}(k+1),\textcolor{black}{\bs x}(k))=\frac{1}{ 2\pi\sqrt{|\Sigma|} }\exp{ ( -\frac{1}{2}( \textcolor{black}{x}(k+1)-A\textcolor{black}{\bs x}(k) )^{T} \Sigma^{-1} ( \textcolor{black}{\bs x}(k+1)-A\textcolor{black}{\bs x}(k) )  ) }.
\end{equation*}
Consider the sample space $D_{\bs X}=[-1,1]\times [-1,1].$
The non-parametric estimation is employed to approximate the transition probability starting from the point $\textcolor{black}{\bs x}(k)=\begin{bmatrix}
0.5  \\
0.5\end{bmatrix}$ to the region $D_{I}=[0,1]\times [0,1]$. The true value of the transition probability is $0.1484$. 
For the estimated transition probability, the bandwidth $h$ is selected based on the Scott's formula, and can guarantee the error in \textcolor{black}{inequality}~\eqref{ineq_multi_Bound_TranProbi} is the order of $O(n^{-2/(4+ \mathsf d)})$, with high probability.
Fig.~\ref{compar_estim}(a) shows the difference between the original and estimated transition probability with varying data scale.
By conducting 10 times experiments at each data scale, their differences diminish and become more tightly clustered.
In Fig.~\ref{compar_estim}(b), the estimated LC using non-parametric estimation also become\textcolor{black}{s} closer to the original LC, which is $0.1931$.

\end{example}

\begin{example}\label{examp2}

\begin{figure*}[hbt!]
\centering
\includegraphics[width=\textwidth]{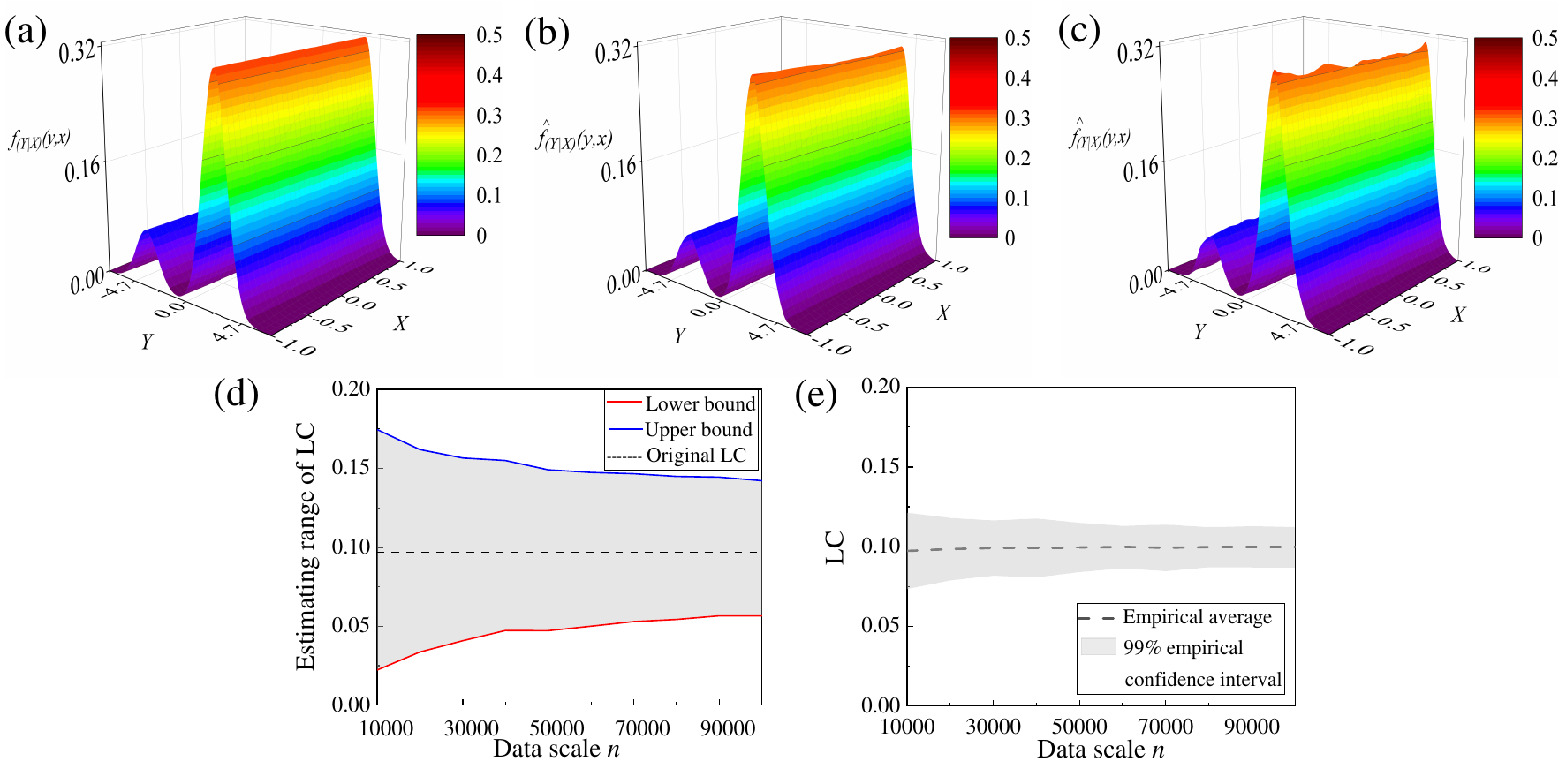}
\vspace{-0.45cm}
\caption{(a) The original CoDF. (b) Estimated CoDF with $h_{\mathsf x}=h_{\mathsf y}=n^{-\frac{1}{8}}$ and data scale $n=6\times 10^{4}$.
(c) Estimated CoDF with bandwidth from the Scott's formula and data scale $n=6\times 10^{4}$.
(d) Asymptotic bound on the original LC provided by Theorem~\ref{main_twodimen} as a function of data scale $n$. The dashed line is the original LC of the CoDF, $L=0.0968$.
(e) The estimated LC averaged over 150 computations with the \textcolor{black}{gray} area indicating the 99\% empirical confidence interval (3 times the empirical standard deviation from the mean).}
\label{case_2}
\end{figure*}
Consider a one-dimensional system with $\textcolor{black}{\bs y}=a\textcolor{black}{\bs x}+\textcolor{black}{\lambda} w_{1}+(1-\textcolor{black}{\lambda})w_{2}$, where
$w_1$ and $w_{2}$ have Gaussian distributions with means $\mu_1$ and $\mu_2$ and variances $\sigma_1^{2}$ and $\sigma_2^{2}$. The variable $\textcolor{black}{\lambda}\in\{0,1\}$ has Bernoulli distribution with success probability 
$P(\textcolor{black}{\textcolor{black}{\lambda}}=1)=p$. 
The CoDF $f_{Y|X}(\textcolor{black}{\bs y},\textcolor{black}{\bs x})$ is
\begin{align}
\label{examp2_density}
    f_{Y|X}(\textcolor{black}{\bs y},\textcolor{black}{\bs x})=& \frac{p}{\sigma_1\sqrt{2\pi}}\!\exp{( -\frac{1}{2\sigma_1^2}(\textcolor{black}{\bs y}-a\textcolor{black}{\bs x}-\mu_1)^{2} ) }
    \!+\!\frac{1-p}{\sigma_2\sqrt{2\pi}}\exp{ ( -\frac{1}{2\sigma_2^2}(\textcolor{black}{\bs y}-a\textcolor{black}{\bs x}-\mu_2)^{2} ) }.
\end{align}
We fix the parameters $a=0.5$, $\mu_1=3$, $\mu_2=-3$, $\sigma_1=\sigma_2=1$, $p=0.8$, and the domain $D_X=[-1,1]$ and  $D_Y=[-7.177,6.965]$.
We assume that Assumption~\ref{asstwod_1} holds with $C_{f}=1$ and the bound of third derivative terms are $ 0.5$. We run Algorithm~\ref{algo_enviro} with $m = 20$.

Fig.~\ref{case_2} shows the original CoDF together with the estimated CoDF using the bandwidth from \textcolor{black}{Theorem~\ref{LC_main}} and from the Scott's formula.
The estimated CoDF based on the bandwidths $h_{\mathsf x}=h_{\mathsf y}=n^{-\frac{1}{8}}$ as shown in Fig.~\ref{case_2}(b), is more similar to the original CoDF in Fig.~\ref{case_2}(a) and shows better smoothness in comparison with the estimation based on Scott's formula in Fig.~\ref{case_2}(c).

Fig.~\ref{case_2}(d) gives the
asymptotic bounds on the original LC $L=0.0968$ provided by \textcolor{black}{Theorem~\ref{LC_main}} as a function of data scale $n\in [10^4, 10^5]$.
For example, $L\in [0.022,0.175]$ using $n=10^4$,
$L\in [0.050,0.147]$ using
$n=6\times 10^4$,
and
$L\in [0.057,0.142]$
using $n=10^5$,
which confirms asymptotic convergence of the bound as a function of $n$.
The empirical mean and the 99\% empirical confidence interval are shown in Fig.~\ref{case_2}(e) using 150 runs of the algorithm for each $n$.
All the values are below the analytical upper bound shown in Fig.~\ref{case_2}(d).



\end{example}

\begin{example}\label{exam_cauchy_lc}
\textcolor{black}{
    Consider a one-dimensional system with $\bs y=a \bs x+w$, with noise $w\sim \mathrm{Cauchy}(x_{0},\gamma)$, where $\bs x_0\in\mathbb{R}$ denotes the location parameter and $\gamma>0$ denotes the scale parameter. We select $a=2$, $\bs x_0=0$, $\gamma=1$, the domain $D_X=[-1,1]$ and $D_Y=[-1,1]$, and the parameters $C_{f}=0.5$, $C_{ij}=12$, $\bar{C}_{i\varsigma}=3$ in Assumption \ref{asstwod_1}. 
    \begin{figure}
        \centering
        \includegraphics[width=0.85\linewidth]{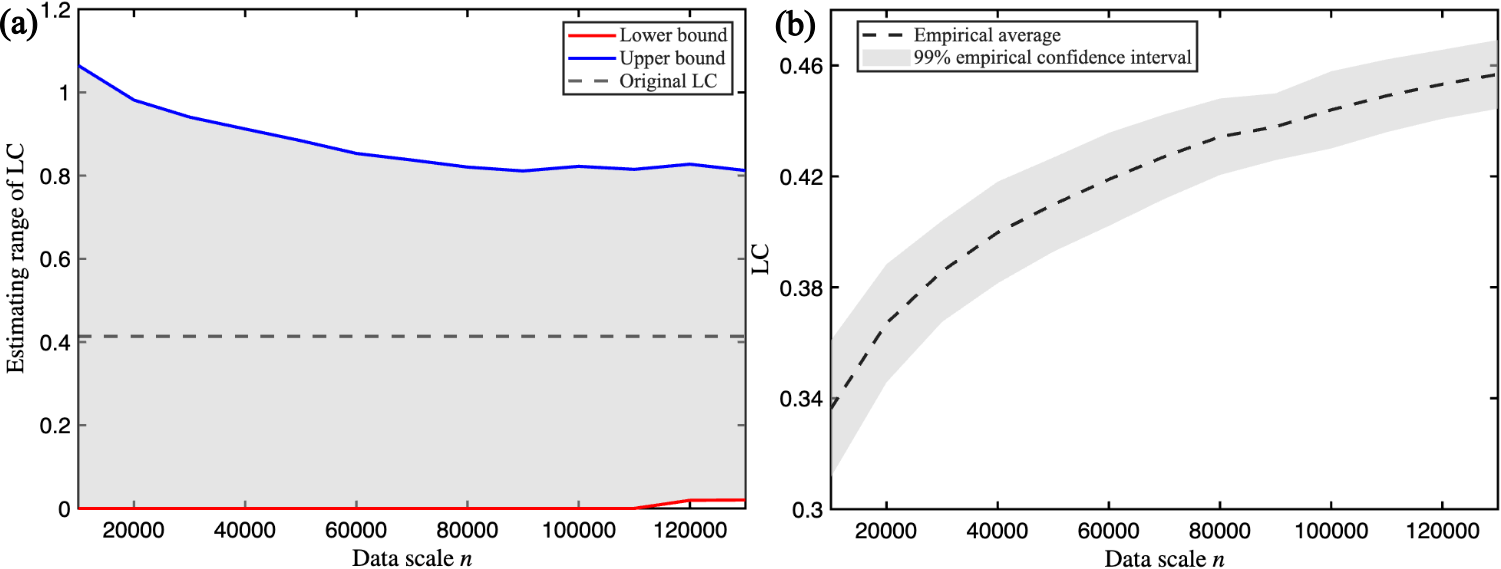}
        \caption{(a) Asymptotic bound on the original LC under Cauchy noise, derived in Theorem~\ref{main_twodimen}, as a function of the data scale $n$. The dashed line is the original LC of the CoDF, $L=0.4135$.
        (b) The estimated LC averaged over $100$ computations with the \textcolor{black}{gray} area indicating the 99\% empirical confidence interval (3 times the empirical standard deviation from the mean).} 
        \label{fig:cauchy_lc}
    \end{figure}
}

\textcolor{black}{Fig.~\ref{fig:cauchy_lc}(a) illustrates the asymptotic bounds on the LC $L=0.4135$ provided by Theorem~\ref{LC_main} as a function of the data scale $n\in[10^4,13\times10^4]$.
As the number of samples increases, the width of the theoretical interval gradually decreases, indicating the asymptotic convergence of the estimated bounds.
The empirical mean and the $99\%$ empirical confidence interval are reported in Fig.~\ref{fig:cauchy_lc}(b) using $100$ independent runs of the algorithm for each value of n.
Furthermore, all the values remain entirely contained within the analytical bounds shown in Fig.~\ref{fig:cauchy_lc}(a).
}
\end{example}

\subsection{Case Studies}\label{sec:case_studies}

\begin{example}
\label{exam_2}
\textcolor{black}{
Consider the two-dimensional stochastic bistable switch in the running example \eqref{run_examp}. The noise has Gaussian distribution with zero mean 
and covariance
$\Sigma=\textsc{diag}([0.1,0.1])$,
which gives the CoDF in \eqref{run_examp_codf2} assumed to be unknown.
We consider the system without action (i.e., $U = \{(0,0)\}$) and verify the temporal specification
$\psi=\neg r_{O} \mathcal{U}^{\leq K} r_D$, where $r_O$ is $[0.4,0.6]\times [0.5,0.6]$
and $r_D$ is $[0.9,1]\times [0.9,1]$.
}
In this case study, we focus on the asymptotic performance of formal verification as the number of samples $n$ increases, as this highlights how the asymptotic behavior of the estimated transition probabilities ensure\textcolor{black}{s} the asymptotic guarantees of the verification process.
We set the partition parameter to a fixed value of 
$\delta=0.1$. 
\textcolor{black}{The asymptotic errors can be found in Fig.~\ref{compar_formal_verif} computed for $K=5$.} With the data $n$ increasing, the formal verification gap between the original and estimated systems narrows, reflected by expanding dark blue areas in Fig.~\ref{compar_formal_verif}. 
\textcolor{black}{The maximum difference between the original satisfaction probability and its approximation is shown in Fig.~\ref{compar_formal_verif_longer}, which increases with $K$ and confirms the bound provided by Theorem~\ref{abstract_guarantee}.}

\begin{figure}
    \centering
    \includegraphics[width=0.75\linewidth]{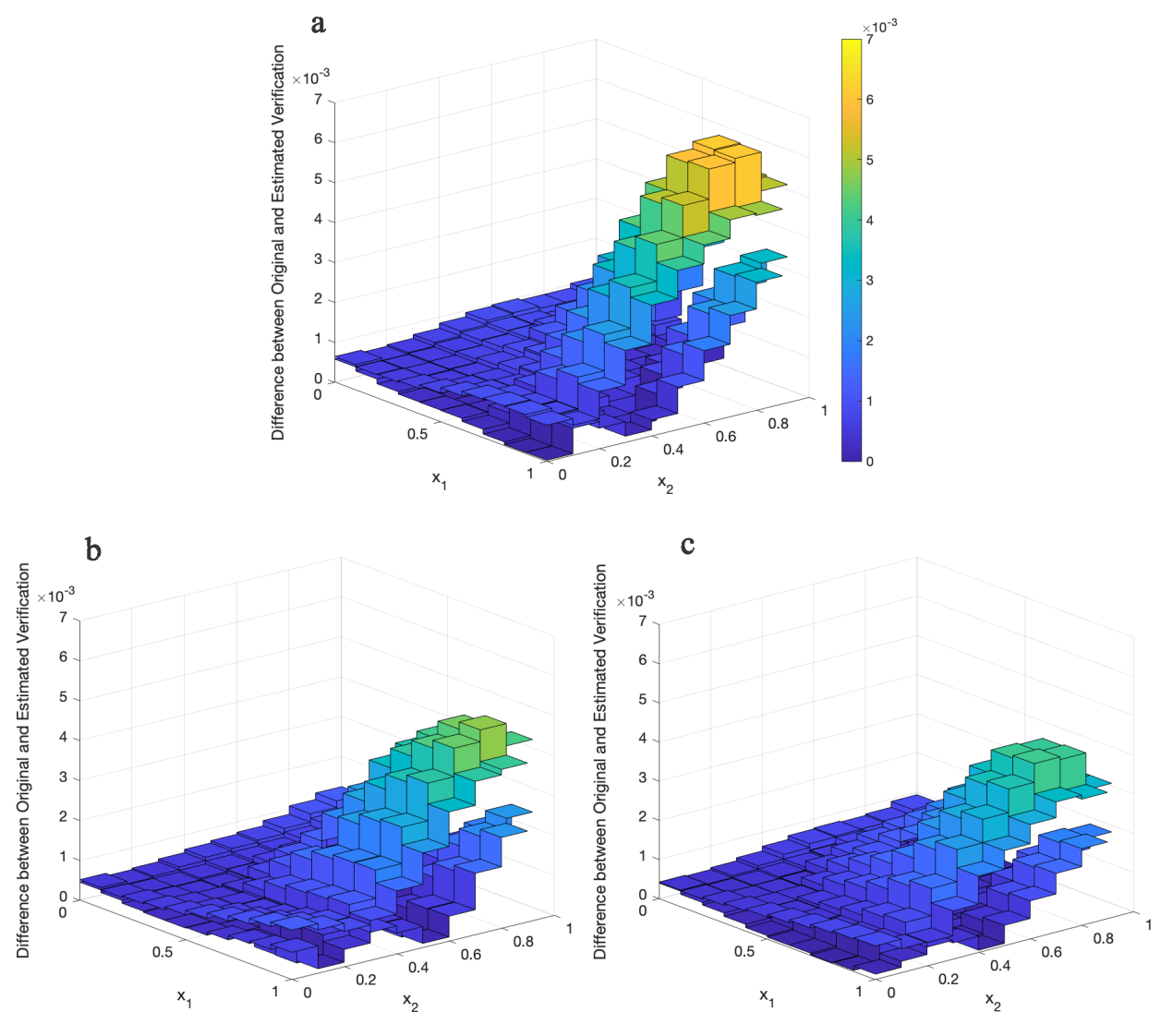}
    \caption{Panel (a), (b) and (c) show the difference between the original and estimated verification with varying data scale $n=2\times 10^4$, $n=5\times 10^4$ and $n=10^{5}$.
    }
    \label{compar_formal_verif}
\end{figure}

\begin{figure}
    \centering
    \includegraphics[width=0.38\linewidth]{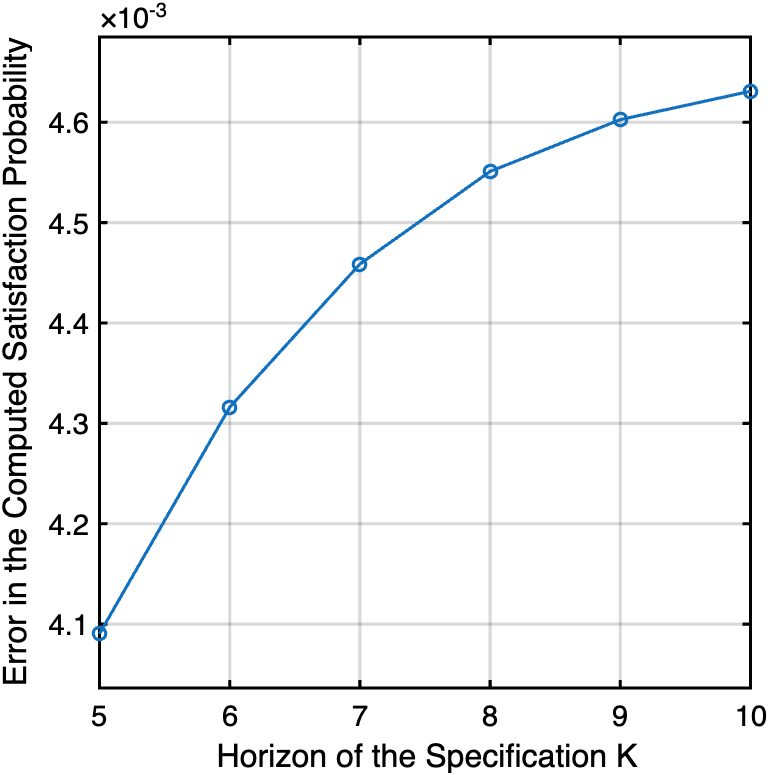}
    \caption{ \textcolor{black}{The maximum difference between the original satisfaction probability and its approximation for the specification $\psi=\neg r_{O} \mathcal{U}^{\leq K} r_D,$ when $n=10^{5}$.}
    }
    \label{compar_formal_verif_longer}
\end{figure}

\textcolor{black}{Consider the same setting with a different noise distribution, namely Cauchy noise $w\sim \mathrm{Cauchy}(\bs x_{0},\gamma),$ where $\bs x_0=\bs 0$ and $\gamma=0.5$.
 Increasing the dataset scale from $2\times 10^4$ to $10^5$ further reduces the estimation error, providing empirical evidence of the convergence behavior predicted by the theoretical analysis, as shown in Fig.~\ref{fig:cauchy_verif}. }

\begin{figure}
    \centering
    \includegraphics[width=0.8\linewidth]{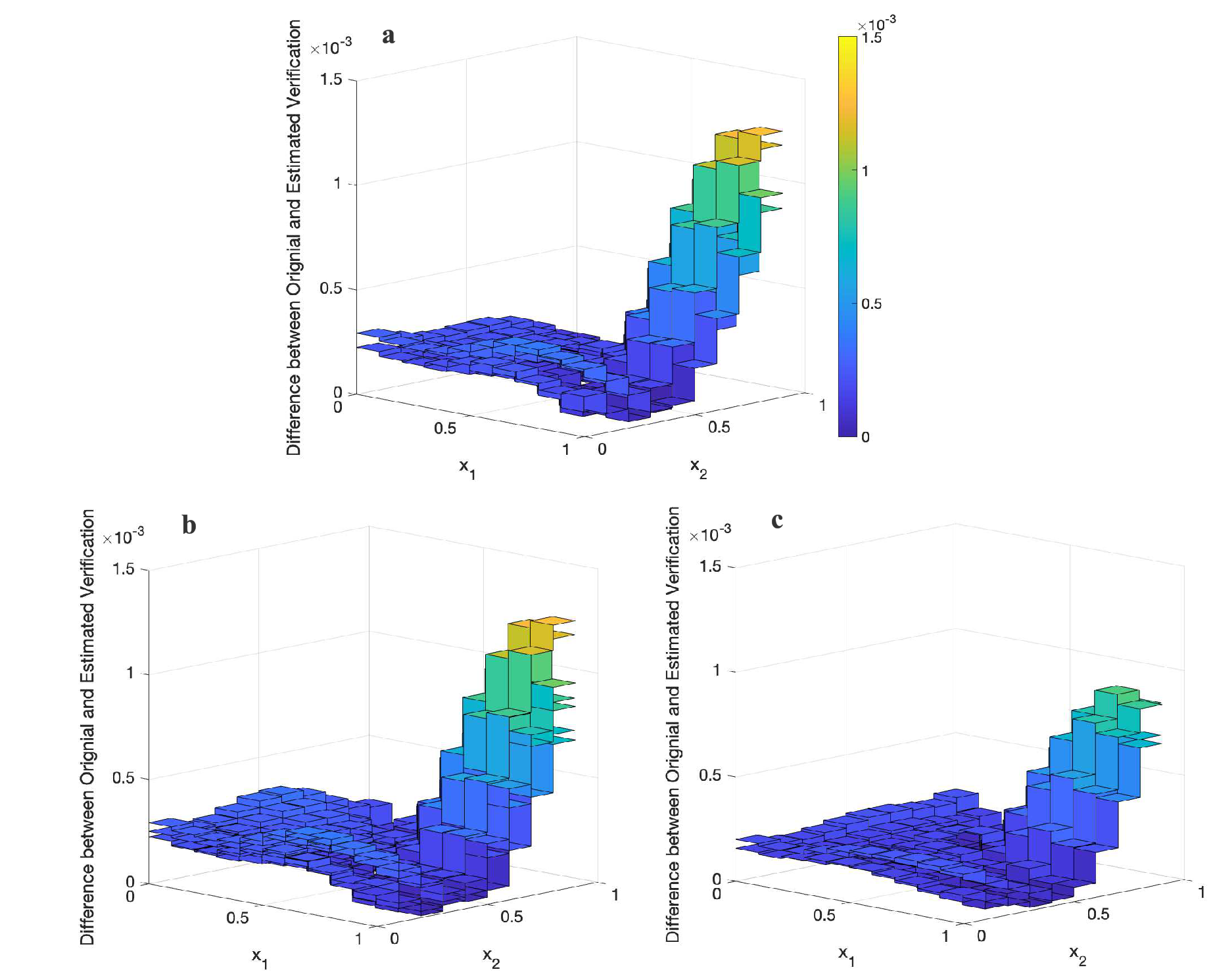}
    \caption{\textcolor{black}{Panel (a), (b) and (c) show the difference between the original and estimated verification under Cauchy noise $w\sim \mathrm{Cauchy}(0,0.5)$ for different data scale: $n=2\times 10^4$, $n=5\times 10^4$ and $n=10^{5}$.}  }
    \label{fig:cauchy_verif}
\end{figure}


\end{example}

\begin{example}
\textcolor{black}{
Consider the two-dimensional stochastic bistable switch in the running example \eqref{run_examp} subject to actions $u_{1}(k)\in \{0,0.15\}$ and $u_{2}(k)=0$.
The noise follows a Gaussian distribution with zero mean and covariance
$\Sigma=\textsc{diag}([0.1,0.1])$,
which gives the CoDF in \eqref{run_examp_codf2} assumed to be unknown.
The rest of parameters and the desired temporal specification are the same as in Example~\ref{exam_2}. The goal is to synthesize a control policy for satisfying the specification with high probability.}
As shown in Fig.~\ref{compar_formal_verif_actions}, the distance between the upper and lower bounds of the satisfaction probabilities decreases with the increasing of data scale. Furthermore, Fig.~\ref{action_synthe} illustrates that the control synthesis with a data scale of $n = 5 \times 10^4$ more closely aligns with the model-based scenario than the case with $n = 2 \times 10^4$. These observations validate our non-parametric framework’s asymptotic consistency and convergence in synthesizing control policies.



\begin{figure}
    \centering
    \includegraphics[width=0.75\linewidth]{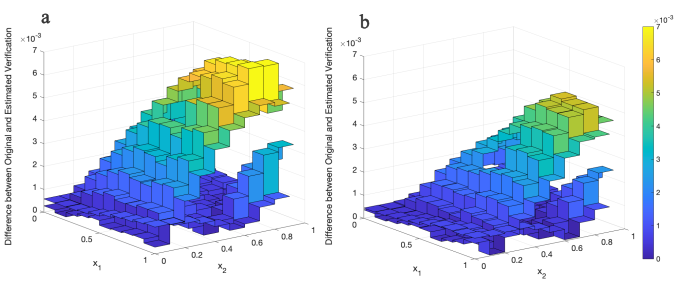}
    \caption{Panel (a) and (b) show the difference between the original and estimated verification with varying data scale $n=2\times 10^4$ and $n=5\times 10^4$ .
    }
    \label{compar_formal_verif_actions}
\end{figure}

\begin{figure}
    \centering
    \includegraphics[width=\linewidth]{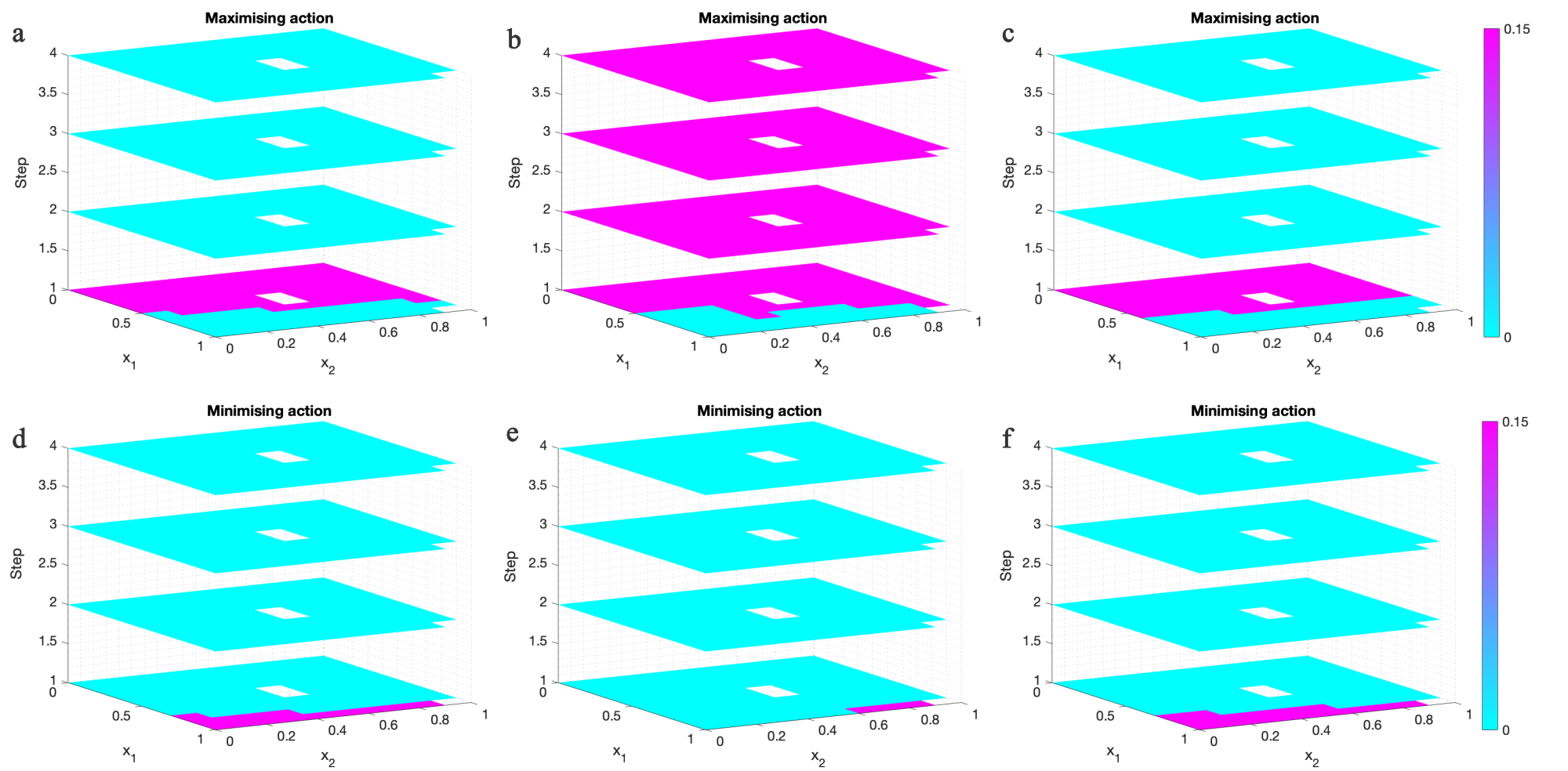}
    \caption{The control policies as a function of state for minimizing and maximizing the satisfaction probability.
    The left column shows the results
    from a model-based approach.
    The middle and right column are for the NPE under data scale $n=2\times 10^4$ and $n=5\times 10^4$, respectively.
    }
    \label{action_synthe}
\end{figure}

\end{example}

\textcolor{blue}{
\begin{figure}[h]
\vspace{-.2cm}
\centerline{\includegraphics[width=0.6\textwidth]{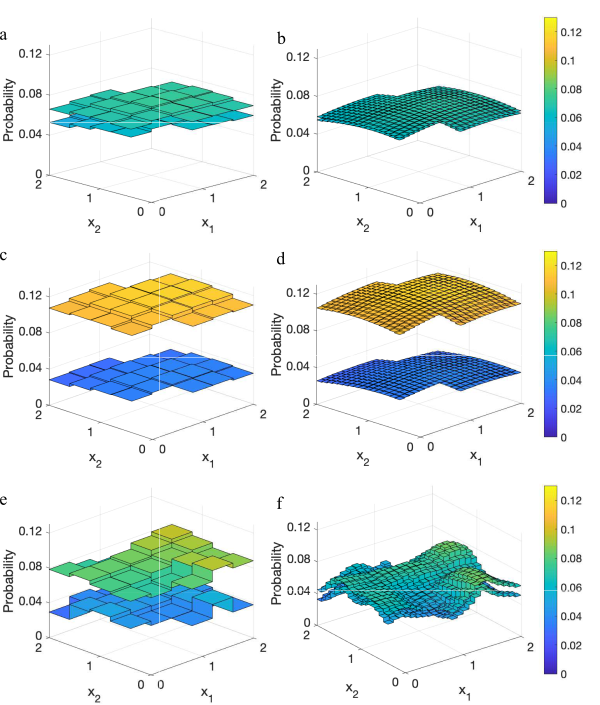}}
\vspace{-.05cm}
\caption{\textcolor{black}{The upper-bound and lower-bound probabilities of paths initialized at each state satisfying $\psi$ for the known system in panels (a) and (b), and its data-driven approximations using empirical IMDP abstraction in panels (c) and (d) and NPE in panels (e) and (f), under $\delta=0.4$ and $\delta=0.1$, respectively, where $r_O=[1.2,2]\times [1.6,2]$ and $r_D=[0,0.8]\times [0,0.4]$.}
}
\label{no_strategy}
\vspace{-0.3cm}
\end{figure}
}

\textcolor{black}{
\subsection{Comparison with an Alternative Approach}
}

\textcolor{black}{In this subsection, we compare our framework with an alternative data-driven method that builds an IMDP directly based on empirical estimation of transition probabilities.}

\smallskip
\noindent
\textcolor{black}{
\textbf{Empirical Construction of the IMDP Abstraction.}
The IMDP abstraction can be constructed empirically by taking samples and using Chebyshev's inequality.
Let $\mathsf T^a$ denote the transition probability matrix of an MDP with a certain partitioning construction under action $a\in S_a$, where its entry $P^{a}_{ij}$ is the transition probability from state $q_i\in Q$ to state $q_{j}\in Q$ under action $a$, and defined as $P^{a}_{ij}=\mathsf \theta^{a}_{q_{i}}(q_j)$. Also we can define the empirical estimation of $P^{a}_{ij}$ as $\bar{P}^{a}_{ij}=\frac{1}{G_{ij}}\sum^{G_{ij}}_{m=1}\boldsymbol{1}(\bar{q}^m_j\in q_j|\bar{q}^m_i\in q_i)$, where $\boldsymbol{1}(\cdot)$ is the indicator function such that $\boldsymbol{1}(\bar{q}^m_j\in q_j|\bar{q}^m_i\in q_i)=1$ if $\bar{q}^m_j\in q_j$, and $\boldsymbol{1}(\bar{q}^m_j\in q_j|\bar{q}^m_i\in q_i)=0$ otherwise, and $G_{ij}$ is the required number of data for this empirical estimation. 
By a-priori fixing the threshold $\bar{\epsilon}\in (0,1]$ and \textcolor{black}{parameter} $\Bar{\beta}\in(0,1)$, according to Chebyshev's inequality \citep{saw1984chebyshev}, we have 
\begin{equation}
    \label{ineq_empircalappro}
    \mathbb{P}\{\Bar{P}^{a}_{ij}-\Bar{\epsilon}\leq P^{a}_{ij}\leq \Bar{P}^{a}_{ij}+\Bar{\epsilon} \}\ge 1-\Bar{\beta},
\end{equation}
where $G_{ij}\ge \frac{1}{4\Bar{\beta}\Bar{\epsilon}^2}.$ From above, we obtain the upper and lower bound of transition probability from state $q_{i}$ to state $q_{j}$ with a confidence at least $1-\Bar{\beta}$.
Let $\epsilon_g>0$ denote the prescribed global tolerance for the verification error induced by the empirical IMDP abstraction. Following the error allocation strategy in \cite{zhang2024formal}, we select the transition probability tolerance as 
$\bar{\epsilon}=\frac{\epsilon_g}{2Kn_Q},$
thereby distributing the global error tolerance across the $K$-step verification horizon and the $n_Q$ abstract states.}


\smallskip

\begin{example}
\label{verif_case1}
\textcolor{black}{
Consider an unknown linear system
$\bs x(k+1)=A\bs x(k)+\bs w(k),$
with noise $\bs w \sim \mathcal{N}(\mu,\Sigma)$, and the PCTL path formula $\psi=\neg r_O \mathcal{U}^{\leq 3} r_D$, where $\mu=\bs 0$ and $\Sigma=\begin{bmatrix}
    1 & 0\\
    0 &1
\end{bmatrix}$.
We assume that the third derivatives of its CoDF are bounded by $0.2$.
Let $\epsilon_g$ denote the prescribed upper bound on the verification error. We fix $\epsilon_g=0.2$ for the empirical IMDP abstraction and use a dataset of size $n=2000$ for the proposed NPE method.
Based on Algorithm \ref{algo_enviro} and Theorem \ref{LC_main}, the asymptotic upper bound of LC is $0.0722$. 
Then we perform verification under two state discretization parameters $\delta=0.4$ and $\delta=0.1$. 
When $\delta=0.4$, the probability of each state of the above system falls within the range between the upper-bound and lower-bound probabilities satisfying $\psi$ as shown in Figs.~\ref{no_strategy}c and \ref{no_strategy}e. With a finer partition under $\delta=0.1$, more details of each state on satisfying specification $\psi$ are presented in Figs.~\ref{no_strategy}d and \ref{no_strategy}f. Unlike original system in Fig.~\ref{no_strategy}b and NPE in Fig.~\ref{no_strategy}f, the distances between the upper-bound and lower-bound probabilities do not decrease significantly for the empirical method in Fig.~\ref{no_strategy}d, but the probabilities in Fig.~\ref{no_strategy}b are still in range between these upper-bound and lower-bound probabilities. 
Meanwhile, for NPE, the probabilities in Fig.~\ref{no_strategy}f are very close to the results of original system in  Fig.~\ref{no_strategy}b and exhibit greater variation at various states than the results in Fig.~\ref{no_strategy}e.
}

\textcolor{black}{
Consider the same dynamical system with a heavy-tailed disturbance following the Cauchy distribution
$w\sim\mathrm{Cauchy}(\bs x_0,\gamma)$,
where $\bs x_0=\bs{0}$ and $\gamma=0.25$.
The state space is
$[0,1]\times[0,1]$
and the PCTL specification is
$\psi=\neg r_O\,\mathcal{U}^{\leq3}r_D$,
where
$r_O=[0.4,0.6]\times[0.5,0.6]$
and
$r_D=[0.9,1]\times[0.9,1]$.
Fig.~\ref{fig:empirical_NPE_comparison} compares the verification errors obtained by the empirical method and the proposed NPE framework with data scale ranging from $10^4$ to $5\times10^4$.
Here, the empirical method considers a confidence level of $95\%$.
We use $G_{ij}=100$ samples for each transition probability when the total data scale is $10^4$, corresponding to an estimation tolerance of $\bar{\epsilon}=0.2236$, and $G_{ij}=500$ samples when the total data scale is $5\times10^4$, corresponding to $\bar{\epsilon}=0.1$.
For both approaches, increasing the amount of data reduces the verification error, which is consistent with the theoretical convergence analysis.
Under the same heavy-tailed disturbance, the proposed NPE framework consistently achieves errors that are much smaller than those of the empirical method, and presents the asymptotic convergence predicted by the theory.
}

\begin{figure}
    \centering
    \includegraphics[width=0.8\linewidth]{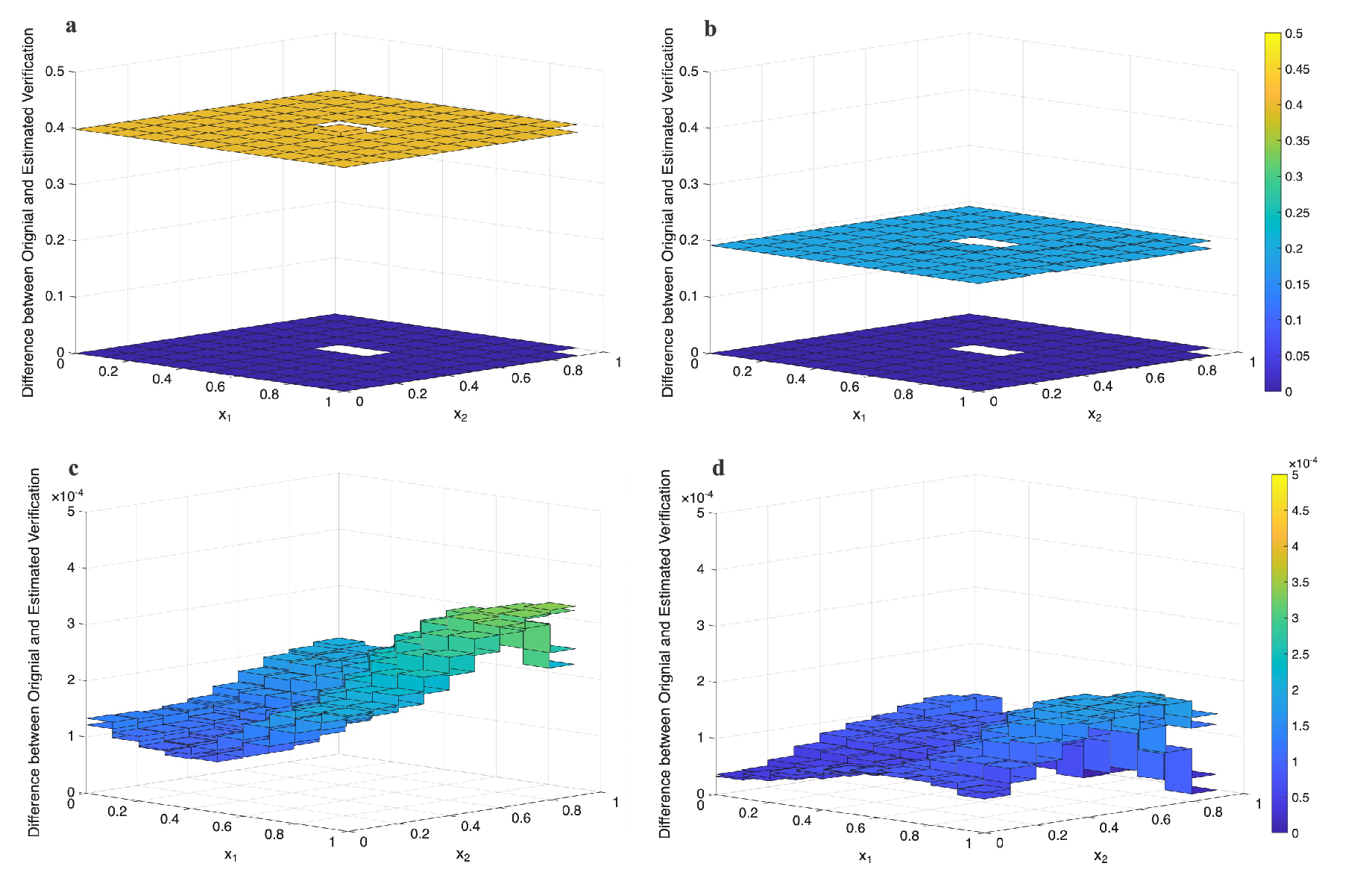}
    \caption{\textcolor{black}{The difference between the original and estimated verification with varying data scale $n=10^4$ and $n=5\times 10^4$ for the empirical method in panels (a) and (b), and 
    for the NPE method in panels (c) and (d).}}
    \label{fig:empirical_NPE_comparison}
\end{figure}

\end{example}

\section{Conclusions}
\label{sec:conclusions}
In this paper, we proposed a data-driven formal verification and synthesis framework for unknown stochastic systems with non-parametric representations. Our approach guarantees probabilistic asymptotic convergence of the verification results by deriving probabilistic asymptotic bounds for the transition probabilities, ensuring that the estimated transition probabilities converge to their true values.
We also demonstrated that by integrating the results with the asymptotic upper bound of the LC, this property can guarantee the asymptotic closeness in formal verification between the original system and its data-driven finite abstraction converging asymptotically to some specified value.
The effectiveness of this data-driven framework has been demonstrated through case studies.
In the future, building on these results, we aim to improve the performance of our approach when additional information on the conditional density function of the system is available, e.g., through known constraints on the moments.

\begin{acks}
This work is supported by the EIC SymAware project 101070802 and the ERC Auto-CyPheR project 101089047.
\end{acks}

\bibliographystyle{ACM-Reference-Format}
\bibliography{sample-base}

\appendix

\section{SUPPLEMENTARY OF SECTION~\ref{sec2}}\label{prel_appen}
\subsection{Supplementary of Section~\ref{sec:prelem}}\label{kernel_bandwidth}

\paragraph{Accuracy of the Kernel Estimator}
The asymptotic bias and variance of $\hat{f}_X(\bs{x})$ in \textcolor{black}{equation}~\eqref{equ:MKDE} are derived by \cite{hardle2004nonparametric} as
\begin{align*}
\Bias[\hat{f}_{X}(\bs{x})] & \approx \frac{1}{2}\mu_2(K) \,\textsf{tr}(H^T\mathcal{H}_f(\bs{x})H) \\
\Var[\hat{f}_{X}(\bs{x})]  &\approx \frac{1}{n\;|H|}\|K\|_2^2f_X(\bs{x}),
\end{align*}
where $\mu_2(K)$ is a constant defined with $\int \bs{u}\bs{u}^TK(\bs{u})\,d\bs{u}=\mu_2(K)\bs{I}_d$,
$\mathcal{H}_f(x)$ is the Hessian matrix of second partial derivatives of $f$, 
$\textsf{tr}(\cdot)$ denotes the trace of a matrix,
and $\|K\|_2$ is the $L_2$-norm of $K$.
Then, the asymptotic mean integrated squared error (AMISE) can be formulated as
\[
\AMISE(\hat{f}_{X} )\!=\!\frac{1}{4}\mu_2^2(K)\!\!\int\! \textsf{tr}(H^T\mathcal{H}_f(\bs{x})H)^2 d\bs{x} +\frac{\|K\|_2^2}{n|H|}.
\]
The AMISE is primarily determined by the choice of the bandwidth, with the kernel function having a minor effect only through specific characteristics such as the order of its first nonzero moments \citep{hardle2004nonparametric,scott2015multivariate}.


\paragraph{Choice of the Kernel Function.}
For two different kernel functions, in practice it is possible to get approximately the same degree of smoothness by multiplying one of the bandwidths with an adjustment factor \citep{hardle2004nonparametric}. This adjustment factor for two kernels $A$ and $B$ can be computed from $h_B=h_A\delta_0^B/\delta_0^A$ where $\delta_0$ is the canonical bandwidth. Table~\ref{tab:kernel} gives commonly used kernels and their canonical bandwidths. For the selection of the kernel, the conservative recommendation is the kernel which is smooth, clearly unimodal, and symmetric 
around the origin. Also, some factors (e.g., ease of computation and differentiability) should be considered, rather than concerning the loss of efficiency.
\begin{table*}[hbt!]
\centering
\caption{Commonly used kernels and their canonical bandwidths}
\begin{tabular}[t]{lcc}
\hline Kernel & Equation & Canonical Bandwidths \\
\hline
Uniform & $\frac{1}{2}I(|u|\leq 1)$ & 1.3510 \\
Triangle & $(1-|u|)I(|u|\leq 1)$ & 1.8890\\
Epanechnikov & $\frac{3}{4}(1-u^2)I(|u|\leq 1)$ & 1.7188 \\
Quartic (Biweight) & $\frac{15}{16}(1-u^2)^2I(|u|\leq 1)$ & 2.0362 \\
Triweight & $\frac{35}{32}(1-u^2)^3I(|u|\leq 1)$ & 2.3122 \\ 

Gaussian & $k(u)=\frac{1}{\sqrt{2\pi}}exp(-u^2/2)$& 0.7764 \\
\hline
\end{tabular}
\label{tab:kernel}
\end{table*}

\paragraph{Related Works on the Selection of Bandwidth.}
There are several data-driven methods for choosing the optimal bandwidth. 
For example,
\cite{hardle2004nonparametric} proposed using Silverman’s rule-of-thumb bandwidth for unimodal distributions that are fairly symmetric and are not heavy-tailed, and using the cross-validation method which is fairly independent of the special structure of the parameter or function estimate.
\cite{tsybakov2009introduction} have employed the cross-validation method to choose the ideal value of the bandwidth and then constructed unbiased risk estimators using the Fourier analysis of density estimators.
\cite{sheather2004density} have provided a practical description of kernel density estimation methods and compared the performance of three methods for selecting the value of the bandwidth, including Rules of Thumb, Cross-Validation, and Plug-in Methods. 
\cite{wand1992error} have demonstrated numerical minimization of the AMISE for general $H$, which can be used as a data-driven method for choosing the optimal bandwidth using a plug-in approach.
Several papers have recommended to construct a family of density estimates based on a number of values of the bandwidth \citep{marron2001presentation,scott2015multivariate}.
In general, the above methods for selection of the bandwidth is with respect to keeping the balance between the bias and the variance to avoid under-smoothing and over-smoothing scenarios.

\vspace{-0.1cm}
\paragraph{Scott's Formula for Selecting the Bandwidth.}
Scott's formula provides a method for selecting the bandwidth of the estimator for a normal distribution with covariance matrix $\Sigma\!=\!diag(\sigma_1^2,\ldots,\sigma_d^2)$, and ensures the optimal convergence rate $O(n^{-\frac{4}{4+d}})$ for the AMISE \citep{hardle2004nonparametric}.
The optimal bandwidth is 
$
H=n^{-1/(d+4)}\hat{\Sigma}^{1/2},
$
with
$\hat\Sigma    :=\frac{1}{n}\sum^{n}_{i=1}(\hat{X}_{i}-\bar{X})^{2}$ and $\bar{X}:=\frac{1}{n}\sum^{n}_{i=1}\hat{X}_{i}$.

\vspace{-0.1cm}

\paragraph{Cross-Validation.}
The cross-validation (CV) method finds the best bandwidth by minimizing an unbiased or approximately unbiased estimator of MISE instead of minimizing MISE. The CV selects the optimal bandwidth $H_{CV}$ by performing the minimization 
\begin{align}
\label{cross-valid}
H_{CV}=&\argmin_{H>0}  CV(H), \text{ with }\\
CV(H) =&
\frac{1}{n^2 |H|}
\sum_{i=1}^{n}\sum_{j=1}^{n}
K\star K\{H^{-1}(\hat{X}_j-\hat{X}_i)\}-\frac{2}{n(n-1)} \sum_{i=1}^{n}\sum_{j=1\,j\ne i}^{n} K\{H^{-1}(\hat{X}_j-\hat{X}_i)\},\nonumber
\end{align}
where $K\!\star\! K(\bs{u})\!:=\!\int\! K(\bs{u}-\bs{v})K(\bs{v})d\bs{v}$.

The asymptotic bias and variance of the univariate conditional density estimator~\eqref{condi_densityestima} are as follows
\begin{align}\label{bias_from1996}
    \Bias[\hat{f}_{Y|X}(y,x)]=&\frac{h_{\mathsf x}^{2}G_{12}^{2}(K)}{2}\left\{ 2\frac{f'_{X}(x)}{f_{X}(x)}\frac{\mathrm{d}}{\mathrm{d}x}f_{Y|X}(y,x)+\frac{\mathrm{d}^{2}}{\mathrm{d}^{2}x}f_{Y|X}(y,x)
    +\frac{h_{\mathsf y}^{2}}{h_{\mathsf x}^{2}}\frac{\mathrm{d}^{2}}{\mathrm{d}y^{2}}f_{Y|X}(y,x)\right\}\nonumber\\
    &+O(h_{\mathsf x}^4)+O(h_{\mathsf y}^4)+O(h_{\mathsf x}^2h_{\mathsf y}^2)+O(\frac{1}{nh_{\mathsf x}}),
\end{align}
and
\begin{equation}\label{var_from1996}
    \Var[\hat{f}_{Y|X}(y,x)]
    =\frac{G_{20}(K)f_{Y|X}(y,x)}{nh_{\mathsf y}h_{\mathsf x}f_{X}(x)}[G_{20}(K)-h_{\mathsf y}f_{Y|X}(y,x)]+O(\frac{1}{n})+O(\frac{h_{\mathsf y}}{nh_{\mathsf x}})+O(\frac{h_{\mathsf x}}{nh_{\mathsf y}}),
\end{equation}
where $G_{12}(K)=\int u^{2}K(u)du$ and $G_{20}(K)=\int K^{2}(u)du$, if $h_{\mathsf x},~h_{\mathsf y}\to 0$ and $n\to +\infty$.

\section{Supplementary of Section~\ref{pro_formula}}\label{supp_pro_formula}

To verify system \eqref{dynamic_evolu} against specification $\psi$, a DTSCS $\Sigma_{ss}$ need to be approximated by a finite MDP $\textcolor{black}{\bar{\Sigma}}_{ss}=(Q, S_{\mathfrak a}, P, AP, L )$ with $Q$ representing a partition of the state space $\mathcal S$ with partition sets denoted by $q\in Q$. The action space $S_{\mathfrak a} = U$. For the transition probabilities $P_{ij,a}$, select representative points $\bar q\in q$ for each partition set. Define $P_{ij,a} :=\textsf{Prob}_w(f(\bar q_i,a,w)\in q_j)$.
Define the state discretizations parameter $\delta:=\sup\{ \| x-x'\|,x,x'\in q,\,\, q\in Q\}$. 
DTSCS $\Sigma_{ss}$ and its finite MDP abstraction $\textcolor{black}{\bar{\Sigma}}_{ss}$ under any strategy $\varpi(\cdot)\in \mathcal{U}_{\mathfrak a}$ are denoted as $\Sigma^{\varpi}_{ss}$ and $\hat{\Sigma}^{\varpi}_{ss}$, respectively.
The following theorem \citep{SA13} provides the closeness guarantee between $\Sigma_{ss}$ and its finite abstraction $\hat{\Sigma}_{ss}$.

\begin{theorem}\label{SA13_bound}
      For a given PCTL specification $\psi$ over a finite horizon and any strategy $\varpi(\cdot)\in \Pi$,
      the closeness between $\Sigma^{\varpi}_{ss}$ and $\textcolor{black}{\bar{\Sigma}}^{\varpi}_{ss}$ can be obtained as 
\begin{align}\label{spec_disc_para}
    |P(\Sigma^{\varpi}_{ss}\models \psi)-P(\textcolor{black}{\bar{\Sigma}}^{\varpi}_{ss}\models \psi) |\leq \epsilon, \text{ with $\epsilon:=\bs T\delta L \mathfrak{L}$,}
\end{align}
where $\bs T$ is the finite time horizon, $\delta$ is the state discretizations parameter, $L$ is the Lipschitz constant of the stochastic kernel, and $\mathfrak{L}$ is the Lebesgue measure of the specification set. 

\end{theorem}

\begin{remark}\label{reason_partition}
The upper bound $L$ of the LC impacts the algorithm as follows: One can initially fix the desired threshold $\epsilon$ in advance, and then select the partition parameter $\delta=\frac{\epsilon}{T L \mathfrak{L}}$ according to the values of $T$, $L$, $\mathfrak{L}$.
This partition provides a guarantee for the verification based on MDP abstraction, which ensures the absolute distance between the satisfaction probability of the original system and that of its finite MDP abstraction is smaller than $\epsilon$.
\end{remark}

\section{SUPPLEMENTARY OF SECTION~\ref{main}}\label{supp_main}

If Assumption~\ref{asstwod_1} is not satisfied, we can have the following weaker result. The values of $\textcolor{black}{\bar{\epsilon}}^{\textcolor{black}{c}}_{\bs x,\bs y}$ and $\epsilon^{\textcolor{black}{c}}_{\bs x,\bs y,d}$ in \textcolor{black}{inequality}~\eqref{mult_mid_bernstain_2_weak} and \eqref{mult_mid_berstain_4_weak} are much larger than  $\epsilon_{\bs x,\bs y}$ and $\epsilon_{\bs x,\bs y,d}$ in \textcolor{black}{inequality}~\eqref{mult_mid_bernstain_2} and \eqref{mult_mid_berstain_4}, respectively. 

\begin{lemma}\label{multi_lemma_midd_weak}
    Suppose $\textcolor{black}{X=(X_{1},\ldots,X_{\mathsf d})}$ and $\textcolor{black}{Y=( Y_{1},\ldots, Y_{\mathsf d})}$ are $\mathsf d$-dimensional random variables, which have a CoDF $f_{Y|X}(\bs y,\bs x)$ for all $(\bs x,\bs y)\in D_{\bs X}\times D_{\bs Y}$, and $X$ is from the uniform distribution with density function $f_{X}(\bs x)$. Suppose that $\textcolor{black}{k}:\mathbb{R}\to \mathbb{R}$ is the Gaussian kernel function with bandwidths $h.$ In addition, samples $\{\textcolor{black}{\hat{X}}_{i},~i=1,\ldots,n \}$ are selected uniformly from $D_{\bs X}$ with density function $f_{X}(\bs x)$, and for each $\textcolor{black}{\hat{X}}_{i}$, sample $\textcolor{black}{\hat{Y}}_{i}$ is generated from $(\textcolor{black}{Y|\hat{X}}_{i})$ with the CoDF $f_{Y|X}(\bs y,\bs x)$, $i\in \{1,\ldots,n \}$, where $\textcolor{black}{\hat{\bs X}_{i}=(\hat{ X}_{i1},\ldots,\hat{X}_{i\,\mathsf d})}$ and $\textcolor{black}{\hat{\bs Y}_{i}=(\hat{ Y}_{i1},\ldots,\hat{Y}_{i\,\mathsf d})}$. Then, for all $\tau_{\bs x},\tau_{\bs x,\bs y}>0$, for any $j\in \{ 1,\ldots,\mathsf d \}$, we have
\begin{align}\label{mult_mid_bernstain_2_weak}
    P\left(\left| \hat{f}(\bs y,\bs x)\Bar{f}(\bs x)-\Bar{f}(\bs y,\bs x)\hat{f}(\bs x) \right|\leq \textcolor{black}{\bar{\epsilon}}^{\textcolor{black}{c}}_{\bs x,\bs y} \right)\ge 1-2\exp(-\tau_{\bs x,\bs y}),
\end{align}
\begin{align}\label{mult_mid_berstain_4_weak}
    P\left( \left|  \frac{\partial}{\partial\textcolor{black}{x}_{j}}\hat{f}(\bs y,\bs x)-\frac{\partial}{\partial\textcolor{black}{x}_{j}}\bar{f}(\bs y,\bs x)  \right| \leq \epsilon^{\textcolor{black}{c}}_{\bs x,\bs y,d} \right)\ge 1-2\exp{(-\tau_{\bs x,\bs y,d})},
\end{align}
    {\color{black}where
    $
    \bar{\delta}^{2,\textcolor{black}{*}}_{\bs x,\bs y}= {\color{black} h^{-\frac{7}{2}\mathsf d} G^{\frac{1}{2}}_{40}f^{\frac{5}{2}}_{X}(\bs x) \max_{u}k^{2\mathsf d}(u) 
    +h^{-4\mathsf d} G^2_{20}(K)f^2_{X}(\bs x) \max_{u}k^{2\mathsf d}(u)
    },
    $
\begin{align*}
    \bar{R}_{\bs x,\bs y}= 2\max_{u}\mathsf d h^{-2\mathsf d}\bs k^2(u)\E_{p}\bs k_{\bs x,h},~\Bar{\delta}^{2,\textcolor{black}{*}}_{\bs x,\bs y,d}=h^{-\frac{3\mathsf d+4}{2}}  G^{\frac{1}{2}}_{22}(K)f^{\frac{1}{2}}_{X}(\bs x)\max_{u}k^{\mathsf d}(u),
\end{align*}    
\begin{align*}
    R_{\bs x,\bs y,d}=h^{-2\mathsf d+1} \max_{u}uK^{2\mathsf d}(u)  +h^{-\frac{2\mathsf d+1}{2}}G^{\frac{1}{2}}_{22}(K)f^{\frac{1}{2}}_{X}(\bs x)h^{-\mathsf d}\max_{u}k^{\mathsf d}(u),
\end{align*}
    \begin{align*}
        \bar{\epsilon}^{\textcolor{black}{c}}_{\bs x,\bs y}=\frac{\tau_{\bs x,\bs y} \Bar{R}_{\bs x,\bs y}}{3n}+\sqrt{ \frac{\tau_{\bs x,\bs y}^{2} \Bar{R}_{\bs x,\bs y}^2}{9n^2}+\frac{2\tau_{\bs x,\bs y}\bar{\delta}_{\bs x,\bs y}^{2,\textcolor{black}{*}}}{n}  },~\epsilon^{\textcolor{black}{c}}_{\bs x,\bs y,d}=\frac{\tau_{\bs x,\bs y,d} R_{\bs x,\bs y,d}}{3n}+\sqrt{ \frac{\tau_{\bs x,\bs y,d}^2 R_{\bs x,\bs y,d}^2}{9n^2}+\frac{2\tau_{\bs x,\bs y,d}\bar{\delta}_{\bs x,\bs y,d}^{2,\textcolor{black}{*}}}{n}  }.
    \end{align*}

}
\end{lemma}
\begin{proof}
The proof is similar with Lemma~\ref{multi_lemma_midd}.
In this part, we only give the proof for \textcolor{black}{inequality}~\eqref{mult_mid_bernstain_2_weak}. We need to give the upper bound of $\delta^2_{\bs x,\bs y}.$ 
    
Since 
\begin{align*}
    \E_{p}k_{\bs x,\bs y,h}=&h^{-2\mathsf d}\E_{p}\left[K(\frac{\bs x-\textcolor{black}{X}}{h})K(\frac{\bs y-\textcolor{black}{Y}}{h})\right]\\
    \leq& h^{-2\mathsf d} \E^{\frac{1}{2}}_{p}\left[ K^{2}(\frac{\bs x-\textcolor{black}{ X}}{h})\right]\E^{\frac{1}{2}}_{p}\left[ K^{2}(\frac{\bs y-\textcolor{black}{Y}}{h})\right]~~~~~~~~~~~~~~\text{(Holder Inequality)}\\
    \leq&h^{-2\mathsf d} \left[h^{\mathsf d} G_{20}(K)f_{X}(\bs x) \right]^{\frac{1}{2}} \max_{u}k^{\mathsf d}(u)
    =h^{-\frac{3}{2}\mathsf d} G^{\frac{1}{2}}_{20}(K)f^{\frac{1}{2}}_{X}(\bs x)\max_{u}k^{\mathsf d}(u),
\end{align*}
\begin{align*}
    \E_{p}k^{2}_{\bs x,\bs y,h}=& h^{-4\mathsf d}\E_{p}\left[ K^2(\frac{\bs x-\textcolor{black}{X}}{h})K^2(\frac{\bs y-\textcolor{black}{Y}}{h}) \right]\\
    \leq& h^{-4\mathsf d}\E^{\frac{1}{2}}_{p}\left[ K^4(\frac{\bs x-\textcolor{black}{X}}{h}) \right] \E^{\frac{1}{2}}_{p}\left[ K^4(\frac{ \bs y-\textcolor{black}{Y}}{h}) \right]\\
    \leq& h^{-4\mathsf d} \left[ h^{\mathsf d}G_{40}(K)f_{X}(x) \right]^{\frac{1}{2}}\max_{u}k^{2\mathsf d}(u)
    =h^{-\frac{7}{2}\mathsf d} G^{\frac{1}{2}}_{40}f^{\frac{1}{2}}_{X}(x) \max_{u}k^{2\mathsf d}(u),
\end{align*}
and 
\begin{align*}
    &\E_{p}\left[ k_{\bs x,\bs y,h}\left(\textcolor{black}{\hat{X}}_{i},\textcolor{black}{\hat{Y}}_{i}\right)k_{\bs x,h}(\textcolor{black}{\hat{X}}_{i}) \right]\\
    =& h^{-3\mathsf d}\E_{p}\left[  K^2\left(\frac{ \bs x-\textcolor{black}{X}}{h}\right) K\left(\frac{ \bs y-\textcolor{black}{X}}{h}\right) \right]\\
    \leq& h^{-3\mathsf d} \E_{p}^{\frac{1}{2}}\left[ K^{4}\left(\frac{ \bs x-\textcolor{black}{X}}{h} \right) \right] \E^{\frac{1}{2}}_{p}\left[ K^{2}\left( \frac{ \bs y-\textcolor{black}{Y}}{h}  \right) \right]\\
    =&h^{-3\mathsf d} \left[  h^{\mathsf d}G_{40}(K)f_{X}(\bs x)  \right]^{\frac{1}{2}}\max_{u}k(u)
    =h^{-\frac{5}{2}\mathsf d}G^{\frac{1}{2}}_{40}(K)f^{\frac{1}{2}}_{X}(\bs x)\max_{u}k(u),
\end{align*}
we can have 
\begin{align*}
    &\delta^2_{\bs x,\bs y}:=\Var_{p}[ g_{\bs x,\bs y}(\textcolor{black}{\hat{ X}}_{i},\textcolor{black}{\hat{ Y}}_{i})]\\
    =&\E_{p}\left[\left[ k_{\bs x,\bs y,h}(\textcolor{black}{\hat{ X}}_{i},\textcolor{black}{\hat{ Y}}_{i})\E_{p}k_{\bs x,h}-k_{\bs x,h}(\textcolor{black}{\hat{ X}}_{i})\E_{p}k_{\bs x,\bs y,h} \right]^2\right] 
    \\
    =&{\color{black}\E_{p} k^2_{\bs x,\bs y,h}\E^2_{p}k_{\bs x,h}+\E_{p}k^2_{\bs x,h}\E^2_{p}k_{\bs x,\bs y,h}-2\E_{p}\left[ k_{\bs x,\bs y,h}(\textcolor{black}{\hat{ X}}_{i},\textcolor{black}{\hat{ Y}}_{i})k_{\bs x,h}(\textcolor{black}{\hat{ X}}_{i}) \right]\E_{p}k_{\bs x,h}\E_{p}k_{\bs x,\bs y,h}}\\
    \leq&{\color{black} h^{-\frac{7}{2}\mathsf d} G^{\frac{1}{2}}_{40}f^{\frac{5}{2}}_{X}(\bs x) \max_{u}k^{2\mathsf d}(u) 
    +h^{-4\mathsf d} G^2_{20}(K)f^2_{X}(\bs x) \max_{u}k^{2\mathsf d}(u)
    }
\end{align*}
Thus, we can obtain 
\begin{align*}
    \bar{\delta}^{2,\textcolor{black}{*}}_{\bs x,\bs y}=&{\color{black} h^{-\frac{7}{2}\mathsf d} G^{\frac{1}{2}}_{40}f^{\frac{5}{2}}_{X}(\bs x) \max_{u}k^{2\mathsf d}(u) 
    +h^{-4\mathsf d} G^2_{20}(K)f^2_{X}(\bs x) \max_{u}k^{2\mathsf d}(u)
    }.
\end{align*}

As $\Bar{R}_{\bs x,\bs y}$ has been obtained in Lemma~\ref{multi_lemma_midd}, we can now obtain the confidence of the upper bound of $\hat{f}(\bs y,\bs x)\Bar{f}(\bs x)-\Bar{f}(\bs y,\bs x)\hat{f}(\bs x)$, as follows.
For any  $n$ and $\tau_{\bs x,\bs y}$, there exists $\textcolor{black}{\bar{\epsilon}^{c}}_{\bs x,\bs y}$ such that 

\begin{align*}
    P\left( \left|  (\hat{f}(\bs y,\bs x)-\Bar{f}(\bs y,\bs x))\Bar{f}(\bs x) +\Bar{f}(\bs y,\bs x)(\Bar{f}(\bs x)-\hat{f}(\bs x)) \right|\leq \textcolor{black}{\bar{\epsilon}^{c}}_{\bs x,\bs y}\right)\ge 1-2\exp(-\tau_{\bs x,\bs y}),
\end{align*}
where 
\begin{align*}
    \textcolor{black}{\bar{\epsilon}^{c}}_{\bs x,\bs y}=&\max\left\{\frac{\tau_{\bs x,\bs y} \Bar{R}_{\bs x,\bs y}}{3n}+\sqrt{ \frac{\tau_{\bs x,\bs y}^{2} \Bar{R}_{\bs x,\bs y}^2}{9n^2}+\frac{2\tau_{\bs x,\bs y}\delta_{\bs x,\bs y}^{2,\textcolor{black}{*}}}{n}  }, \frac{\tau_{\bs x,\bs y} \Bar{R}_{\bs x,\bs y}}{3n}-\sqrt{ \frac{\tau_{\bs x,\bs y}^{2}\Bar{R}_{\bs x,\bs y}^2}{9n^2}+\frac{2\tau_{\bs x,\bs y}\delta_{\bs x,\bs y}^{2,\textcolor{black}{*}}}{n} }  \right\}\\
    =&\frac{\tau_{\bs x,\bs y} \Bar{R}_{\bs x,\bs y}}{3n}+\sqrt{ \frac{\tau_{\bs x,\bs y}^{2} \Bar{R}_{\bs x,\bs y}^2}{9n^2}+\frac{2\tau_{\bs x,\bs y}\delta_{\bs x,\bs y}^{2,\textcolor{black}{*}}}{n}  }.
\end{align*}
Since $\bar{\delta}^2_{\bs x,\bs y}\ge \delta^2_{\bs x,\bs y},$ it still holds that 
\begin{align*}
    P\left( \left|  (\hat{f}(\bs y,\bs x)-\Bar{f}(\bs y,\bs x))\Bar{f}(\bs x) +\Bar{f}(\bs y,\bs x)(\Bar{f}(\bs x)-\hat{f}(\bs x)) \right|\leq \bar{\epsilon}_{x,y}\right)\ge 1-2\exp(-\tau_{\bs x,\bs y}),
\end{align*}
where $\bar{\epsilon}_{\bs x,\bs y}=\frac{\tau_{\bs x,\bs y} \Bar{R}_{\bs x,\bs y}}{3n}+\sqrt{ \frac{\tau_{\bs x,\bs y}^{2} \Bar{R}_{\bs x,\bs y}^2}{9n^2}+\frac{2\tau_{\bs x,\bs y}\bar{\delta}_{\bs x,\bs y}^2}{n}  }.$

\end{proof}



\textcolor{black}{
\begin{proposition}
\label{prop:uniform_transition_probability}
Suppose that the assumptions of
Theorem~\ref{multi_Bound_TranProbi} hold, every abstract state
$q_i$ is compact, and the transition-error processes satisfy the
Lipschitz condition in
Remark~\ref{rem:uniform_transition_probability}.
For any abstract state $q_{i},q_j,~i,j\in \{1.\ldots, n_{Q}\}$, and action $a$, assume that there exists a
deterministic Lipschitz constant $L_{tr}<+\infty$
such that, for all $\bs x,\bs z\in q_i,$
\begin{equation}
\left|
Z_{i j,a}(\bs x)-Z_{i j,a}(\bs z)
\right|
\le
L_{tr}\|\bs x-\bs z\|
\label{eq:transition_error_lipschitz}
\end{equation}
where $Z_{ij,a}(\bs x)=
\hat P_{ij,a}(\bs x)-P_{ij,a}(\bs x).$
For any $\eta_{i j,a}>0$, let $\mathcal N_{\eta_{i j,a}}$ be a finite $\eta_{i j,a}$-cover of $q_i$ with cardinality $N_{\eta_{i j,a}}$, where $\eta_{i j,a} =
\frac{1}
{n^2(1+L_{tr})}$.
Define $\bar\epsilon_{p,j}(\tau):=\max_{\bs x_{r}\in \mathcal N_{\eta_{i j,a}}}\{\epsilon_{p,j}(\bs x_r,\tau)\},$ where 
$\epsilon_{p,j}(\bs x_r,\tau)$ denotes the
pointwise error bound in
Theorem~\ref{multi_Bound_TranProbi} for the transition probability
from $q_i$ to $q_j$, evaluated at
$\bs x_r\in q_{i}$ with probability parameter
$\tau$, $j\in\{1,\ldots,n_{Q}\}$.
Denote $\bar\epsilon_p(\tau)
:=
\max_{1\le j\le n_Q}
\bar\epsilon_{p,j}(\tau),$
and $D_i
:=
\operatorname{diam}(q_{i})
=
\sup_{\bs x,\bs z\in q_{i}}
\|\bs x-\bs z\|.$
Then, for all
$\tau_p>0$, 
\[
\max_{\substack{
1\leq j\leq n_Q
}}
\sup_{\bs x\in q_i}
\left|
\hat P_{i j,a}(\bs x)
-
P_{i j,a}(\bs x)
\right|
\lesssim
\bar\epsilon_p\bigl(\tau_p+O(\log n)\bigr)
\]
with probability at least $1-2\exp(-\tau_p)$.
\end{proposition}
}
\begin{proof}
\textcolor{black}{
Let $\mathcal N_{i j,a}
=
\{
\bs x_1,\ldots,\bs x_{N_{i j,a}}
\}
\subset q_i$
be an $\eta_{i j,a}$-cover of $q_i$. Thus, for every
$\bs x\in q_i$, there exists a point
$\bs x_r\in\mathcal N_{ i j,a}$ such that $\|\bs x-\bs x_r\|
\le
\eta_{i j,a}.$
Since $q_i$ is a compact subset of $\mathbb R^{\mathsf d}$, its
covering number satisfies
\begin{equation}
N_{ij,a}
\le
\left(
1+
\frac{2D_{i}}{\eta_{i  j,a}}
\right)^{\mathsf d}
=
M_{ i  j,a}.
\label{eq:covering_number_bound_direct}
\end{equation} }
\textcolor{black}{For every fixed point $\bs x_r$,
Theorem~\ref{multi_Bound_TranProbi} gives 
$P\left(
\text{E}_r
\right)
\ge
1-2\exp(-\tau),$ where $\text{E}_r=\left\{\left|
Z_{i j,a}(\bs x_r)
\right|
\lesssim
\epsilon_{p}(\tau)\right\}.$
Define $\bar\epsilon_{p,j}(\tau)=\max\{\epsilon_{p,j}(\bs x_1,\tau),\ldots, \epsilon_{p,j}(\bs x_{N_{i j,a}},\tau)\}.$
It holds $P\left(
\bar{\text{E}}_r
\right)
\ge
1-2\exp(-\tau)$ for every $\tau\ge 0,$ where $\bar{\text{E}}_r=\left\{\left|
Z_{i j,a}(\bs x_r)
\right|
\lesssim
\bar{\epsilon}_{p,j}(\tau)\right\}.$
}

\textcolor{black}{
Applying the union bound over all points in $\mathcal N_{i j,a}$ yields
\begin{align*}
&
P\left(
\bigcup^{N_{ij,a}}_{r=1} \bar{\text{E}}^{c}_r
\right)
\le
\sum_{r=1}^{N_{i j,a}}
P\left( \bar{\text{E}}^{c}_r
\right)
\le
2N_{i j,a}\exp(-\tau)
\le
2M_{i j,a}\exp(-\tau).
\label{eq:union_bound_over_net_direct}
\end{align*}
Then, we have $P\left(
\bigcap_{r=1}^{N_{i j,a}}
\bar E_r
\right)
\ge
1-2M_{i j,a}e^{-\tau}.$
On the event
$\bigcap_{r=1}^{N_{i j,a}}\bar E_r$, it follows that
$\max_{1\le r\le N_{i j,a}}
|Z_{i j,a}(\mathbf x_r)|
\lesssim
\bar\epsilon_{p,j}(\tau).$
Therefore,
\[
P\left(
\max_{1\le r\le N_{i j,a}}
|Z_{i j,a}(\mathbf x_r)|
\lesssim
\bar\epsilon_{p,j}(\tau)
\right)
\ge
P\left(
\bigcap_{r=1}^{N_{i j,a}}
\bar E_r
\right)
\ge
1-
2M_{ij,a}\exp{(-\tau)}.
\]
}

\textcolor{black}{
Now choose $\tau
=
\Lambda_{i j,a}
=
\tau_p+\log M_{ij,a},$ for all $\tau_p> 0.$
Then
\begin{align*}
2M_{i j,a}\exp(-\Lambda_{ ij,a})
&=
2M_{i  j,a}
\exp\left(
-\tau_p-\log M_{i  j,a}
\right)
=
2\exp(-\tau_p).
\end{align*}
Therefore, with probability at least
$1-2\exp(-\tau_p)$, 
\begin{equation*}
\max_{1\le r\le N_{i j,a}}
\left|
Z_{i  j,a}(\bs x_r)
\right|
\lesssim
\bar\epsilon_{p,j}
\left(
\Lambda_{i  j,a}
\right).
\label{eq:simultaneous_net_bound_direct}
\end{equation*}
On this event, take any $\bs x\in q_i$. 
We have 
\begin{align*}
\left|
Z_{i j,a}(\bs x)
\right|
&\le
\left|
Z_{i  j,a}(\bs x_r)
\right|
+
\left|
Z_{i j,a}(\bs x)
-
Z_{i  j,a}(\bs x_r)
\right|
\lesssim
\bar\epsilon_{p,j}
\left(
\Lambda_{i  j,a}
\right)
+
L_{tr}\eta_{i j,a}.
\label{eq:extension_from_net_direct}
\end{align*}
Choosing $\eta_{i j,a} =
\frac{1}
{n^2(1+L_{tr})}$, we have $L_{tr}\eta_{i  j,a}
=
\frac{
L_{tr}
}{
n^2(1+L_{tr})
}
\le
\frac{1}{n^2}.$
Thus, according to Theorem~\ref{multi_Bound_TranProbi}, we can have that for any $\tau_{p},$
\begin{equation*}
\sup_{\bs x\in q_{\iota}}\left|
Z_{i j,a}(\bs x)
\right|
\lesssim
\bar\epsilon_{p,j}
\left(
\tau_p+O(\log n)
\right), \text{ with probability at least
$1-2\exp(-\tau_p)$.}
\end{equation*}
}

\textcolor{black}{
We can extend the preceding result simultaneously over all target
abstract states. 
Define
\[
\mathcal E_j
:=
\left\{
\sup_{\bs x\in q_{i}}
\left|
Z_{i j,a}(\bs x)
\right|
\lesssim
\bar\epsilon_{p,j}
\left(
\tau_p+\log n_Q+O(\log n)
\right)
\right\},
\qquad
j=1,\ldots,n_Q,
\]
where $\bar\epsilon_{p,j}$ denotes the uniform error radius
corresponding to the transition from $q_i$ to $q_j$.
\\
Applying the preceding result with confidence parameter
$\tau_p+\log n_Q$ gives, for every $j$,
\[
P(\mathcal E_j^c)
\le
2\exp\left(
-\tau_p-\log n_Q
\right)
=
\frac{2\exp(-\tau_p)}{n_Q}.
\]
Therefore, by the union bound,
\begin{align*}
P\left(
\bigcup_{j=1}^{n_Q}\mathcal E_j^c
\right)
&\le
\sum_{j=1}^{n_Q}
P(\mathcal E_j^c)
\le
\sum_{j=1}^{n_Q}
\frac{2\exp(-\tau_p)}{n_Q}
=
2\exp(-\tau_p).
\end{align*}
Equivalently, $P\left(
\bigcap_{j=1}^{n_Q}\mathcal E_j
\right)
\ge
1-2\exp(-\tau_p).$
On the event
$\bigcap_{j=1}^{n_Q}\mathcal E_j$, the uniform error bound holds
simultaneously for every $j$.
Hence, $\max_{1\le j\le n_Q}
\sup_{\bs x\in q_i}
\left|
Z_{ij,a}(\bs x)
\right|
\lesssim
\max_{1\le j\le n_Q}
\bar\epsilon_{p,j}
\left(
\tau_p+\log n_Q+O(\log n)
\right).$
Since the partition is fixed, the term
$\log n_Q$ can be absorbed into the $O(\log n)$ confidence
correction. Defining $\bar\epsilon_p(\tau)
:=
\max_{1\le j\le n_Q}
\bar\epsilon_{p,j}(\tau),$
we obtain
\[
\max_{1\le j\le n_Q}
\sup_{\bs x\in q_i}
\left|
\hat P_{i j,a}(\bs x)
-
P_{i j,a}(\bs x)
\right|
\lesssim
\bar\epsilon_p
\left(
\tau_p+O(\log n)
\right)
\]
with probability at least
$1-2\exp(-\tau_p)$.
}
\end{proof}
\textcolor{black}{
\begin{remark}
    Since the numbers of abstract states and actions are finite,
the above result extends immediately to all source states and
actions by an additional union-bound argument.
The resulting confidence correction is $O(1)$ and is therefore
absorbed into the existing $\tau_p+O(\log n)$ term.
\end{remark}
}

\textcolor{black}{
\begin{proposition}
\label{prop:uniform_y_transition}
Fix $\bs x\in D_{\bs X}$ and let
$A\subseteq D_{\bs Y}$ be compact.
Assume that $$R_n(\bs x,\bs y)= \left|(\hat{f}_{YX}(\bs y,\bs x)-\Bar{f}(\bs y,\bs x))\Bar{f}(\bs x) +\Bar{f}(\bs y,\bs x)(\Bar{f}(\bs x)-\hat{f}_{X}(\bs x))\right|$$ is Lipschitz continuous with respect to
$\bs y$ on $A$, with a deterministic Lipschitz constant
$L_Y$ independent of $n$.
Let
$\mathcal N_\eta=\{\bs y_1,\ldots,\bs y_{N_\eta}\}$
be a finite $\eta$-cover of $A$, where $\eta=\frac{1}{n^2(1+L_Y)}.$
Define
$\bar{\epsilon}^*_{\bs x,\bs y}(\bs x;\tau_{\eta})
:=
\max_{\bs y_k\in\mathcal N_\eta}
\bar{\epsilon}_{\bs x,\bs y}(\bs x,\bs y_k;\tau_\eta),$
where $\tau_\eta=\tau_p+\log N_\eta$ for all $\tau_{p}>0,$ and $\bar{\epsilon}_{\bs x,\bs y}$ denotes the pointwise error bound in \eqref{multi_confi_MiddSecond} evaluated at $\bs x\in D_{X}$ and $\bs y\in A.$
Then,
$\int_A R_n(\bs x,\bs y)d\bs y
\lesssim
\int_A \bar{\epsilon}^*_{\bs x,\bs y}\left(\bs x;\tau_p+O(\log n)\right)d\bs y,$
 with probability at least $1-2\exp(-\tau_p).$
\end{proposition}
The proof follows the same finite-cover and union-bound argument as that of Proposition~\ref{prop:uniform_transition_probability}, and is therefore omitted.
}

\end{document}